\documentclass[11pt]{article}

\usepackage{amsmath,amsfonts,amsthm,amssymb}
\usepackage{graphicx,psfrag,epsf}
\usepackage{enumerate}
\usepackage[numbers]{natbib}
\usepackage{color}
\usepackage{bm}
\usepackage[pdftex,bookmarks=true,breaklinks,hypertexnames=false]{hyperref} 
\usepackage{tikz}
\usepackage{algorithm}
\usepackage[noend]{algpseudocode}
\usepackage{multirow}
\makeatletter
\def\algbackskip{\hskip-\ALG@thistlm}
\makeatother

\graphicspath{ {./pictures/} }

\newtheorem{theorem}{Theorem}
\newtheorem{lemma}{Lemma}
\newtheorem{remark}{Remark}
\newtheorem{definition}{Definition}

\newcommand{\eps}{\varepsilon}
\newcommand{\R}{\mathbb{R}}
\newcommand{\tr}{\text{Tr}}
\newcommand{\mc}{\text{mc}}
\newcommand{\KL}{\text{KL}}

\title{Variational Bayes for high-dimensional linear regression with sparse priors}

\author{Kolyan Ray\footnote{Department of Mathematics, Imperial College London. E-mail: \href{mailto:kolyan.ray@imperial.ac.uk}{kolyan.ray@imperial.ac.uk}}
~and Botond Szab\'o\footnote{Department of Mathematics, Vrije Universiteit Amsterdam. E-mail: \href{mailto:b.t.szabo@vu.nl}{b.t.szabo@vu.nl}\newline \indent {Botond Szab\'o received funding from the Netherlands Organization for Scientific Research (NWO) under Project number: 639.031.654.}}\\
\\
\textit{Imperial College London and Vrije Universiteit Amsterdam}}

\date{}

\begin{document}

\maketitle

\vspace{-0.75cm}

\begin{abstract}
We study a mean-field spike and slab variational Bayes (VB) approximation to Bayesian model selection priors in sparse high-dimensional linear regression. Under compatibility conditions on the design matrix, oracle inequalities are derived for the mean-field VB approximation, implying that it converges to the sparse truth at the optimal rate and gives optimal prediction of the response vector. The empirical performance of our algorithm is studied, showing that it works comparably well as other state-of-the-art Bayesian variable selection methods. We also numerically demonstrate that the widely used coordinate-ascent variational inference (CAVI) algorithm can be highly sensitive to the parameter updating order, leading to potentially poor performance. To mitigate this, we propose a novel prioritized updating scheme that uses a data-driven updating order and performs better in simulations. The variational algorithm is implemented in the R package \texttt{sparsevb}.\\

\end{abstract}

\section{Introduction}
\label{sec:intro}

Inference under sparsity constraints has found many applications in statistics and machine learning \cite{mitchell:1988,tipping:2001}. Perhaps the most widely applied such model is sparse linear regression, where we observe
\begin{equation}\label{eq:model}
Y = X \theta +  Z,
\end{equation}
where $Y \in \mathbb{R}^n$, $X$ is a given, deterministic $n \times p$ design matrix, $\theta \in \mathbb{R}^p$ is the parameter of interest and $Z \sim N_n(0,I_n)$ is additive Gaussian noise. We are interested in the \textit{sparse high-dimensional} setting, where $n\leq p$ and typically $n\ll p$, and many of the coefficients $\theta_i$ are (close to) zero.

From a Bayesian perspective, perhaps the most natural way to impose sparsity is through a \textit{model selection} prior, which assigns probabilistic weights to each potential model, i.e. each subset of $\{1,\dots,p\}$ corresponding to selecting the non-zero coordinates of $\theta\in\mathbb{R}^p$. This is one of the most widely used approaches within the Bayesian community \cite{efron:2008,george:1993,mitchell:1988,west:2003} and includes the popular spike-and-slab prior, which is often considered the gold standard in sparse Bayesian linear regression. Such priors have been shown to perform well for estimation and prediction \cite{johnstone:2004,castillo:2012,castillo:2015,chae:2019}, uncertainty quantification \cite{ray:2017,castillo:2018} and multiple hypothesis testing \cite{castillo:roquain:2018}, see \cite{bcg:review:2020} for a recent review.

However, while these priors perform excellently both empirically and theoretically, the discrete model selection component of the prior can make computation hugely challenging. For $\theta \in \mathbb{R}^p$, inference using the spike-and-slab prior generally involves a combinatorial search over all $2^p$ possible models, a hugely expensive task for even moderate $p$. Fast algorithms for exact posterior computation are thus usually restricted to the diagonal design case \cite{castillo:2012,erven:szabo:2018}, while Markov chain Monte Carlo methods are known to have problems mixing for typical problem sizes of interest \cite{griffin2017}.
 
A popular scalable alternative is variational Bayes (VB), which recasts posterior approximation as an optimization problem. One minimizes the VB objective function, consisting of the Kullback-Leibler (KL) divergence between  a family of tractable distributions, called the variational family, and the posterior. Though the resulting approximation does not provide exact Bayesian inference, picking a computationally convenient variational class can dramatically increase scalability, see for example \cite{blei:2003,hoffman:2013}. An especially popular variational family consists of distributions under which the model parameters are independent, so called \textit{mean-field variational Bayes}. For a nice recent review of VB, see \cite{blei:2016}.

In this work, we consider a mean field family consisting of distributions independently assigning each coordinate of $\theta$ an independent mixture of a Gaussian and Dirac mass at zero, thereby mirroring the form of the spike-and-slab prior (but crucially not the form of the posterior). Such a computational relaxation is significant, reducing the posterior dimension to a much more tractable $O(p)$. This is a natural approximation since it keeps the discrete model selection aspect and many of the interpretable features of the original posterior, for example access to posterior probabilities of submodels and inclusion probabilities of particular covariates. This sparse variational family has been applied in practice \cite{logsdon2010variational,titsias:2011,carbonetto:2012,huang:2016,ormerod:2017}, but comes with few theoretical guarantees.

We study this VB procedure under the frequentist assumption that the data $Y$ has been generated according to a given sparse parameter $\theta_0$. Under standard conditions on the design matrix, we obtain refined oracle type contraction rates for the mean-field VB approximation of model selection priors. As a consequence, these imply that the VB posterior performs optimally regarding both estimation of a sparse $\theta$ and for prediction of the response vector. This provides a theoretical justification for this attractive approximation algorithm in a sparsity context.

While similar VB approaches have been applied in the methodological literature \cite{logsdon2010variational,titsias:2011,carbonetto:2012,huang:2016,ormerod:2017}, our contribution also possesses a crucial methodological difference. These existing works typically use Gaussian slabs for the prior, which allows analytic evaluation of certain formulas in the variational algorithm leading to fast optimization. However, Gaussian slabs are inappropriate for recovering the true signal $\theta_0$ since the \textit{true} underlying posterior performs excessive shrinkage causing poor performance \cite{castillo:2012}. One cannot typically expect a VB approximation based on a poorly performing underlying posterior to perform well for recovery. We instead consider Laplace slabs for the prior, which result in optimal recovery when using the true posterior \cite{castillo:2012,castillo:2015}. We are thus using a similar variational family to estimate a \textit{different posterior distribution} compared to previous works. Another way to correct the original posterior is to explicitly shift the posterior mean using an empirical Bayes approach \cite{martin:2017,martin:2019,belitser2020,belitser2020a}.

We provide the methodological details for applying the widely-used coordinate-ascent variational inference (CAVI) algorithm \cite{blei:2016} with Laplace slabs and investigate our method numerically on both simulated and real world ozone interaction data. As predicted by the theory, our method performs well in a number of settings and typically outperforms VB approaches with prior Gaussian slabs. In fact, we find that our approach generally performs at least as well as other state-of-the-art Bayesian variable selection methods. We have implemented our algorithm in the R-package \texttt{sparsevb} \cite{sparsevb}.

Our simulations also show that the CAVI algorithm is highly sensitive to the updating order of the parameters. Since the VB objective function is non-convex and typically has multiple local minima, a poorly chosen updating order can trap the algorithm near a highly suboptimal local minimum causing poor performance. To resolve this, we propose a novel \textit{prioritized} update scheme where we base the CAVI parameter update order on the estimated size of the coefficients via a preliminary estimator. Our simulations indicate that such a data-driven updating order performs better than using either a naive or random update order and provides more robustness against being trapped at a suboptimal local minimum. This idea is applicable beyond the present setting and may be useful for other CAVI approaches.

\textbf{Related work.} Whilst VB has found increasing usage in practice, its theoretical understanding is still in the early stages. In low dimensional settings, some Bernstein-von Mises type results have been derived \cite{lu:2016,wang:2017}, while in high-dimensional and nonparametric settings, first results have only recently appeared \cite{zhang:2017,zhang:gao:2017,pati:2017}. There has also been theoretical work on studying variational approximations to \textit{fractional posteriors}, which down-weight the likelihood \cite{alquier:2017,yang:2017,yang:martin:2020}. The papers \cite{zhang:gao:2017,pati:2017,yang:2017} provide general proof methods which employ the classical prior mass and testing approach of Bayesian nonparametrics \cite{ghosal:2000}. However, since it is known that posterior convergence rates, let alone oracle rates as we derive here, for model selection priors cannot easily be established using this approach \cite{castillo:2012,castillo:2015}, their results do not apply to our setting. We have extended some of the present results to high-dimensional logistic regression in follow up work \cite{ray:szabo:clara:2020}.

\textbf{Organization.} In Section \ref{sec:prior+design} we give details of the prior, variational approximation and conditions on the design matrix. We present our main results in Section \ref{sec:main_results}, details of the VB algorithm in Section \ref{sec: VBalgorithm}, numerical results in Section \ref{sec:numerical} and conclusions in Section \ref{sec:conclusion}. In the supplementary material, we give additional numerical results in Section A, full oracle results and proofs in Section B, additional methodological details in Section C and further discussion of the design matrix assumptions in Section D.

\textbf{Notation.} Let $P_\theta$ be the probability distribution of the observation $Y$ arising in model \eqref{eq:model} and let $E_\theta$ denote the corresponding expectation. For two probability distributions $P,Q$, $\KL(P\|Q) = \int \log \tfrac{dP}{dQ} dP$ denotes the Kullback-Leibler divergence. For $x\in \mathbb{R}^d$, we write $\|x\|_2 = (\sum_{i=1}^d |x_i|^2)^{1/2}$ for the Euclidean norm. For a vector $\theta \in \mathbb{R}^p$ and a subset $S \subseteq \{1,\dots,p\}$ of indices, set $\theta_S$ to be the vector $(\theta_i)_{i\in S}$ in $\R^{|S|}$, where $|S|$ denotes the cardinality of $S$. Further let $S_\theta = \{ i: \theta_i \neq 0\}$ be the set of non-zero coefficients of $\theta$. We will often write $S_0 = S_{\theta_0}$ and $s_0 = |S_{\theta_0}|$, where $\theta_0$ is the true vector. For $X_{\cdot i}$ the $i^{th}$ column of $X$, set
\begin{align}\label{X}
\|X\| := \max_{1\leq i \leq p} \|X_{\cdot i}\|_2 = \max_{1\leq i \leq p} (X^TX)_{ii}^{1/2}.
\end{align}

\section{Prior, variational families and design matrix}\label{sec:prior+design}

\subsection{Model selection priors}

We first present the desirable, but computationally challenging, model selection priors that underlie our VB approximation. Consider a prior for $\theta\in \mathbb{R}^p$ that first selects a \textit{dimension} $s$ from a prior $\pi_p$ on $\{0,\dots,p\}$, then uniformly selects a random subset $S \subset \{1,\dots,p\}$ of cardinality $|S|=s$ and lastly a set of non-zero values $\theta_S = \{ \theta_i :i\in S\}$ from a prior density $g_S$ on $\R^{|S|}$. Since it is known that the `slab' distribution should have exponential tails or heavier to achieve good recovery \cite{castillo:2012}, we restrict to the case where $g_S = \prod_{i\in S} \text{Lap}(\lambda)$ is a product of centered Laplace densities with parameter $\lambda>0$ on $\mathbb{R}^s$. This yields the hierarchical prior:
\begin{equation}\label{eq:prior}
\begin{split}
s \sim \pi_p(s) \\
S | |S|=s \sim \text{Unif}_{p,s}\\
\theta_i \stackrel{ind}{\sim} \begin{cases} 
      \text{Lap}(\lambda), & i\in S, \\
      \delta_0, & i \not\in S,
   \end{cases}
\end{split}
\end{equation}
where $\text{Unif}_{p,s}$ selects $S$ from the $p\choose s$ possible subsets of $\{1,\dots,p\}$ of size $s$ with equal probability and $\delta_0$ denotes the Dirac mass at zero. Since we wish the prior to perform model selection via the prior $\pi_p$ on the dimension $s$ rather than via shrinkage of the Laplace distribution, the choice of prior $\pi_p$ is crucial. The aim is to select a distribution which sufficiently downweights large models while simultaneously placing enough mass to the true model. Following \cite{castillo:2015}, we select an exponentially decreasing prior: we assume that there are constants $A_1,A_2,A_3,A_4>0$ with
\begin{align}\label{eq:prior_cond}
A_1 p^{-A_3} \pi_p (s-1) \leq \pi_p (s) \leq A_2 p^{-A_4} \pi_p (s-1), \quad \quad s= 1,\dots,p.
\end{align}
Assumption \eqref{eq:prior_cond} is satisfied by a variety of piors, including those of the form $\pi_p(s) \propto a^{-s} p^{-bs}$ for constants $a,b>0$ (`complexity priors' \cite{castillo:2012}) and binomial priors. The spike-and-slab prior, where we model $\theta_i \stackrel{iid}{\sim} r\text{Lap}(\lambda) + (1-r)\delta_0$, falls within this framework by taking $\pi_p$ to be $\text{Bin}(p,r)$. The value $r$ is the prior inclusion probability of the coordinate $i$ and controls the model selection. Taking a hyperprior $r\sim \text{Beta}(1,p^u)$ for $u>1$ also satisfies \eqref{eq:prior_cond} (\cite{castillo:2012}, Example 2.2), allows mixing over the sparsity level $r$ and gives a prior that does not depend on unknown hyper-parameters.

The regularization parameter $\lambda$ in the slab distribution in \eqref{eq:prior} is allowed to vary with $p$ within the range
\begin{align}\label{prior_lambda}
\frac{\|X\|}{p} \leq \lambda \leq 2\bar{\lambda}, \qquad \qquad \bar{\lambda} = 2\|X\| \sqrt{\log p},
\end{align}
where the norm $\|X\|$ is the maximal column norm defined in \eqref{X}. The quantity $\bar{\lambda}$ is the usual value of the regularization parameter of the LASSO (\cite{buhlmann:2011}, Chapter 6). Large values of $\lambda$ may shrink many coordinates $\theta_i$ in the slab towards zero, which is undesirable in our Bayesian setup since we wish to induce sparsity via $\pi_p$ instead. Indeed, since the slab component identifies the non-zero coordinates, it is unnatural to further shrink these values. It is natural to take fixed values of $\lambda$ or $\lambda \to 0$, both of which are typically allowed by \eqref{prior_lambda} depending on the specific design matrix and regression setting. Specific values of $\|X\|$ for some examples of design matrices are given in Section D in the supplement.

The theoretical frequentist behaviour of the full posterior arising from prior \eqref{eq:prior} has been studied in \cite{castillo:2012,castillo:2015}, who obtain oracle contraction rates amongst other things. We build on their work to show that these results extend to the scalable variational approximation. 

We briefly comment on the more realistic situation that the model has unknown variance $\varsigma^2$, in which case we instead observe $Y=X\theta+\varsigma Z$. Since then
\begin{equation}\label{rescaled_var}
Y/\varsigma = (X/\varsigma)\theta + Z,
\end{equation}
one may first rescale the data using an estimate $\hat{\varsigma}$ of $\varsigma$ and as before endow $\theta$ with the prior \eqref{eq:prior}, thereby obtaining an empirical Bayes approach. We investigate this empirical Bayes approach numerically in Section \ref{sec:unknown_var}, showing that our method continues to perform well in the more realistic scenario of unknown noise level. One can alternatively use a hierarchical Bayesian approach by endowing $\varsigma$ with a hyper-prior, common choices including the inverse Gamma distribution, $c/\varsigma^2$ or the improper prior $1/\varsigma$.

\subsection{Variational approximations}

The posterior $\Pi(\cdot |Y)$ arising from the prior \eqref{eq:prior} and data \eqref{eq:model} assigns weights to all the $2^p$ possible models, except for very special instances of the design matrix $X$ and prior. Since the posterior is difficult to compute for even moderate $p$, we take a VB approximation using the mean-field variational family
\begin{align}\label{def: variational:family}
\mathcal{P}_{MF}= \left\{P_{\mu,\sigma,\gamma} =  \bigotimes_{i=1}^p \left[ \gamma_i N(\mu_i,\sigma_i^2)+  (1-\gamma_i)\delta_0\right] : \mu_i \in\mathbb{R}, \, \, \sigma_i \in \mathbb{R}^+, \,\, \gamma_i \in [0,1] \right\},
\end{align}
with corresponding VB posterior
\begin{align}\label{def: var}
\widetilde{\Pi} = \underset{P_{\mu,\sigma,\gamma} \in \mathcal{P}_{MF}}{\text{argmin}} \KL(P_{\mu,\sigma,\gamma} || \Pi (\cdot | Y)),
\end{align}
the minimizer of the Kullback-Leibler (KL) divergence with respect to the posterior. Under $P_{\mu,\sigma,\gamma}$, we have $\theta_i \sim \gamma_i N(\mu_i,\sigma_i^2)+  (1-\gamma_i)\delta_0$ independent. We thus approximate the posterior with a spike-and-slab distribution with Gaussian slabs under which every coordinate is independent. Note that while the prior may take the form \eqref{def: variational:family}, the posterior will in general not. The key reduction here is that we replace the $2^p$ model weights with the $p$ VB inclusion probabilities $(\gamma_i)$, thereby dramatically shrinking the posterior dimension. The VB approximation \eqref{def: var} forces (substantial) additional independence into the resulting distribution, breaking dependencies between the variables. For instance, pairwise information that two coefficients $\theta_i$ and $\theta_j$ are likely to be selected simultaneously or not at all is lost.

While we use Gaussian slabs in our variational family, it is crucial the \textit{true prior} has slab distributions with at least exponential tails (e.g. Laplace) \cite{castillo:2012}. The reason a Gaussian approximation works well here is that the likelihood induces Gaussian tails in the posterior. We emphasize that we use the same variational family to estimate a different posterior compared to previous works \cite{logsdon2010variational,titsias:2011,carbonetto:2012,huang:2016,ormerod:2017}, which use Gaussian \textit{prior} slabs. While using Gaussian prior slabs is particularly efficient computationally, it can yield poor performance due to excessive shrinkage of the estimated coefficients, as we demonstrate numerically in Section A.2 in the supplement. Computing the VB estimate \eqref{def: var} is an optimization problem that can be tackled using coordinate-ascent variational inference (CAVI), see Section \ref{sec: VBalgorithm} for details.

While the family $\mathcal{P}_{MF}$ is our main object of interest, our proofs yield similar theoretical results for two other closely related variational families. Consider the family of distributions consisting of products of a single multivariate normal distribution with a Dirac measure:
\begin{equation}
\begin{split}
\mathcal{Q}=\{N_S(\mu_S,\Sigma_S)\otimes \delta_{S^c}:&\, S\subseteq \{1,2,...,p\}, \mu_S\in \mathbb{R}^{|S|},\, \\
&\Sigma_S\in\mathbb{R}^{|S|\times|S|} \text{ a positive definite covariance matrix} \},\label{def: variational:family2}
\end{split}
\end{equation}
where $\delta_{S^c}$ denotes the Dirac measure on the coordinates $S^c$. This family is more rigid on the model selection level than $\mathcal{P}_{MF}$, selecting a distribution with a single fixed support set $S$. On this set, however, the family permits a richer representation for the non-zero coefficients, allowing non-zero correlations. Next consider the mean field subclass of $\mathcal{Q}$:
\begin{equation}
\begin{split}
\mathcal{Q}_{MF} =\{N_S(\mu_S,D_S)\otimes \delta_{S^c}:&\, S\subseteq \{1,2,...,p\}, \mu_S\in \mathbb{R}^{|S|},\, \\
&D_S\in\mathbb{R}^{|S|\times|S|}  \text{ a positive definite diagonal matrix} \}.\label{def: variational:family3}
\end{split}
\end{equation}
This family again allows distributions with only a single fixed support set $S$, but further forces independence of the non-zero coefficients. This class is contained in $\mathcal{P}_{MF}$ by considering distributions $P_{\mu,\sigma,\gamma}$ with inclusion probabilities restricted to $\gamma_i \in \{0,1\}$. We define the corresponding VB posteriors by
\begin{align}
\hat{Q} = \underset{Q \in \mathcal{Q}}{\text{argmin}} ~ \KL(Q || \Pi (\cdot | Y)), \qquad \qquad  \widetilde{Q}=\underset{Q \in \mathcal{Q}_{MF}}{\text{argmin}} ~\KL(Q\| \Pi(\cdot|Y)). \label{def: VB2}
\end{align}
While all our theoretical results also apply to the VB posteriors $\hat{Q}$ and $\widetilde{Q}$, these seem to perform worse in practice than $\widetilde{\Pi}$, see Section A.2 in the supplement. This is potentially due to the discrete constraint $\gamma_i \in \{0,1\}$ for these two families, which renders the highly non-convex optimization problems \eqref{def: VB2} difficult to solve.

\subsection{Design matrix}\label{sec:design_matrix}

The parameter $\theta$ in model \eqref{eq:model} is not estimable without further conditions on the regression matrix $X$. For the high-dimensional case $p>n$, which is of most interest to us, $\theta$ is not even identifiable without additional assumptions. We thus assume that there is some ``true'' sparse $\theta_0$ generating the observation \eqref{eq:model} with at most $s_n$ non-zero coefficients:
\begin{align*}
\theta_0\in \{\theta:\, \#(j:\, \theta_j\neq 0)\leq s_n\},\quad \text{for some }s_n = o(n).
\end{align*}
In the sparse setting, it suffices for estimation to have `local invertibility' of the Gram matrix $X^T X$. The notion of invertibility can be made more precise using the following definitions, which are based on the sparse high-dimensional literature (e.g. \cite{buhlmann:2011}), and have been adapted to the Bayesian setting in \cite{castillo:2015}. We provide only a brief description, referring the interested reader to Section 2.2 of \cite{castillo:2015} for further discussion.
\begin{definition}[Compatibility]\label{def:compat}
A model $S\subseteq \{1,\dots,p\}$ has \emph{compatibility number}
$$\phi(S) = \inf \left\{ \frac{\|X\theta\|_2 |S|^{1/2}}{\|X\|\|\theta_S\|_1} : \|\theta_{S^c}\|_1 \leq 7 \|\theta_S\|_1, \theta_S \neq 0 \right\}.$$
\end{definition}
A model is considered `compatible' if $\phi(S)>0$, in which case $\|X\theta\|_2|S|^{1/2}\geq \phi(S) \|X\| \||\theta_S\|_1$ for all $\theta$ in the above set. The number 7 is not important and is taken in Definition 2.1 of \cite{castillo:2015} to provide a specific numerical value; since we use several results from \cite{castillo:2015}, we employ the same convention. The compatibility number does not directly require sparsity, but reduces the problem to approximate sparsity by considering only vectors $\theta$ whose coordinates are small outside $S$. Conversely, the following two definitions deal only with sparse vectors.

\begin{definition}[Uniform compatibility for sparse vectors]\label{def:unif_compat}
The compatibility number for vectors of dimension $s$ is
$$\overline{\phi}(s) = \inf \left\{ \frac{\|X\theta\|_2 |S_\theta|^{1/2}}{\|X\| \|\theta\|_1}: 0 \neq |S_\theta| \leq s \right\}.$$
\end{definition}

\begin{definition}[Smallest scaled sparse singular value]\label{def:ssssv}
The \textit{smallest scaled sparse singular value} of dimension $s$ is
\begin{align*}%\label{def:compat}
\widetilde{\phi}(s) := \inf \left\{ \frac{\|X\theta\|_2}{\|X\|\|\theta\|_2} : 0 \neq |S_\theta| \leq s \right\}.
\end{align*}
\end{definition}
We shall require that these numbers are bounded away from zero for $s$ a multiple of the true model size. If $\|X\| =1$, then $\widetilde{\phi}(s)$ is simply the smallest scaled singular value of a submatrix of $X$ of dimension $s$. Note that Definitions \ref{def:compat}-\ref{def:ssssv} are Definitions 2.1-2.3 of \cite{castillo:2015}. Such compatibility conditions are standard for sparse recovery problems, see Sections 6.13 and 7.15 of \cite{buhlmann:2011} for further discussion.

These compatibility type constants are bounded away from zero for many standard design matrices, such as diagonal matrices, orthogonal designs, i.i.d. (including Gaussian) random matrices and matrices satisfying the `strong irrepresentability condition' of \cite{zhao:2006}. Details of these examples are provided in Section D in the supplement.

\section{Main results}\label{sec:main_results}

We now provide the main theoretical results of this paper concerning the frequentist behaviour of the VB posterior $\widetilde{\Pi}$ in the asymptotic regime $n,p\rightarrow \infty$. While the results are obtained assuming Gaussian noise in model \eqref{eq:model}, they are in fact robust to misspecification of the error distribution, see Remark B.1 in Section B. This robustness to misspecification is reflected in practice, see Section A.4 in the supplement for numerical results.

Our first result establishes contraction rates for the VB posterior to a sparse truth in $\ell_1$-loss, $\ell_2$-loss and \textit{prediction error} $\|X(\theta - \theta_0)\|_2$. Apart from the sparsity level, the rate also depends on compatibility. For $M>0$, set
\begin{equation}\label{tilde_compat}
\begin{split}
& \overline{\psi}_{M} (S) = \overline{\phi} \left( \left( 2 + \frac{4M}{A_4}\left( 1+ \frac{16}{\phi(S)^2}\frac{\lambda}{\bar{\lambda}} \right) \right)|S| \right),\\
& \widetilde{\psi}_{M} (S) = \widetilde{\phi} \left( \left( 2 + \frac{4M}{A_4}\left( 1+ \frac{16}{\phi(S)^2}\frac{\lambda}{\bar{\lambda}} \right) \right)|S| \right).
\end{split}
\end{equation}
In the natural case $\lambda \ll \bar{\lambda}$, these constants are asymptotically bounded from below by $\overline{\phi} ((2 + \tfrac{4M}{A_4})|S|)$ and $\widetilde{\phi} ((2 + \tfrac{4M}{A_4})|S|)$ if $\phi(S)$ is bounded away from zero. Our results are uniform over vectors in sets of the form
\begin{equation}
\Theta_{\rho_n,s_n}:=\{ \theta\in \mathbb{R}^{p}:\,\phi(S_0) \geq c_0, \quad |S_0| \leq s_n, \quad  \widetilde{\psi}_{\rho_n}(S_0) \geq c_0\},\label{eq: assump:compatibility}
\end{equation}
for $S_0 = S_{\theta_0}$, $s_n \geq 1$, $c_0>0$ and $\rho_n \to \infty$ (arbitrarily slowly).

\begin{theorem}[Recovery]\label{thm:recovery}
Suppose the model selection prior \eqref{eq:prior} satisfies \eqref{eq:prior_cond}, \eqref{prior_lambda} and $\lambda =O(\|X\| \sqrt{\log p}/s_n)$. Then the variational Bayes posterior $\widetilde{\Pi}$ satisfies, with $S_0 = S_{\theta_0}$,
\begin{align*}
& \sup_{\theta_0\in\Theta_{\rho_n,s_n}}  E_{\theta_0} \widetilde{\Pi} \left( \theta: \|X(\theta-\theta_0)\|_2\geq \frac{M\rho_n^{1/2} }{\overline{\psi}_{\rho_n}(S_0)} \frac{\sqrt{|S_0|\log p}}{\phi(S_0)} \right) \to 0,
\end{align*}
$$\sup_{\theta_0\in\Theta_{\rho_n,s_n}} E_{\theta_0} \widetilde{\Pi} \left( \theta : \|\theta-\theta_0\|_1 > \frac{M\rho_n }{\overline{\psi}_{\rho_n}(S_0)^2} \frac{|S_0| \sqrt{\log p}}{\|X\| \phi(S_0)^2}  \right) \to 0,$$
$$\sup_{\theta_0\in\Theta_{\rho_n,s_n}} E_{\theta_0} \widetilde{\Pi} \left( \theta : \|\theta-\theta_0\|_2 > \frac{M\rho_n^{1/2} }{\|X\| \widetilde{\psi}_{\rho_n}(S_0)^2} \frac{\sqrt{|S_0|\log p}}{\phi(S_0)}   \right) \to 0 $$
for any $\rho_n \to \infty$ (arbitrarily slowly), $\Theta_{\rho_n,s_n}$ defined in \eqref{eq: assump:compatibility} and where $M>0$ depends only on the prior. Moreover, the same holds true for the variational Bayes posteriors $\hat{Q}$ and $\widetilde{Q}$.
\end{theorem}

Theorem \ref{thm:recovery} follows directly from the oracle type Theorem \ref{thm:oracle_recovery} below upon setting $\theta_* = \theta_0$. Recall that we are working under the frequentist model where there is a ``true" $\theta_0$ generating data $Y$ of the form \eqref{eq:model}. Since the above rates equal the minimax estimation rates over $|S_0|$-sparse vectors, Theorem \ref{thm:recovery} states that the VB posterior puts most of its mass in a neighbourhood of optimal size around the truth with high $P_{\theta_0}$-probability in terms of $\ell_1$, $\ell_2$ and prediction loss. Thus for estimating $\theta_0$, the VB approximation behaves optimally from a theoretical frequentist perspective. This backs up the empirical evidence that VB can provide excellent scalable estimation.

The VB posterior mean often provides a good point estimator and the VB posterior is known to typically underestimate the marginal posterior variance (see e.g. \cite{blei:2016} - this is a result of using the KL divergence as optimization criterion). The combination of good centering point and the posterior shrinking at least as fast as the true posterior explains why the VB posterior still provides optimal recovery, despite the loss of information from using a mean-field approximation.

Since the prior and variational family do not depend on the unknown sparsity level $|S_0|$ and the VB estimate contracts around the truth at the minimax rate, the procedure is \textit{adaptive}. That is, the procedure can recover an $|S_0|$-sparse truth nearly as well as if we knew the exact level of sparsity of the unknown $\theta_0$. However, the choice of tuning parameters still has an effect on the finite-sample performance, see Section A.3 for a numerical investigation of the effect of the hyper-parameter $\lambda$. Note that Theorem \ref{thm:recovery} does not imply that the VB posterior $\widetilde{\Pi}$ converges to the true posterior $\Pi(\cdot|Y)$. Indeed, this is neither a typical situation nor a necessary property since the VB estimate should be substantially simpler than the true posterior to be useful.

Theorem \ref{thm:recovery} implies the variational families $\mathcal{Q}$ and $\mathcal{Q}_{MF}$ also provide optimal asymptotic estimation of $\theta_0$ in $\ell_1$, $\ell_2$ and prediction loss. {However, the corresponding optimization routine seems to yield worse performance in practice, see Section A.2.

An important motivation for using model selection priors is their ability to perform variable selection. The following result shows that the variational approximation puts most of its mass on models of size at most a multiple of the true dimension, thereby bounding the number of false positives.

\begin{theorem}[Dimension]\label{thm:dim}
Suppose the model selection prior \eqref{eq:prior} satisfies \eqref{eq:prior_cond}, \eqref{prior_lambda} and $\lambda =O(\|X\| \sqrt{\log p}/s_n)$. Then the variational Bayes posterior $\widetilde{\Pi}$ satisfies, with $S_0 = S_{\theta_0}$,
\begin{align*}
& \sup_{\theta_0\in\Theta_{\rho_n,s_n}}  E_{\theta_0} \widetilde{\Pi} \left( \theta: |S_\theta| \geq |S_0| + M\rho_n \left(1 +\tfrac{16}{\phi(S_0)^2}\tfrac{\lambda}{\bar{\lambda}} \right)|S_0|   \right) \to 0,
\end{align*}
for any $\rho_n \to \infty$ (arbitrarily slowly), $\Theta_{\rho_n,s_n}$ defined in \eqref{eq: assump:compatibility} and where $M>0$ depends only on the prior. Moreover, the same holds true for the variational Bayes posteriors $\hat{Q}$ and $\widetilde{Q}$.
\end{theorem}
Theorem \ref{thm:dim} follows directly from the oracle type Theorem \ref{thm:oracle_dim} below upon setting $\theta_* = \theta_0$. In the interesting case $\lambda \ll \bar{\lambda}$, the factor in Theorem \ref{thm:dim} can be simplified to $(1+M\rho_n)$ if the true parameter is compatible. Note also that under the conditions of Theorems \ref{thm:recovery} and \ref{thm:dim}, it is not possible to consistently estimate the true support $S_{\theta_0}$ of $\theta_0$ since one cannot separate small and exactly zero signals.

Since the variational families $\mathcal{Q}$ and $\mathcal{Q}_{MF}$ contain only distributions with a single support set $S$, the last statement says the resulting VB posteriors will select such a set of size at most a multiple times $|S_{0}|$ with high $P_{\theta_0}$-probability. The VB estimates based on these two variational families perform model selection in a hard-thresholding manner, reporting only whether a variable is selected or not. On the other hand, the more flexible family $\mathcal{P}_{MF}$ quantifies the individual variable selection via the reported non-trivial inclusion probabilities $0\leq \gamma_i \leq1$, and in this regard provides a richer approximation of the target posterior. Information on pairwise variable inclusion is obviously lost given the mean-field nature of the approximation. Nevertheless, it is interesting to note that all these families still permit good estimation of $\theta_0$.

We now provide more refined \textit{oracle}-type versions of Theorems \ref{thm:recovery} and \ref{thm:dim} as are known to hold for the true posterior \cite{castillo:2015}.

\begin{theorem}[Oracle recovery]\label{thm:oracle_recovery}
Suppose the model selection prior \eqref{eq:prior} satisfies \eqref{eq:prior_cond}, \eqref{prior_lambda} and $\lambda =O(\|X\| \sqrt{\log p}/s_n)$. For $\theta_0 \in \R^p \backslash \{0\}$, let $\theta_*\in \R^p$ be any vector satisfying $1 \leq s_* = |S_{\theta_*}| \leq |S_{\theta_0}| = s_0$ and $\|X(\theta_0 - \theta_*)\|_2^2 \leq (s_0 - s_*) \log p.$
Then the variational Bayes posterior $\widetilde{\Pi}$ satisfies, for any $\theta_*$ as above,
\begin{align*}
& \sup_{\theta_0\in\Theta_{\rho_n,s_n}}  E_{\theta_0} \widetilde{\Pi} \left( \theta: \|X(\theta-\theta_0)\|_2\geq \frac{M\rho_n^{1/2} }{\overline{\psi}_{\rho_n}(S_0)} \left[ \frac{\sqrt{s_*\log p}}{\phi(S_*)} +\|X(\theta_0-\theta_*)\|_2 \right] \right) \to 0,
\end{align*}
$$\sup_{\theta_0\in\Theta_{\rho_n,s_n}} E_{\theta_0} \widetilde{\Pi} \left( \theta : \|\theta-\theta_0\|_1 > \| \theta_0 - \theta_*\|_1 + \frac{M\rho_n }{\overline{\psi}_{\rho_n}(S_0)^2} \left[ \frac{s_* \sqrt{\log p}}{\|X\| \phi(S_*)^2} +\frac{\|X(\theta_0-\theta_*)\|_2^2}{\|X\| \sqrt{\log p}} \right] \right) \to 0,$$
$$\sup_{\theta_0\in\Theta_{\rho_n,s_n}} E_{\theta_0} \widetilde{\Pi} \left( \theta : \|\theta-\theta_0\|_2 > \frac{M\rho_n^{1/2} }{\|X\| \widetilde{\psi}_{\rho_n}(S_0)^2} \left[ \frac{\sqrt{s_*\log p}}{\phi(S_*)} +\|X(\theta_0-\theta_*)\|_2 \right]  \right) \to 0 $$
for any $\rho_n \to \infty$ (arbitrarily slowly), $\Theta_{\rho_n,s_n}$ defined in \eqref{eq: assump:compatibility} and where $M>0$ depends only on the prior. Moreover, the same holds true for the variational Bayes posteriors $\hat{Q}$ and $\widetilde{Q}$.
\end{theorem}
This can yield better rates than Theorem \ref{thm:recovery} for certain parameters and choices of $\theta_*$. For example, if $X = I$ is the identity matrix so that $\overline{\psi}_{\rho_n}(S) = \phi(S) = 1$ for all $S$, setting $\theta_* = 0$ yields
\begin{align*}
& \sup_{\theta_0} E_{\theta_0} \widetilde{\Pi} \left( \theta: \|\theta-\theta_0\|_2\geq M\rho_n^{1/2} \|\theta_0\|_2 \right) \to 0
\end{align*}
for any $\rho_n \to \infty$. If $\|\theta_0\|_2^2 \ll |S_0| \log p$, this improves upon the rate $\sqrt{|S_0|\log p}$ in Theorem \ref{thm:recovery} by accounting for the size of the coefficients of $\theta_0$ and not only its sparsity level.

The advantage of the oracle bound is it can take into account small non-zero coefficients of $\theta_0$ and capture its `effective sparsity'. If $S_* \subset S_0$, as one typically takes, the condition $\|X(\theta_0-\theta_*)\|_2^2 = \|X\theta_{0,S_*^c}\|_2^2\leq (s_0 - s_*) \log p$ implies that the coordinates of $\theta_0$ in $S_0 \backslash S_*$ contribute on average at most $\log p$ to the squared prediction error. Thus if the coefficient contributes less than $\log p$ to the squared prediction loss, it is preferable to assign it as bias rather than pay the full $\log p$ term required by the squared minimax rate $s_0 \log p$, which accounts only for sparsity irrespective of signal size.

\begin{theorem}[Oracle dimension]\label{thm:oracle_dim}
Suppose the model selection prior \eqref{eq:prior} satisfies \eqref{eq:prior_cond}, \eqref{prior_lambda},and $\lambda =O(\|X\| \sqrt{\log p}/s_n)$. For $\theta_0 \in \R^p \backslash \{0\}$, let $\theta_*\in \R^p$ be any vector satisfying $1 \leq s_* = |S_{\theta_*}| \leq |S_{\theta_0}| = s_0$ and $\|X(\theta_0 - \theta_*)\|_2^2 \leq (s_0 - s_*) \log p.$
Then the variational Bayes posterior $\widetilde{\Pi}$ satisfies, for any $\theta_*$ as above,
\begin{align*}
& \sup_{\theta_0\in\Theta_{\rho_n,s_n}}  E_{\theta_0} \widetilde{\Pi} \left( \theta: |S_\theta| \geq |S_*| + M\rho_n \left[  \left(1 +\tfrac{16}{\phi(S_*)^2}\tfrac{\lambda}{\bar{\lambda}} \right)|S_*|  + \tfrac{\|X(\theta_0-\theta_*)\|_2^2}{\log p} \right] \right) \to 0,
\end{align*}
for any $\rho_n \to \infty$ (arbitrarily slowly), $\Theta_{\rho_n,s_n}$ defined in \eqref{eq: assump:compatibility} and where $M>0$ depends only on the prior. Moreover, the same holds true for the variational Bayes posteriors $\hat{Q}$ and $\widetilde{Q}$.
\end{theorem}
Theorems \ref{thm:oracle_recovery} and \ref{thm:oracle_dim} are special cases of the finite-sample Theorems B.1 and B.2 in the supplement. Our proofs are based on the following crucial result, which allows one to exploit exponential probability bounds for the posterior to control the corresponding probability under the variational approximation. 

\begin{theorem}\label{thm: general:VB}
Let $\Theta_n$ be a subset of the parameter space, $A$ be an event and $Q$ be a distribution for $\theta$. If there exist $C>0$ and $\delta_n>0$ such that
\begin{align}\label{cond:exp_decay}
E_{\theta_0} \Pi(\theta\in \Theta_n|Y)1_A \leq Ce^{-\delta_n},
\end{align}
then 
\begin{align*}
E_{\theta_0} Q(\theta\in \Theta_n )1_A \leq \frac{2}{\delta_n} \Big[ E_{\theta_0} \emph{KL}(Q\| \Pi(\cdot|Y))1_A + Ce^{-\delta_n/2} \Big].
\end{align*}
\end{theorem}

\begin{proof}
Recall the duality formula for the Kullback-Leibler divergence (\cite{boucheron:2013}, Corollary 4.15)
$$\KL (Q\|P) = \sup_{f} \left[ \int f dQ - \log \int e^f dP \right],$$
where the supremum is taken over all measurable $f$ such that $\int e^f dP < \infty$. In particular,
$$\int f(\theta)dQ(\theta)\leq \KL\big(Q\big \| \Pi(\cdot|Y)\big)+ \log\int e^{f(\theta)}d\Pi(\theta|Y).$$
Applying this inequality with $f(\theta)= \tfrac{1}{2}\delta_n 1_{\Theta_n}(\theta)$ and using that $\log (1+x) \leq x$ for $x\geq 0$,
\begin{align*}
\tfrac{1}{2} \delta_n Q(\theta\in \Theta)1_A
& \leq \KL(Q\| \Pi (\cdot |Y))1_A + \log \Big(1+  \Pi(\theta\in\Theta_n|Y)e^{\delta_n/2}   \Big)1_A \\
&\leq \KL(Q\| \Pi (\cdot |Y))1_A+
  e^{\delta_n/2}\Pi(\theta\in\Theta_n|Y)1_A.
\end{align*}
Taking $E_{\theta_0}$-expectations on both sides and using \eqref{cond:exp_decay} gives the result.
\end{proof}
When deriving oracle rates for the original posterior, the exponent $e^{-\delta_n}$ in \eqref{cond:exp_decay} depends on the oracle quantity, see Section B.3. To apply Theorem \ref{thm: general:VB}, we must thus develop novel oracle type bounds on the KL divergence $\KL(\widetilde{\Pi}\| \Pi(\cdot|Y))$, which is the main technical difficulty in establishing our results, see Section B.2. The proof uses an iterative structure, using successive posterior localizations to eventually bound the KL divergence (see e.g. \cite{nicklray2019} for a similar idea).

\section{Variational Bayes algorithm}\label{sec: VBalgorithm}

\subsection{Coordinate update equations}

We now provide a coordinate-ascent variational inference (CAVI) algorithm (see for instance \cite{blei:2016}) to compute the mean-field VB posterior $\widetilde{\Pi}$ based on the spike-and-slab prior with Laplace slabs. Since in the literature \cite{logsdon2010variational,carbonetto:2012,huang:2016} the VB approximation is typically considered for Gaussian prior slabs, and can therefore take advantage of explicit analytic formulas, our algorithm requires modification.

Introducing binary latent variables $(z_i)_{i=1}^p$, the spike and slab prior can be rewritten as
\begin{equation}\label{def: spikeslab}
\begin{split}
w & \sim \text{Beta}(a_0,b_0),\\
z_i|w&\stackrel{iid}{\sim} \text{Bernoulli}(w),\\ 
\theta_i | z_i &\stackrel{ind}{\sim} z_i \text{Lap}(\lambda) + (1-z_i) \delta_0.
\end{split}
\end{equation}
The prior inclusion probability equals $\Pi(z_i=1)= \int w d\pi(w)=a_0/(a_0+b_0)$, the expectation of a beta random variable. In CAVI, we sequentially update the parameters $\gamma_i,\sigma_i,\mu_i$, $i=1,...,p$, of the VB posterior by minimizing the KL divergence between the variational class with the rest of the parameters kept fixed and the true posterior. We iterate this algorithm until convergence, measured by the change in entropy. 

%\begin{remark}In the recent paper \cite{huang:2016} for Gaussian slabs, the authors propose to deviate from the standard CAVI method (see for instance \cite{logsdon2010variational,carbonetto:2012} for a more standard component-wise CAVI implementation) and propose a batch-wise update of the parameters. This approach is, however, not really feasible in our setting due to the lack of analytic formulas when using Laplace slabs. Nevertheless, we show in our simulation study in Section \ref{sec:numerical} that our approach performs at least as well as, and often better than, both the component-wise and batch-wise variational algorithms applied for Gaussian slabs.
%\end{remark}

We now give the component-wise variational updates in the algorithm. Fixing the latent variable $z_i=1$ and all variational factors except $\mu_i$ or $\sigma_i$ (i.e. using vector notation, $\boldsymbol{\mu}_{-i},\boldsymbol{\sigma},\boldsymbol{\gamma}$ or $\boldsymbol{\mu},\boldsymbol{\sigma}_{-i},\boldsymbol{\gamma}$ are all fixed), the minimizer of the conditional KL divergence between $\mathcal{P}_{MF}$ and the posterior is the same as the minimizer of
\begin{equation}\label{eq: joint:KL}
\begin{split}
f_i(\mu_i|\boldsymbol\sigma,\boldsymbol\mu_{-i},\boldsymbol\gamma,z_i=1)&= \mu_i \sum_{k\neq i} (X^TX)_{ik} \gamma_k\mu_k+\frac{1}{2}(X^TX)_{ii}\mu_i^2-(Y^TX)_i\mu_i+\lambda\sigma_i\sqrt{2/\pi}e^{-\mu_i^2/(2\sigma_i^2)}\\
&\qquad +\lambda\mu_i(1-2\Phi(-\mu_i/\sigma_i)),\\
g_i(\sigma_i|\boldsymbol\sigma_{-i},\boldsymbol\mu,\boldsymbol\gamma,z_i=1)&= \tfrac{1}{2}(X^TX)_{ii}\sigma_i^2 +\lambda\mu_i\sigma_i\sqrt{2/\pi}e^{-\mu_i^2/(2\sigma_i^2)}+\lambda\mu_i(1-\Phi(\mu_i/\sigma_i))-\log \sigma_i,
\end{split}
\end{equation}
respectively (see Section C.1 of the supplement for the proof of the above assertion), where $\Phi$ denotes the cdf of the standard normal distribution. The minimizers of these functions do not have closed form expressions and hence must be computed by optimization; in our R implementation, we used the built-in optimize() function.

The minimizer $\gamma_i$ of the conditional KL divergence given $\boldsymbol{\mu},\boldsymbol{\sigma},\boldsymbol{\gamma}_{-i}$ solves
\begin{align}
\log \frac{\gamma_i}{1-\gamma_i}&=\log \frac{a_0}{b_0}+\log \frac{\sqrt{\pi}\sigma_i \lambda}{\sqrt{2}}+
(Y^TX)_i \mu_i- \mu_i \sum_{k\neq i} (X^TX)_{ik}\gamma_k\mu_k-\tfrac{1}{2}(X^TX)_{ii}(\sigma_i^2 + \mu_i^2) \nonumber\\
&\qquad- \lambda\sigma_i\sqrt{2/\pi}e^{-\mu_i^2/(2\sigma_i^2)}- \lambda\mu_i(1-2\Phi(-\mu_i/\sigma_i))+\frac{1}{2}=:\Gamma_i(\boldsymbol{\mu},\boldsymbol{\sigma},\boldsymbol{\gamma}_{-i}),\label{eq: VB:step:lambda}
\end{align}
see Section C.1 of the supplement for the proof.

Following \cite{huang:2016}, we terminate the procedure once the coordinate-wise maximal change in binary entropy of the posterior inclusion probabilities falls below a prespecified small threshold $\eps$ (e.g. $\eps=10^{-3}$), i.e. stop when $\Delta_H:= \max_{i=1,...,p} | H(\gamma_i)-H(\gamma_i')|\leq \eps$, where $H(p)=-p\log p-(1-p)\log(1-p)$, $p\in(0,1)$, and $\gamma_i$, $\gamma_i'$ are the $i$th coordinate of the starting and updated parameters $\boldsymbol{\gamma}$, $\boldsymbol{\gamma}'$, respectively. The full algorithm is present in Algorithm \ref{alg: VB}.

\subsection{Prioritized updating order}

The VB objective function is generally non-convex and so CAVI can be sensitive to initialization \cite{blei:2016}. It turns out the algorithm is also highly sensitive to the order of the component-wise updates. In fact, naively updating the coordinates in lexicographic order $i=1,...,p$ is typically suboptimal in our setting. We demonstrate in the next section on various simulated data sets that, unless the significant non-zero coefficients are located at the beginning of the signal, the procedure typically converges to a poor local minimum and gives misleading, inconsistent answers. In particular, CAVI returns a solution that is far from the desired VB posterior it is trying to compute. It is clearly undesirable that the algorithm's performance depends on the arbitrary ordering of the parameter coordinates. A natural fix is to randomize the order of the coordinate-wise updates and use different initializations, choosing the local minimum which provides the smallest overall KL-divergence to the posterior. We show, however, that due to the large number of local minima and their substantially different behaviour, this approach can also perform badly (although somewhat better than the lexicographic approach).

We instead propose a novel \textit{prioritized} update scheme. In a first preprocessing step, we compute an initial estimator $\hat{\mu}^{(0)}$ of the mean vector $\boldsymbol{\mu}$ of the variational class. We then place the coefficients in decreasing order with respect to the absolute value of their estimate and update the parameters coordinate-wise in the corresponding order, i.e. denoting by $\boldsymbol{a}=(a_1,...,a_p)$ the permutation of the indices $(1,2,...,p)$ such that $|\hat{\mu}^{(0)}_{a_i}|\geq |\hat{\mu}^{(0)}_{a_j}|$ for every $1\leq i<j\leq p$, we update the coordinates in the order $\mu_{a_i},\sigma_{a_i},\lambda_{a_i}$, $i=1,...,p$.

The intuition behind this method is that when CAVI begins by updating indices whose signal coefficients are small or zero in the target VB posterior, it may incorrectly assign signal strength to such indices to better fit the data (this is especially the case if the initialization value of the signal coefficient is far from its value in the target VB posterior). Consequently, the estimates of the significant non-zero signal components may be overly small since part of the signal strength has already been falsely assigned to signal coefficients that should in fact be small under the VB posterior. This can trap the algorithm near a poor local minimum from which it cannot escape, see the corresponding simulation study in Section \ref{sec:numerical}.

To avoid this, we wish to first update those coefficients which are large in the target VB posterior. Since these are unknown, the idea here is to identify them using a preliminary estimator: if the target VB posterior does a good job of estimating the signal, these large coefficients should roughly match those that are large in the true underlying signal, which can be identified using a reasonable estimator. The algorithm is given in Algorithm \ref{alg: VB}, where the function $order(|\boldsymbol{\mu}|)$ returns the indices of $|\boldsymbol{\mu}|$ in descending order.

\begin{algorithm}\label{alg:VB:Laplace}
\caption{Variational Bayes for Laplace prior slabs}\label{alg: VB}
\begin{algorithmic}[1]
\BState \textbf{Initialize}: $(\Delta_H,\boldsymbol{\sigma},\boldsymbol{\gamma})$, $\boldsymbol{\mu}:=\hat{\mu}^{(0)}$ (for a preliminary estimator $\hat{\mu}^{(0)}$), $\boldsymbol{a}:=order( |\boldsymbol{\mu}|)$
%\State{$\Delta_{H}:=1$}
%\State{$\boldsymbol{\mu}:=\hat{\mu}^{(0)}$ (for some preliminary estimator $\hat{\mu}^{(0)}$)}
%\State{$\boldsymbol{\sigma}:= (1,...,1)$, $\boldsymbol{\gamma}=\left(\tfrac{a_0}{a_0+b_0},\dots,\tfrac{a_0}{a_0+b_0} \right)$}
%\State{$\boldsymbol{a}:=order( |\boldsymbol{\mu}|)$}
\While{$ \Delta_{H}\geq \eps$}
\For {$ j=1$ to $p$}
\State{$i:= a_j$}
\State{$\mu_i:=\text{argmax}_{\mu_i}f_i(\mu_i|\boldsymbol\mu_{-i},\boldsymbol\sigma,\boldsymbol\gamma,z_i=1)$, \qquad \qquad \qquad \  // see equation \eqref{eq: joint:KL}}
\State{$\sigma_i:=\text{argmax}_{\sigma_i}g_i(\sigma_i|,\boldsymbol\mu,\boldsymbol\sigma_{-i},\boldsymbol\gamma,z_i=1)$,  \qquad \qquad \qquad  // see equation \eqref{eq: joint:KL}}
\State{$\gamma_{old,i}=\gamma_i$,\,  $\gamma_i=\text{logit}^{-1}\big(\Gamma_i(\boldsymbol{\mu},\boldsymbol{\sigma},\boldsymbol{\gamma}_{-i}) \big)$  \qquad \qquad \quad \quad \ // see equation \eqref{eq: VB:step:lambda}}
\EndFor
\State{ $\Delta_{H}:=\max_i\{ | H(\gamma_i)-H(\gamma_{old,i}) | \}$}
\EndWhile
\end{algorithmic}
\end{algorithm}
Instead of the prior \eqref{def: spikeslab}, one can instead take the $w_i \stackrel{iid}{\sim} \text{Beta}(a_0,b_0)$ and $z_i|w_i \sim^{ind} \text{Bernoulli}(w_i)$, so that the probabilities $w_i$ vary with $i$. This results in exactly the same variational algorithm since we are using a mean-field approximation. If one instead takes deterministic weights $w_i$, the above algorithm can be easily adapted by using the same update steps for $\mu_i$ and $\sigma_i$, while updating $\gamma_i$ as the solution to
\begin{align*}
\log \frac{\gamma_i}{1-\gamma_i}&=\log \frac{w_i}{1-w_i}+\log \frac{\sqrt{\pi}\sigma_i\lambda}{\sqrt{2}}+
(Y^TX)_i \mu_i- \sum_{j\neq i} (X^TX)_{ij}\gamma_j\mu_j\mu_i-\frac{\mu_i^2+\sigma_i^2}{2}(X^TX)_{ii} \\
&\qquad- \lambda\sigma_i\sqrt{2/\pi}e^{-\mu_i^2/(2\sigma_i^2)}- \lambda\mu_i(1-2\Phi(-\mu_i/\sigma_i))+\frac{1}{2}.
\end{align*}
The closely related algorithm for computing the VB posterior $\widetilde{Q}$ based on the family $\mathcal{Q}_{MF}$ is given in Algorithm \ref{alg:VB:Laplace2} in Section C.2 in the supplement.

\section{Numerical study}\label{sec:numerical}

In this section, we empirically compare the performance of our VB method using Laplace prior slabs, implemented in the \texttt{sparsevb} package \cite{sparsevb}, with various state-of-the-art Bayesian model selection methods on simulated data. We also demonstrate the importance of the prioritized updating scheme compared with standard CAVI implementations.

Additional numerical results are provided in the supplementary material as follows:
\begin{itemize}
\item[-] Section A.1: we apply our method and other Bayesian model selection methods to real world data.
\item[-] Section A.2: we show that Laplace prior slabs provide better estimation and model selection than Gaussian prior slabs. We also show that the optimization problem for finding the KL-optimizer for the class $\mathcal{Q}_{MF}$ is substantially harder than for the class $\mathcal{P}_{MF}$, with the former typically ending up at a poor local minimum.
\item[-] Section A.3: we show that although the theory indicates that the VB approach is (asymptotically) robust to the choice of the hyper-parameter $\lambda$, in finite-sample cases it can still have an effect and it may be helpful to use a data-driven choice in practice (e.g. cross validation).
\item[-] Section A.4: we show that several Bayesian model selection methods are robust to noise misspecification
\item[-] Section A.5: we compare different Bayesian model selection methods when the inputs are correlated.
\end{itemize}

We ran each experiment multiple times and report the average $\ell_2$-distance between the posterior mean (or maximum a posteriori (MAP) estimate for the SSLASSO) and the true parameter $\theta_0$, the false discovery rate (FDR), the true positive rate (TPR) and the computational time in seconds. We also report the standard deviations of these indicators to quantify their spread. For our computations, we used a MacBook Pro laptop with 2.9 GHz Intel Core i5 processor and 8 GB memory. Throughout the numerical study, we use the hyper-parameter choices $a_0=1$, $b_0=p$, $\lambda=1$ (except in Section A.3) and set the stopping threshold for the entropy change to $\Delta_H=10^{-5}$, see Algorithm \ref{alg: VB}. In each experiment and for every method, we take the ridge regression estimator $\hat\mu^{(0)}=(X^TX+I)^{-1}X^TY$ as initialization. Given the sparsity framework, it may be tempting to take the LASSO as initialization, however this is not recommended. The LASSO shrinks some coordinates to exactly zero and so is not suitable for $\mu$, which represents the estimated coefficients \textit{given} that they are included in the model, i.e. non-zero [the LASSO solution should be compared to $(\gamma_1\mu_1,\dots,\gamma_p\mu_p)$ rather than $(\mu_1,\dots,\mu_p)$].

%In this section we empirically investigate the performance of our variational Bayes method with Laplace slabs compared to various Bayesian (based) model selection methods on both simulated and real world data. In the simulation study we demonstrate the importance of the prioritized updating scheme and compare the performance of our method with other, state-of-the-art (Bayesian) model selection methods in case of unknown noise variance parameter $\sigma^2$. We also demonstrate the applicability of our method on an ozone interaction data set \cite{breiman:1985}. We further extend the simulation study in the appendix, when we demonstrate the advantages of working woth Laplace slabs instead of Gaussian slabs, compare the computational time of our algorithm with other variational Bayes methods, and finally consider and extended comparison of our method with other, state-of-the-art (Bayesian) model selection methods in case of known noise variance parameter $\sigma^2=1$.

\subsection{Prioritized updates}\label{sec:prioritized_sims}

We demonstrate here the relevance of our prioritized updating scheme for CAVI by comparing its performance with lexicographic and randomized updating orders, which are standard implementations for CAVI. We take $n=100$, $p=200$, $s=20$, $\theta_i=10$ for the non-zero coefficients, $\varsigma=1$ assumed to be known, $X_{ij}\stackrel{iid}{\sim}N(0,1)$, and consider four scenarios for the locations of the non-zero signal components. We place all non-zero coordinates (i) at the beginning of the signal, (ii) at the end of the signal, (iii) in the middle of the signal and (iv) uniformly at random. We ran the experiments 200 times and report the results in Table \ref{table:error:compare:update} (for the FDR and TPR, the $i$th coefficient is selected if $\gamma_i>0.5$). We also plot the posterior means resulting from a typical run in Figure \ref{fig: priority}.

Apart from the first scenario, where the significant signal coefficients are all located at the beginning of the signal, the prioritized method substantially outperforms both the randomized and lexicographic updating schemes for parameter estimation and model selection (recall that all three methods are trying to compute the \textit{same} VB estimate). The random updating order also slightly improves upon the lexicographic order, except for the first scenario, where the lexicographic order naturally updates the largest coefficients first. As well as being sensitive to initialization \cite{blei:2016}, it seems CAVI can also be very sensitive to the updating order of the parameters. Indeed, we see here that without prioritized ordering, the algorithm often terminates at poor local minima of the VB objective function. Since the VB objective is non-convex, naive (or random) update orderings may cause CAVI to return a solution that is far from the true minimizer of the KL divergence that it is trying to compute. Performing updates in a prioritized order can add some robustness against this, see Section \ref{sec: VBalgorithm} for some heuristics behind this idea. We also note that the runtime is comparable for the three updating orders.

\begin{figure}[tbp]
  \centering
  \begin{minipage}[b]{0.45\textwidth}
 \includegraphics[width=\textwidth]{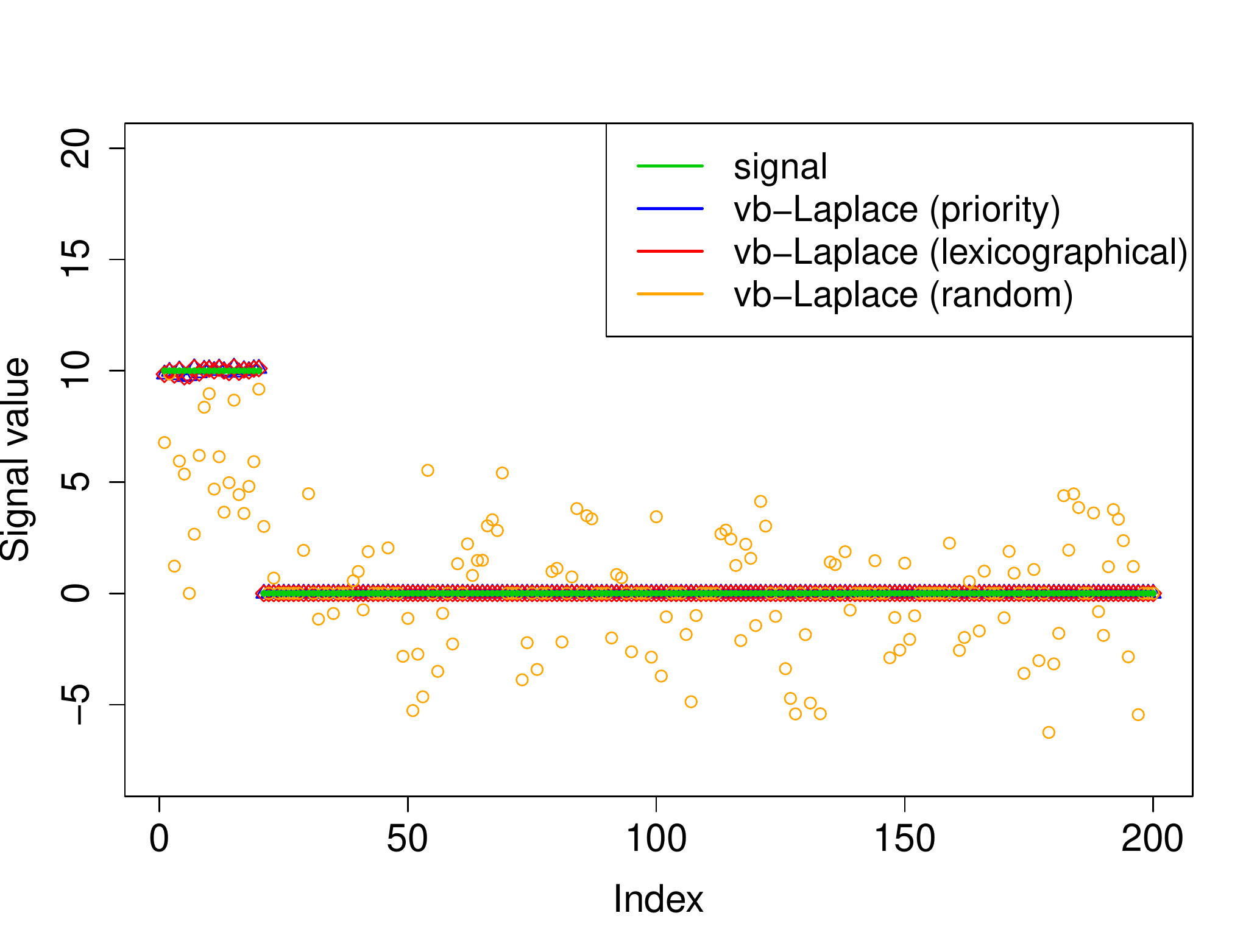}\\
\includegraphics[width=\textwidth]{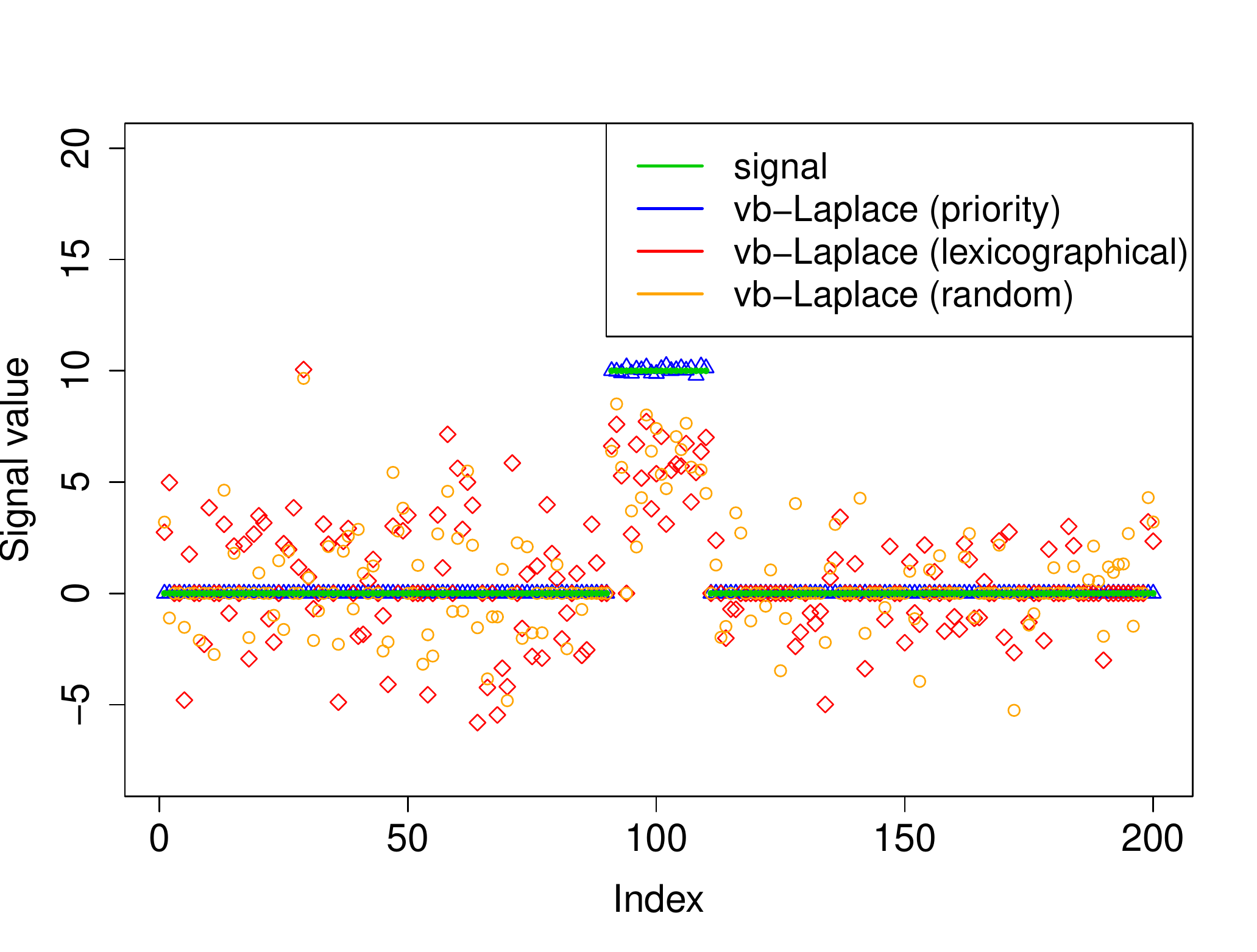}
  \end{minipage}
  \hfill
  \begin{minipage}[b]{0.45\textwidth}
  \includegraphics[width=\textwidth]{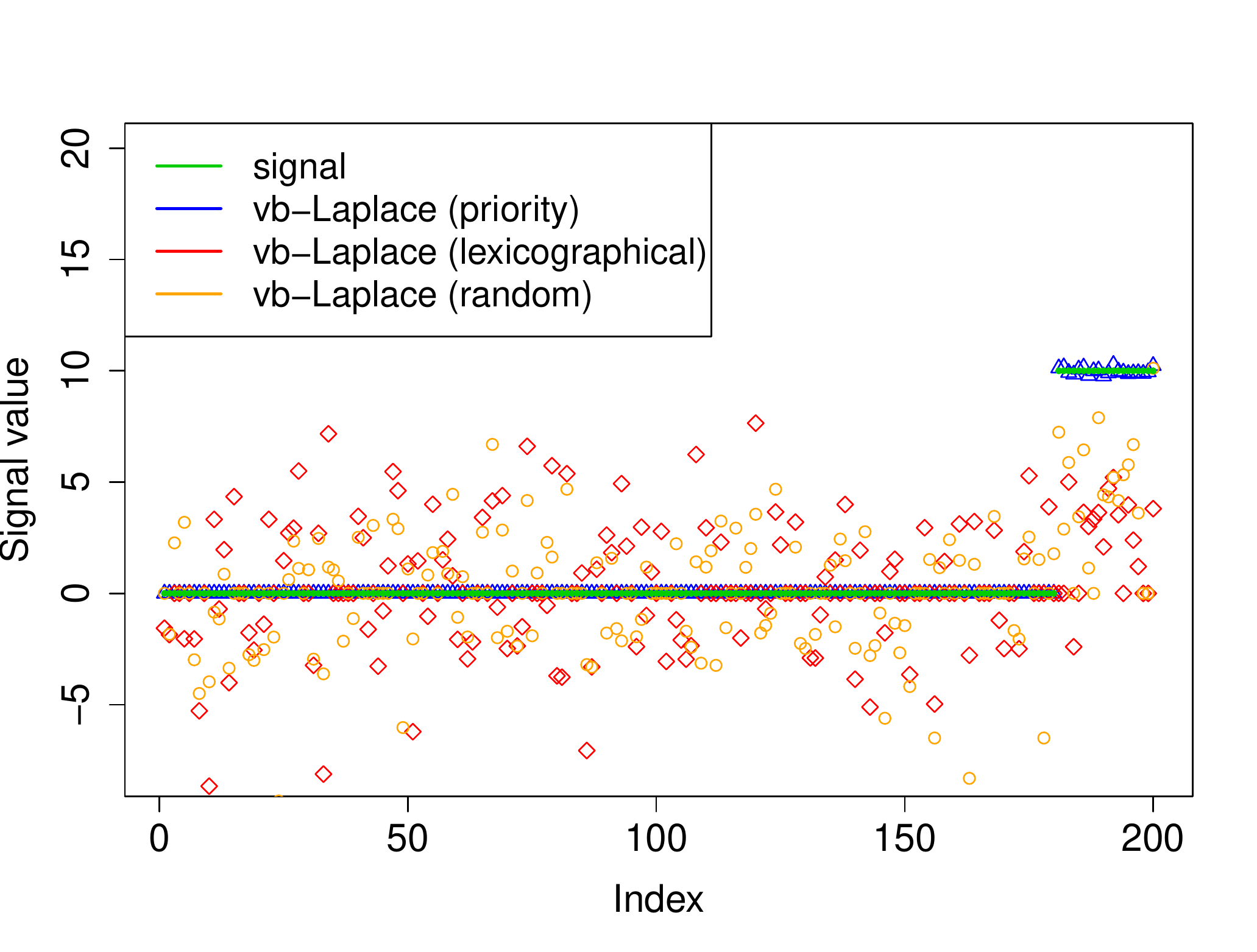}\\
\includegraphics[width=\textwidth]{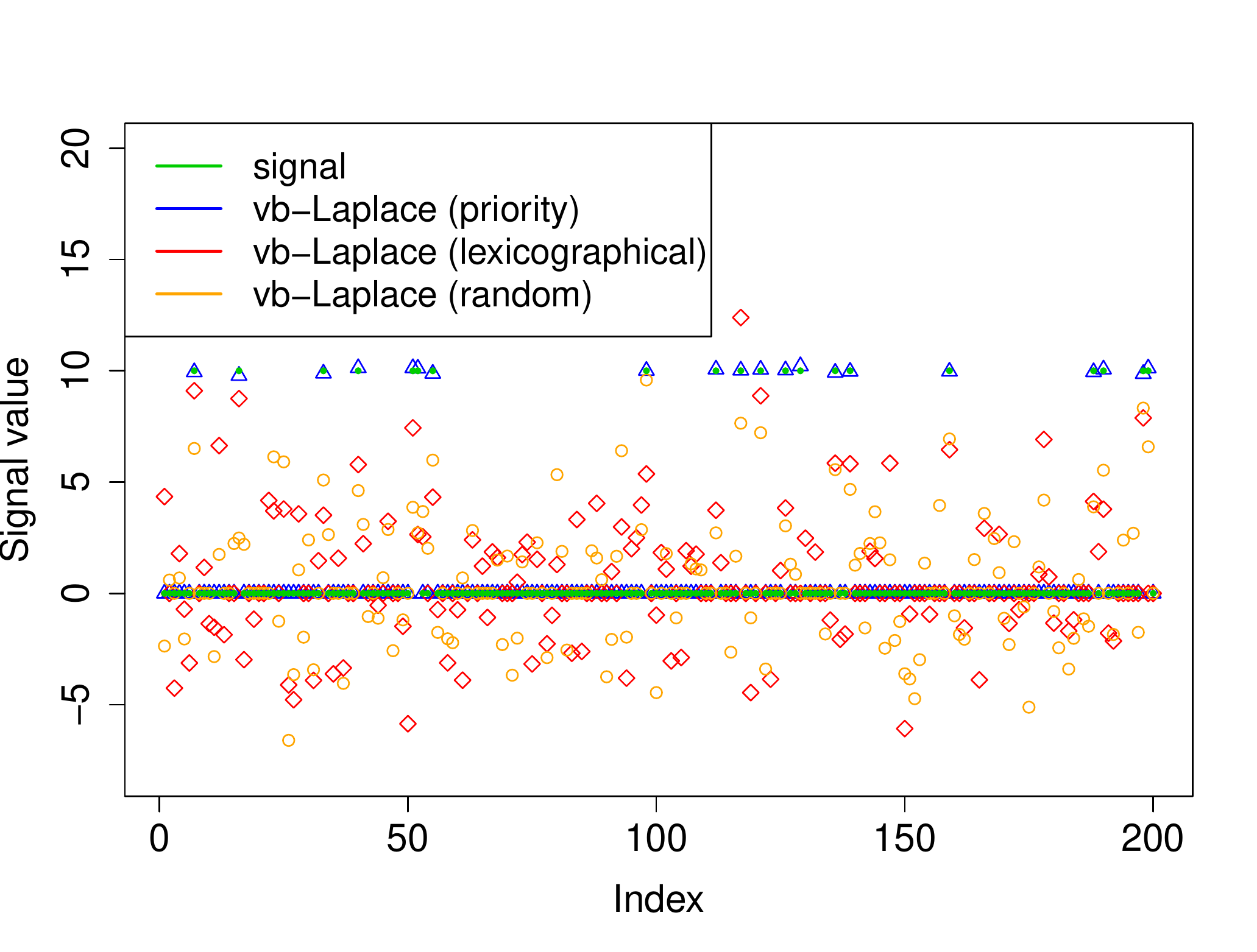}
  \end{minipage}
\caption{Linear regression with Gaussian design $X_{ij}\stackrel{iid}{\sim}N(0,1)$. We plot the underlying signal (green) and posterior means of the VB method with Laplace prior slabs computed using CAVI with parameter updates ordered in a prioritized way (blue), lexicographically (red) and randomly (light blue). We took $n=100$, $p=200$, $s=20$, $\theta_i=10$. From left to right and top to bottom we have: the non-zero coordinates are at the beginning, end, middle and at random locations of the signal. 
}
\label{fig: priority}
\end{figure}

\begin{table}[tbp]
\centering
\begin{footnotesize}
\begin{tabular}{ |l|l|c|c|c|c|    }
\hline
%\multicolumn{6}{ |c| }{Team sheet} \\
%\hline

Metric& Method&(i) & (ii)&(iii)&(iv) \\ \hline
\hline
\multirow{3}{*}{$\ell_2-\text{error}$} &prioritized & 1.03 $\pm$ 3.39  &\textbf{1.18 $\pm$ 3.86} &\textbf{1.06 $\pm$ 3.48} &\textbf{0.61 $\pm$ 1.65} \\
&lexicographic &\textbf{0.71 $\pm$ 2.14}  & 26.61 $\pm$ 15.04 &45.72 $\pm$ 5.45 &37.91 $\pm$ 5.63\\
 & randomized& 27.81 $\pm$ 13.30 &27.26 $\pm$ 13.78 & 25.14 $\pm$ 14.70& 35.08 $\pm$ 8.28\\
 \hline
\hline
\multirow{3}{*}{FDR} & prioritized &0.02 $\pm$ 0.12  &\textbf{0.02 $\pm$ 0.13} &\textbf{0.02 $\pm$ 0.12}& \textbf{0.05 $\pm$ 0.18}\\
&lexicographic&\textbf{0.01 $\pm$ 0.08}  &0.63 $\pm$ 0.35 &0.87 $\pm$ 0.03& 0.54 $\pm$ 0.38\\
 &randomized&0.68 $\pm$ 0.31&0.66 $\pm$ 0.32 &0.62 $\pm$ 0.352 & 0.69 $\pm$ 0.30\\
 \hline
\hline
\multirow{3}{*}{TPR} & prioritized &\textbf{1.00 $\pm$ 0.00}&\textbf{1.00 $\pm$ 0.01} &\textbf{1.00 $\pm$ 0.01}  & \textbf{1.00 $\pm$ 0.01}\\
&lexicographic &\textbf{1.00 $\pm$ 0.00}&0.93 $\pm$ 0.06 &0.75 $\pm$ 0.11  & 0.95 $\pm$ 0.05\\
 &randomized& 0.93 $\pm$ 0.07 &0.92 $\pm$ 0.06 &0.93 $\pm$ 0.07 &0.91 $\pm$ 0.07 \\
 \hline
\hline
\multirow{3}{*}{runtime (sec)} & prioritized &0.28 $\pm$ 0.09 & 0.24  $\pm$ 0.06 &0.26 $\pm$ 0.06 &0.24 $\pm$ 0.08\\
& lexicographic&\textbf{0.22 $\pm$ 0.06} & \textbf{0.21 $\pm$ 0.05} &\textbf{0.21 $\pm$ 0.04} &\textbf{0.23 $\pm$ 0.06}\\
 &randomized&0.24 $\pm$ 0.08& 0.22 $\pm$ 0.05 & 0.23 $\pm$ 0.05& 0.25 $\pm$ 0.06 \\
\hline
\hline
\end{tabular}
\caption{We compare the prioritized, lexicographic and random updating schemes in the CAVI algorithm.  We take $X_{ij}\stackrel{iid}{\sim}N(0,1)$, $n=100$, $p=200$, $s=20$, $\theta_i=10$ for the non-zero coefficients, which are located at the (i) beginning, (ii) middle, (iii) end, (iv) (uniformly) random locations of the signal. We report the means and standard deviations over 200 runs.}
\label{table:error:compare:update}
\end{footnotesize}
\end{table}

\subsection{Comparing Bayesian variable selection methods}\label{sec:unknown_var}

We consider here the realistic situation of unknown noise variance $\varsigma^2$, that is the model $Y = X\theta + \varsigma Z$. As mentioned in Section \ref{sec:prior+design} (see \eqref{rescaled_var}), dividing both sides of this model by an empirical estimator $\hat{\varsigma}$ for the noise standard deviation $\varsigma$ gives $\tilde{Y}=\tilde{X}\theta+\tilde{Z},$ where $\tilde{Y}=Y/\hat\varsigma$, $\tilde{X}=X/\hat\varsigma$ and $\tilde{Z}=(\varsigma/\hat{\varsigma})Z$, $Z\sim N(0,I_n)$. Endowing $\theta$ with the spike-and-slab prior and if the estimator $\hat{\varsigma}$ is close to $\varsigma$, we should approximately recover the $\varsigma=1$ case studied above. We thus compute our VB estimator as described above based on the design matrix $\tilde{X}$ and data $\tilde{Y}$. For estimating $\varsigma$, we have used the R package \texttt{selectiveInference}, see \cite{reid:2016}.

We compare the performance of our VB method with various Bayesian (based) variable selection algorithms for sparse linear regression using simulated data. We consider the \texttt{varbvs} R-package (variational Bayes for spike-and-slab priors with Gaussian prior slabs using an importance sampling outer circle for estimating the posterior inclusion probabilities and noise variance \cite{carbonetto:2012}), \texttt{EMVS} R-package (an expectation-maximization algorithm for spike-and-slab \cite{rockova:george:2018}), \texttt{SSLASSO} R-package (spike-and-slab LASSO \cite{rockova:george:2018b}) and `ebreg.R' R-function (a fractional likelihood empirical Bayes approach using MCMC for re-centered Gaussian slab priors \cite{martin:2017} - the function is available on the first author's website).

%We plot in Figure \ref{fig: compare:methods:unknown:sigma} the true signal (green) and posterior means for our VB algorithm (blue), varbvs (red), EMVS (light blue) and spikeslab (purple).

For \texttt{varbvs}, we set $tol = 10^{-4}$ and $maxiter = 10^4$. For \texttt{EMVS} we took $v_0\in\{0.1,0.2,...,2\}$, $v_1=1000$ (these quantities were used in one of the examples provided in the package), $a=1$, $b=p$ and $\epsilon=10^{-5}$ and report the posterior mean corresponding to the $v_0=0.1$ case.  For \texttt{SSLASSO}, we took  $\lambda_1=0.01$, $\lambda_0$ an arithmetic series between $\lambda_1$ and $p$ with 200 elements, set the variance ``unknown'', $a=1$, $b=p$, and penalty=``adaptive'', and report the results corresponding to the stabilized $\lambda_0$ value as recommended by the authors \cite{rockova:george:2018b}. In the ebreg algorithm, we took the default parameters $M=5000$, $\alpha=0.99$, $\gamma=0.001$ and used the \texttt{selectiveInference} R-package for the estimation of $\varsigma$. We note that for most of these methods, additional careful hyper-parameter tuning beyond the default settings can often lead to improved performance, see Section A.3 for our VB method or Section 5 of \cite{george:rockova:2020} for discussion concerning the SSLASSO.

We first consider (i) $n = 100$, $p = 400$, $s = 20$, $\varsigma=5$ with the non-zero signal coefficients set to $\theta_i=A$, with $A=\log n$, and located at the end of the signal. The entries of the design matrix are
taken to be iid normal random variables $X_{ij}\stackrel{iid}{\sim} N(0,1)$. In the other experiments, we take $(n,p,s,\varsigma)$ equal to (ii) $(100,1000,40,1)$ (with non-zero coefficients at the beginning of the signal) and set the non-zero parameters to be $1,2,3$; (iii) $(200,800,5,0.2)$ (in the middle) and take $\theta_i\stackrel{iid}{\sim}U(-5,5)$; (iv) $(100,400,20,5)$ (at the end) and take $\theta_i=2\log n$. We ran each algorithm 100 times and report the results in Table \ref{table:error:compare:methods:unknown:sigma}. 
%Additional simulation results are provided in Section \ref{sec:sim:exta_sim} of the supplement.
 Our method performs well compared to the other methods, in some cases providing substantially better estimation and model selection.

\begin{table}[tbp]
\centering

\begin{footnotesize}
\begin{tabular}{ |l|l|c|c|c|c|    }
\hline
%\multicolumn{6}{ |c| }{Team sheet} \\
%\hline

Metric& Method &(i) & (ii)&(iii)&(iv) \\ \hline
\hline
\multirow{3}{*}{$\ell_2-\text{error}$} & sparsevb &10.48 $\pm$ 6.84 & 0.21 $\pm$ 0.14 &\textbf{0.03 $\pm$ 0.01}&\textbf{6.55 $\pm$ 7.80}  \\
&varbvs &14.23 $\pm$ 6.51  & 0.18 $\pm$ 0.07 &\textbf{0.03 $\pm$ 0.01} &20.43 $\pm$ 17.15 \\
 & EMVS& 14.02$\pm$ 2.46&3.57 $\pm$ 0.03 & 5.04 $\pm$ 0.33 & 21.52 $\pm$ 11.29 \\
 & SSLASSO &20.62 $\pm$ 0.17 & \textbf{0.16 $\pm$ 0.11} & 0.09 $\pm$ 0.12 &  37.92 $\pm$ 9.84 \\
 & ebreg & \textbf{9.38 $\pm$ 6.05} & 0.18 $\pm$ 0.07& 0.17 $\pm$ 0.04 & 7.39 $\pm$ 7.42 \\
 \hline
\hline
\multirow{3}{*}{FDR} & sparsevb&0.12 $\pm$ 0.17  &0.06 $\pm$ 0.16 & \textbf{0.00 $\pm$ 0.00}& 0.02 $\pm$ 0.07\\
& varbvs& 0.06 $\pm$ 0.11  &0.01 $\pm$ 0.04 & \textbf{0.00 $\pm$ 0.00}&0.07 $\pm$ 0.15\\
 & EMVS&0.24 $\pm$ 0.13&\textbf{0.00 $\pm$ 0.00} & \textbf{0.00 $\pm$ 0.00} & 0.43 $\pm$ 0.25\\
 & SSLASSO &\textbf{0.00 $\pm$ 0.00}&\textbf{0.00 $\pm$ 0.00}& \textbf{0.00 $\pm$ 0.00}&\textbf{0.00 $\pm$ 0.00} \\
 & ebreg &0.38 $\pm$ 0.20 &0.01 $\pm$ 0.02 &  \textbf{0.00 $\pm$ 0.00}  & 0.28 $\pm$ 0.16\\
\hline
\hline
\multirow{3}{*}{TPR} & sparsevb & 0.70 $\pm$ 0.31&\textbf{1.00 $\pm$ 0.00} & 0.96 $\pm$ 0.13   & 0.94 $\pm$ 0.18  \\
& varbvs &0.340 $\pm$ 0.37&\textbf{1.00 $\pm$ 0.00} &0.57 $\pm$ 0.43 &0.53 $\pm$ 0.44\\
 & EMVS& 0.59 $\pm$ 0.14&0.00 $\pm$ 0.00 &0.86 $\pm$ 0.09 & 0.88 $\pm$ 0.10\\
 & SSLASSO & 0.01 $\pm$ 0.01 &0.94 $\pm$ 0.13  &0.10 $\pm$ 0.29  & 0.09 $\pm$ 0.28\\
 & ebreg & \textbf{0.88 $\pm$ 0.18}&\textbf{1.00 $\pm$ 0.00}  & \textbf{0.98 $\pm$ 0.07} &  \textbf{1.00 $\pm$ 0.04} \\
 \hline
\hline
\multirow{3}{*}{runtime (sec)} & sparsevb &0.43 $\pm$ 0.27  & 0.71 $\pm$ 0.21&0.35 $\pm$ 0.25 & 0.65 $\pm$ 0.53\\
& varbvs&0.60 $\pm$ 0.28  & 2.02 $\pm$ 0.50 &0.51 $\pm$ 0.38& 0.56 $\pm$ 0.23\\
 & EMVS&0.20 $\pm$ 0.07& 1.72 $\pm$ 0.44 & 0.21 $\pm$ 0.05& 0.19 $\pm$ 0.09\\
 & SSLASSO  &\textbf{0.06 $\pm$ 0.03} &\textbf{0.37 $\pm$ 0.11} & \textbf{0.06 $\pm$ 0.01}  &  \textbf{0.07 $\pm$ 0.03}   \\
 & ebreg & 35.05 $\pm$ 7.03 & 21.20 $\pm$ 6.05 & 31.42 $\pm$ 4.01 &   36.33 $\pm$ 9.13 \\
\hline
\hline
\end{tabular}
\caption{Linear regression with Gaussian design $X_{ij}\stackrel{iid}{\sim}N(0,1)$, unknown noise variance $\varsigma^2$ and non-zero signal coefficients $\theta_i=A$, with parameter $(n,p,s,A,\varsigma)$ choices (i) $(100,400,20,\log n,5)$ (non-zero coefficients at the beginning); (ii) $\big(100,1000,3,(1,2,3),1\big)$ (at the end); (iii) $(200,800,5,\stackrel{iid}{\sim}U(-5,5),0.2)$ (in the middle); (iv)  $(100,400,20,2\log n,5)$ (at the end) .}
\label{table:error:compare:methods:unknown:sigma}
\end{footnotesize}
\end{table}

\section{Conclusion}\label{sec:conclusion}

We studied theoretical oracle contraction rates of a natural sparsity-inducing mean-field VB approximation to posteriors arising from widely used, but computationally challenging, model selection priors in high-dimensional sparse linear regression. We showed that under compatibility conditions on the design matrix, such an approximation converges to a sparse truth at an oracle rate in $\ell_1$, $\ell_2$ and prediction loss, implying optimal (minimax) recovery, and also performs suitable dimension selection. This provides a theoretical justification for this approximation algorithm in a sparsity context. Minimax guarantees for this VB method extend to high-dimensional logistic regression, as we show in the follow up work \cite{ray:szabo:clara:2020}.

We investigated the empirical performance of our algorithm via simulated and real world data and showed that it generally performs at least as well as other state-of-the-art Bayesian variable selection methods, including existing VB approaches. We also demonstrated how the widely used coordinate-ascent variational inference (CAVI) algorithm can be highly sensitive to the updating order of the parameters. We therefore proposed a novel prioritized updating scheme that uses a data-driven updating order and performs better in simulations. This idea may be applicable for CAVI approaches in other settings. Our variational algorithm is implemented in the R-package \texttt{sparsevb} \cite{sparsevb}.

\bigskip

\textbf{Acknowledgements.} We would like to thank two referees for valuable comments that helped considerably improve this manuscript.

\bigskip
\begin{center}
{\large\bf SUPPLEMENTARY MATERIAL}
\end{center}

\begin{description}
\item [Supplementary material.] In Section A, additional numerical results are given. First, we provide a real world data example, where we compare Bayesian model selection methods. We then consider various VB methods, demonstrating the advantages of using Laplace instead of Gaussian prior slabs, investigate the effect of the hyper-parameter $\lambda$ and further study Bayesian variable selection methods under noise misspecification and correlated inputs. Section B contains full oracle results and all proofs, Section C contains additional methodological details and Section D contains further discussion of the design matrix assumption, including examples.
\end{description}

%\newpage

\appendix

\renewcommand{\theequation}{\Alph{section}.\arabic{equation}}
\setcounter{equation}{0}

\section{Additional numerical results}\label{sec:Num:Stud:Extra}

\subsection{Ozone interaction data}\label{sec:realworld}
We apply our method to the real world ozone interaction data investigated in \cite{breiman:1985}. The dataset contains $n=203$ readings of maximal daily ozone measured in the Los Angeles basin and $p=134$ variables modeling the pairwise interaction of 9 meteorological and 3 time variables. We firstly normalize the design matrix by centering and rescaling each column to have Euclidean norm equal to $\sqrt{n}$ and then add a column vector of ones to add an intercept to the model.\footnote{Except for EMVS, since adding an intercept resulted in an error message.} We apply the four methods investigated above (i.e. our method \texttt{sparsevb} \cite{sparsevb}, \texttt{varbvs}, \texttt{EMVS}, \texttt{SSLASSO}) with unknown noise variance $\varsigma^2$, using the method settings described in Section \ref{sec:unknown_var}. We also tried to apply the ebreg method, but due to the highly co-linear nature of the design matrix, the code gave errors when trying to compute the Cholesky decomposition.

As we do not know the underlying truth, we consider the 10-fold cross validation prediction error, i.e. we use nine folds to compute the posterior mean or MAP $\hat{\theta}$ and then use the 10th fold to compute the prediction error $\|Y-X\hat\theta\|_2$. We report the averaged out cross-validation errors in Table \ref{table:cverror:realworld}, together with the runtimes and number of selected covariates. Our method outperforms the other approaches in cross-validated prediction loss. Furthermore, while there is some overlap between the models selected by the various methods, the results are quite different, see Figure \ref{fig: realworld}.

\begin{figure}[tbp]
\center
  \includegraphics[width=0.7\textwidth]{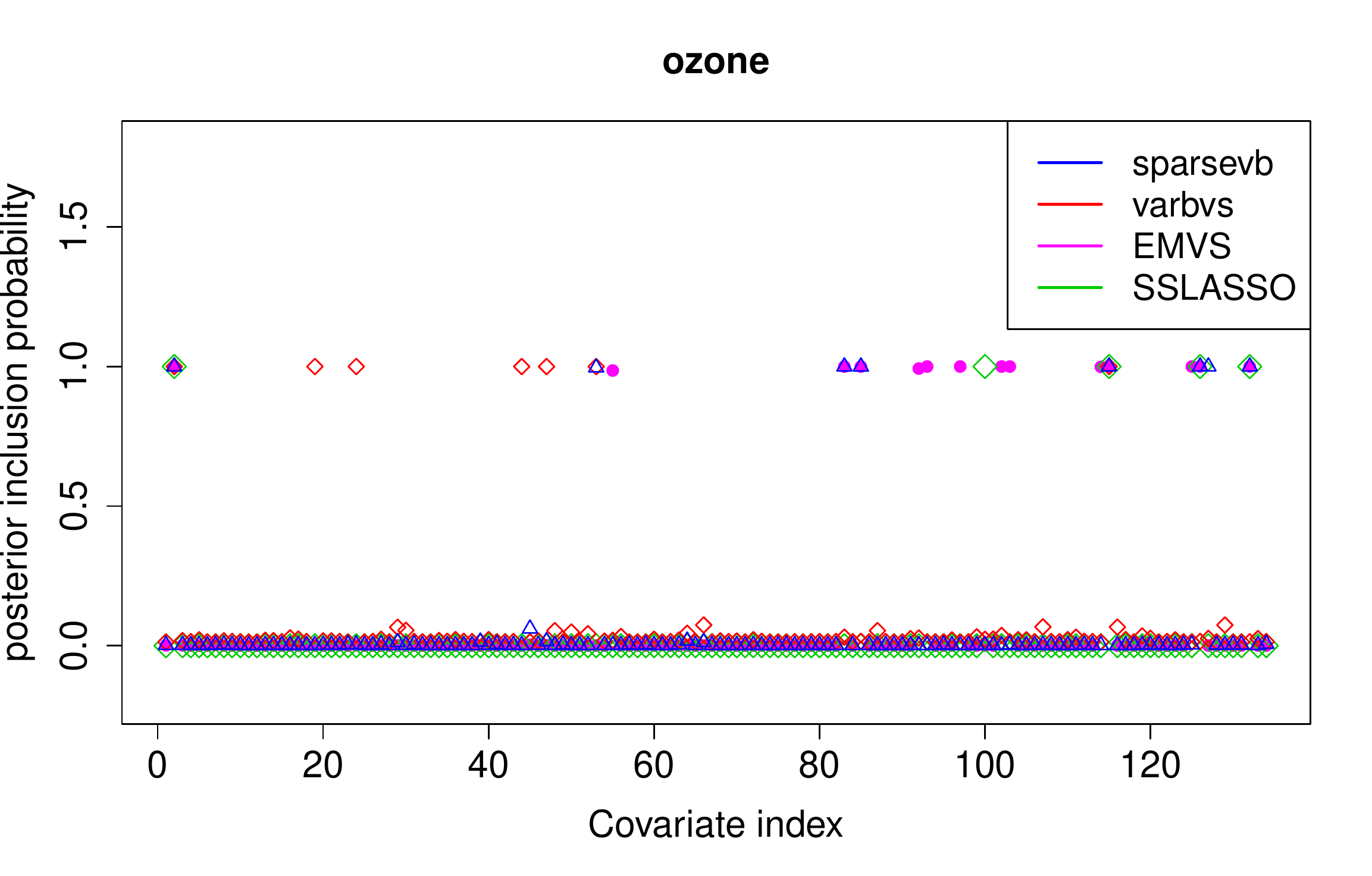}\\
\caption{Marginal inclusion probabilities of the variables for the ozone interaction data using sparsevb (blue), EMVS (purple), SSLASSO (green) and varbvs (red).}
\label{fig: realworld}
\end{figure}

\begin{table}[tbp]
\centering
\caption{Cross-validated $\ell_2$-estimation error of Bayesian model selection methods}
\label{table:cverror:realworld}
 \setlength{\tabcolsep}{1pt}%
 \renewcommand{\arraystretch}{1.3}%
\begin{tabular}{|l||l|l|l|l|l|l|l|}
 %\hline
 %\multicolumn{8}{|c|}{Approximation Error} \\
 %\hline
 \hline
data~\textbackslash~Method\,&\, sparsevb \,&\, varbvs\, &\,EMVS\,&\, SSLASSO\,  \\
 \hline
 \hline
\text{CV error}&
  \, $\textbf{16.43}$
&\,$59.49$
 &\, $74.45$
 &\, $53.28$
\\
\hline
model size&
  \, $ 9$
&\, $7 $
 &\, $14 $
 &\, $5  $
\\
\hline
runtime (sec)&
  \, $1.49$
&\, $1.14$
 &\, $\textbf{0.02} $
 &\, $0.10  $
\\
\hline
\end{tabular}
\end{table}

\subsection{Comparing the VB algorithms}\label{sec:sim:compare:vb} 

We compare our VB method with Laplace slabs (Algorithm \ref{alg: VB}) with different variations of the VB algorithm. First, we consider the other mean-field VB posterior $\widetilde{Q}$ derived from the variational class $\mathcal{Q}_{MF}$ (Algorithm \ref{alg:VB:Laplace2}  in Section \ref{sec:gauss_VB_algos}). Next, we consider the VB method with Gaussian prior slabs, which is the standard choice in the literature, see for instance \cite{logsdon2010variational,carbonetto:2012,huang:2016}, both with component-wise and batch-wise computational approaches, see Algorithms \ref{alg: VB_Gauss} and \ref{alg: VB_Gauss:batch} in Section \ref{sec:gauss_VB_algos}. To compensate for the over-shrinkage of the posterior mean caused by the light tail of the Gaussian slabs, we also consider centered Gaussian prior slabs with standard deviation set to the (unknown) oracle $\rho=\|\theta_0\|_2$, as proposed by \cite{castillo:2012} for the sequence model (i.e. $X=I$ the identity matrix).

In all experiments, we placed the non-zero signal components $\theta_i=A$ at the beginning of the signal. In the first experiment, (i) we take the identity design matrix $X=I_n$ and set $n=p=400$, $s=40$, $A=4\sqrt{\log n}$. In the other three experiments, we consider a Gaussian design matrix with entries $X_{ij} \stackrel{iid}{\sim} N(0,\tau^2)$ and vary the parameters  $n,p, s, \tau$ and $A$. We take (ii) $(n,p, s, \tau)=(100,200,20,1)$, $A\stackrel{iid}{\sim}U(0,2\log n)$; (iii) $(n,p, s, \tau)=(200,800,40,0.1)$, $A=2\log n$; (iv) $(n,p, s, \tau)=(100,400,15,0.5)$, $A\stackrel{iid}{\sim}U(-8,8)$. In all experiments, we take $\varsigma=1$ assumed to be known. The results over 200 runs are reported in Table \ref{table:compare:various:VB} and we plot the outcome of a typical run in Figure \ref{fig: linreg1}.

Our Laplace VB method (sparsevb) with variational class $\mathcal{P}_{MF}$ typically outperforms the other VB algorithms. From the identity design case (i), it is clear that Gaussian prior slabs provide suboptimal recovery for $\theta_0$ unless the prior slab variance is rescaled by the norm of $\theta_0$. However, the rescaled Gaussian slabs perform much less well in the Gaussian design cases (ii)-(iv). The other mean-field variational class $\mathcal{Q}_{MF}$ performs similarly to our main method in the identity design case, but significantly worse in the more complicated Gaussian design cases. This is due to discrete nature of the variational parameter $\gamma\in \{0,1\}$ in this family, which makes the optimization problem even more difficult, causing the method to frequently get stuck at a poor local minimum. We do not report run times as the \texttt{sparsevb} R-package is optimized for computation and therefore runs substantially faster than the other methods, which are more simply implemented.

\begin{table}[tbp]
\centering
\begin{footnotesize}
\begin{tabular}{ |l|l|c|c|c|c|    }
\hline
%\multicolumn{6}{ |c| }{Team sheet} \\
%\hline

Metric& Method$\backslash$ Experiment &(i) & (ii)&(iii)&(iv) \\ \hline
\hline
\multirow{4}{*}{$\ell_2-\text{error}$} &Laplace $\mathcal{P}_{MF}$ &8.80 $\pm$ 0.85  & \textbf{1.30 $\pm$ 0.26} &\textbf{9.25 $\pm$ 9.73} &\textbf{1.08 $\pm$ 0.20} \\
&Laplace $\mathcal{Q}_{MF}$ &8.80 $\pm$ 0.85  & 6.73 $\pm$ 1.79 &39.98 $\pm$ 6.88 &6.56 $\pm$ 1.97 \\
 & Gauss& 31.06 $\pm$ 0.49 &1.93 $\pm$ 0.51 &43.58 $\pm$ 2.94&1.40 $\pm$ 0.29  \\
 & Gauss (batch-wise) &31.11 $\pm$ 0.48 &16.38 $\pm$ 0.79 &66.98 $\pm$ 0.00&18.03 $\pm$ 0.00 \\
 & Gauss ($\rho=\|\theta_0\|_2$)&\textbf{6.26 $\pm$ 0.72}& 1.42 $\pm$ 0.32 &58.12 $\pm$ 19.01&2.05 $\pm$ 3.59\\ \hline
\hline
\multirow{4}{*}{FDR} & Laplace $\mathcal{P}_{MF}$&\textbf{0.00 $\pm$ 0.00}  &\textbf{0.00 $\pm$ 0.01} &\textbf{0.03 $\pm$ 0.11}& \textbf{0.00 $\pm$ 0.02}\\
& Laplace $\mathcal{Q}_{MF}$&\textbf{0.00 $\pm$ 0.00}  &0.70 $\pm$ 0.07 &0.45 $\pm$ 0.08& 0.55 $\pm$ 0.14\\
 & Gauss&\textbf{0.00 $\pm$ 0.00}&\textbf{0.00 $\pm$ 0.00} &0.50 $\pm$ 0.03& 0.01 $\pm$ 0.03\\
 & Gauss (batch-wise) &\textbf{0.00 $\pm$ 0.00}&0.87 $\pm$ 0.01&0.62 $\pm$ 0.03& 0.82 $\pm$ 0.03 \\
 & Gauss ($\rho=\|\theta_0\|_2$)&\textbf{0.00 $\pm$ 0.00}&0.25 $\pm$ 0.36 &0.57 $\pm$ 0.21& 0.06 $\pm$ 0.21\\ \hline
\hline
\multirow{4}{*}{TPR} & Laplace $\mathcal{P}_{MF}$ &\textbf{1.00 $\pm$ 0.00}&\textbf{0.89 $\pm$ 0.02} &\textbf{0.99 $\pm$ 0.06}   & \textbf{0.81 $\pm$ 0.03}\\
& Laplace $\mathcal{Q}_{MF}$ &\textbf{1.00 $\pm$ 0.00}&0.81 $\pm$ 0.06 &0.88 $\pm$ 0.08  & 0.74 $\pm$ 0.07\\
 & Gauss&\textbf{1.00 $\pm$ 0.00}&\textbf{0.89 $\pm$ 0.02} &0.94 $\pm$ 0.05 &\textbf{0.81 $\pm$ 0.03} \\
 & Gauss (batch-wise)& \textbf{1.00 $\pm$ 0.00}&0.88 $\pm$ 0.07  &0.82 $\pm$ 0.07  &0.68 $\pm$ 0.08 \\
 & Gauss ($\rho=\|\theta_0\|_2$)&\textbf{1.00 $\pm$ 0.00}&0.81 $\pm$ 0.10  &0.58 $\pm$ 0.17 & 0.78 $\pm$ 0.10 \\ \hline
\hline
\end{tabular}
\caption{Linear regression with (i) identity design $X=I_n$,  and $(ii)-(iv)$ Gaussian design $X_{ij}\stackrel{iid}{\sim}N(0,\tau^2)$. The non-zero coefficients are located in the beginning of the signal. The parameters $(n,p,s,A)$ are set to (i) $(400,400,40,4\sqrt{\log n})$; (ii) $(100,200,20,U(0,2\log(n)))$; (iii) $(200,800,40,2\log n)$; (iv) $(100,400,15,U(-8,8))$. We set (ii) $\tau=1$; (iii) $\tau=0.1$; (iv) $\tau=0.5$. We compare the means and standard deviations over 200 runs for our method and other variations of the VB algorithm.}
\label{table:compare:various:VB}
\end{footnotesize}
\end{table}

\begin{figure}[tbp]
  \centering
  \begin{minipage}[b]{0.45\textwidth}
 \includegraphics[width=\textwidth]{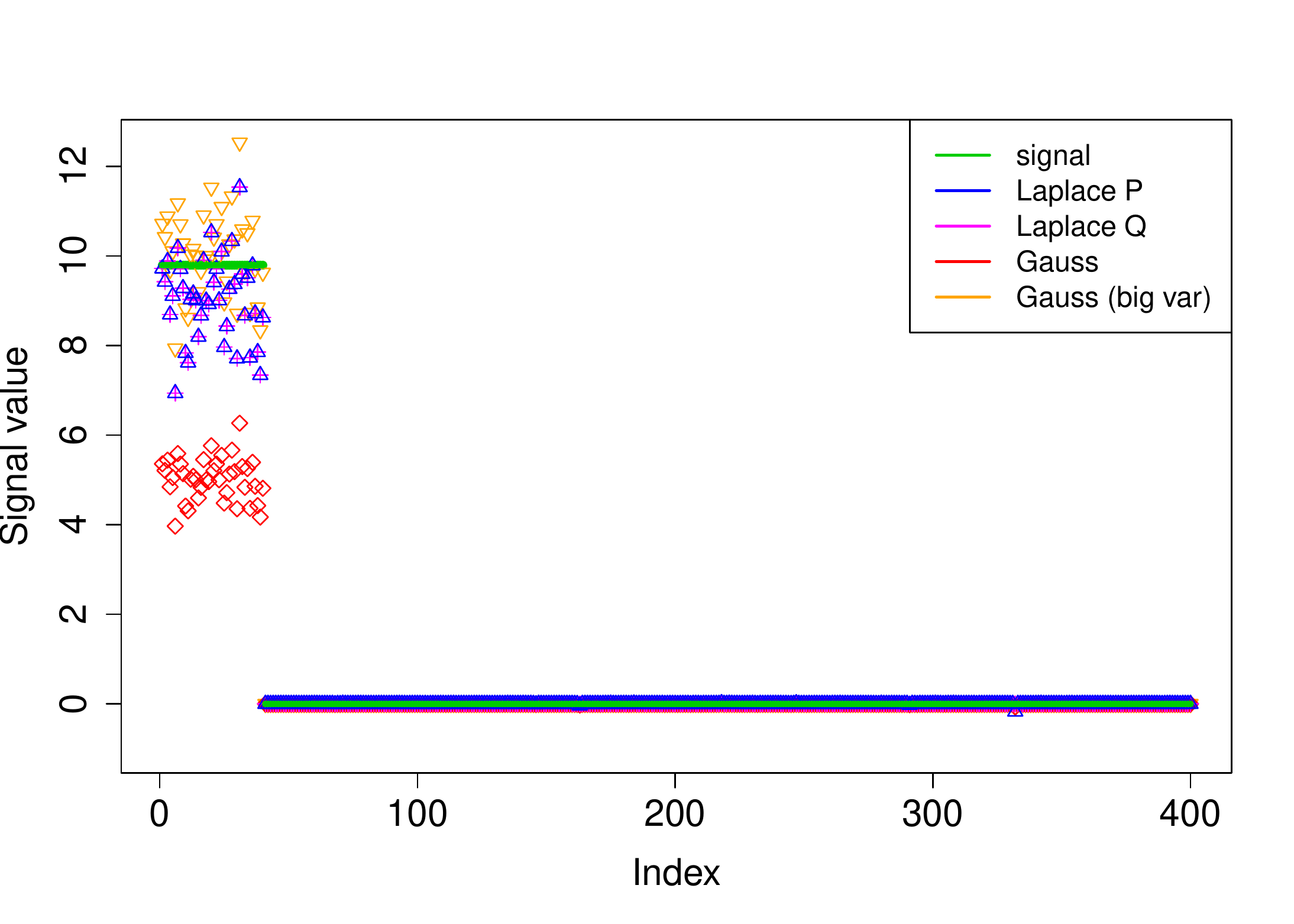}\\
\includegraphics[width=\textwidth]{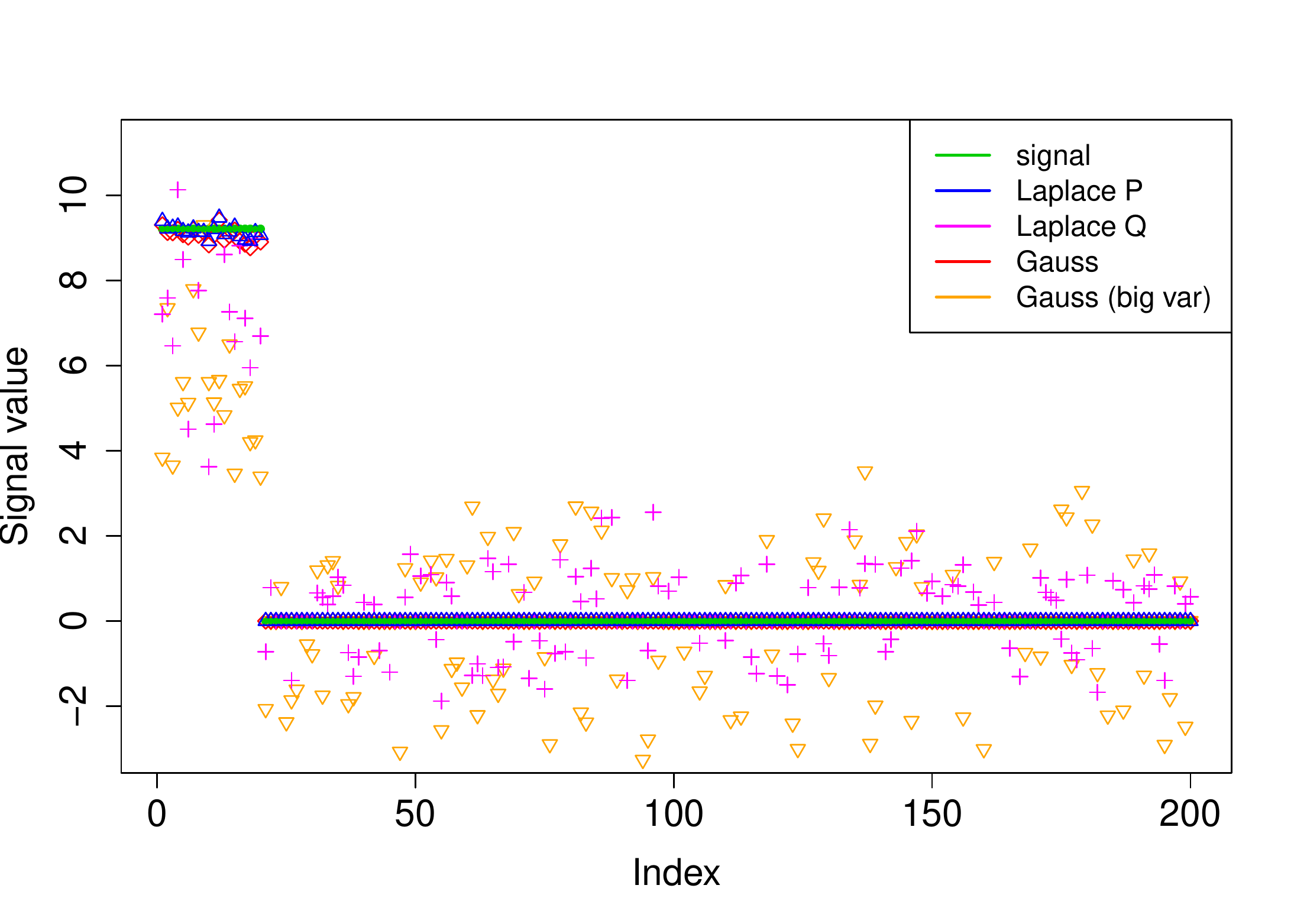}
  \end{minipage}
  \hfill
  \begin{minipage}[b]{0.45\textwidth}
  \includegraphics[width=\textwidth]{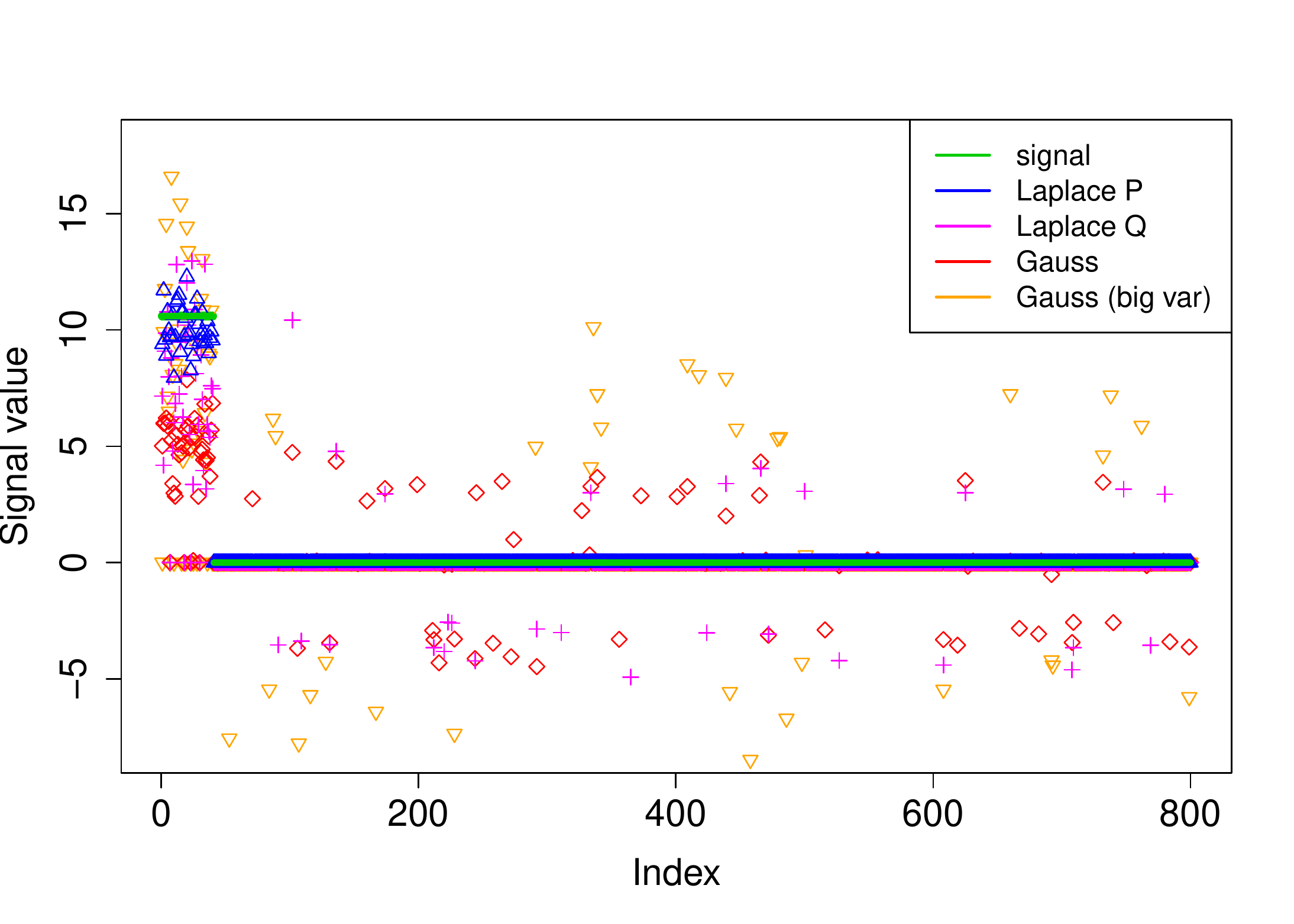}\\
\includegraphics[width=\textwidth]{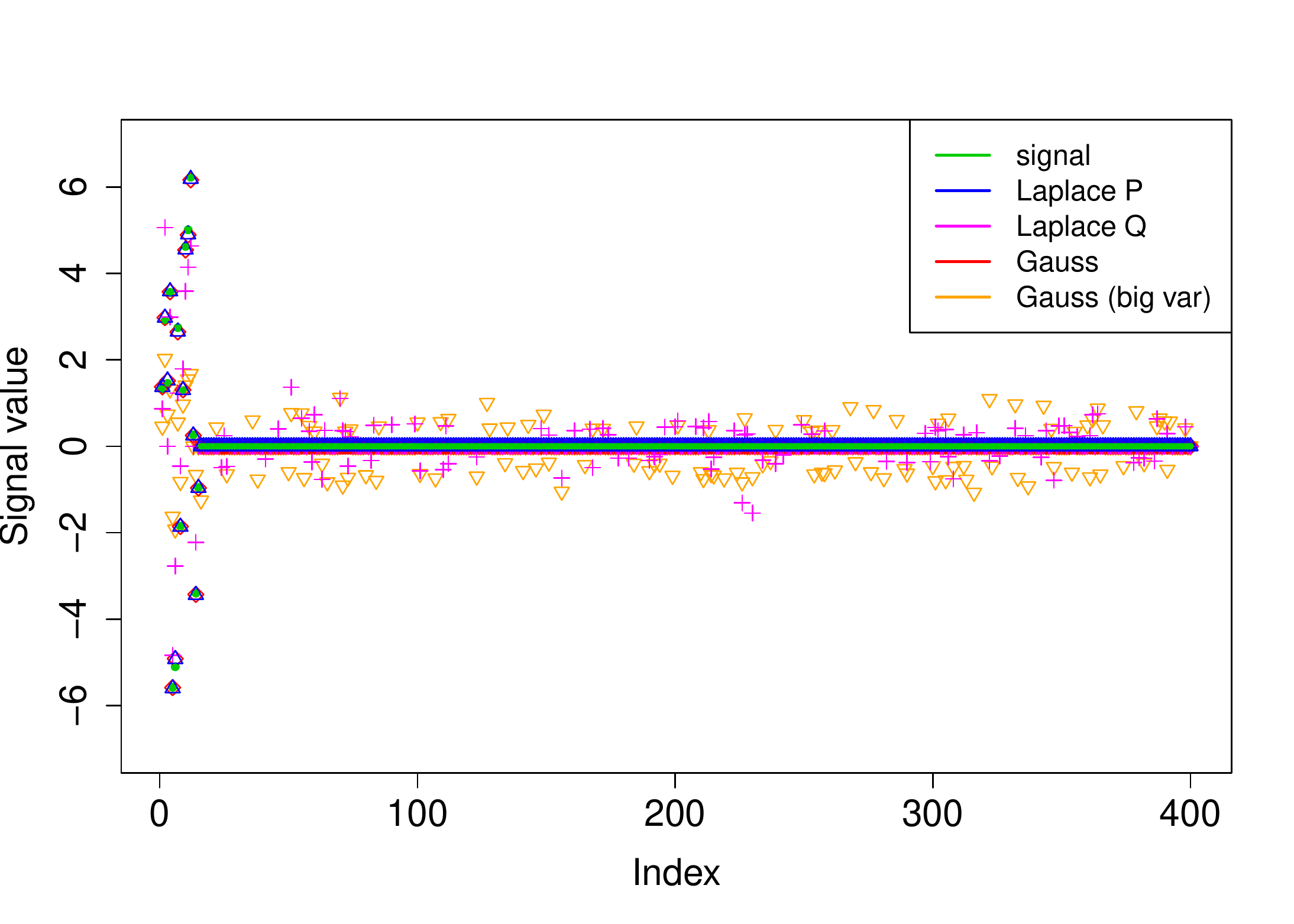}
  \end{minipage}
\caption{Linear regression with (i) identity design $X=I_n$ and (ii)-(iv) Gaussian design $X_{ij}\stackrel{iid}{\sim}N(0,\tau^2)$. We plot the underlying signal with non-zero components $\theta_i=A$ (green) and the posterior means of $\widetilde\Pi$ (blue), $\widetilde{Q}$ (purple), VB with Gaussian slabs (red) and VB with rescaled Gaussian slabs (orange). From left to right and top to bottom, we set the parameters $(n,p,s,A)$: (i) $(400,400,40,4\sqrt{\log n})$; (ii) $(100,200,20,2\log n)$; (iii) $(200,800,40,2\log n)$; (iv) $(100,400,15,U(-8,8))$. We set (ii) $\tau=1$; (iii) $\tau=0.1$; (iv) $\tau=0.5$.}
\label{fig: linreg1}
\end{figure}

\subsection{The effect of the hyper-parameter $\lambda$} \label{sec:sim:compare:lambda}

Theorem \ref{thm:recovery} states that for a wide range of hyper-parameter values $\lambda\in[\frac{\|X\|}{p}, \frac{C\|X\|\sqrt{\log p}}{s_0}]$, our VB algorithm has good asymptotic properties. However, the finite-sample performance depends on $\lambda$ as we now investigate. We ran our algorithm for different choices of $\lambda$, ranging from $1/20$ to $20$, on simulated data similar to that in the preceding subsections. 

We consider four different settings, each with Gaussian design with entries $X_{ij} \stackrel{iid}{\sim} N(0,\tau^2)$, non-zero signal components set to $\theta_i=A$ and noise variance $\varsigma^2=1$ assumed to be known. We take (i) $(n,p,s,\tau) = (200,300,15,0.5)$, $A=2\log n$; (ii) $(n,p,s,\tau) = (500,1000,50,1)$, $A=2\log n$; (iii) $(n,p,s,\tau) = (200,500,20,0.2)$, $A\stackrel{iid}{\sim}U(-10,10)$; and (iv) $(n,p,s,\tau) = (1000,2000,15,2)$, $A\stackrel{iid}{\sim}U(-8,8)$. In all cases, the non-zero signal components are located at the beginning of the signal. We ran each algorithm 200 times and report the results in Table \ref{table:compare:lambda}. The choice of $\lambda$ can indeed significantly influence the finite-sample behaviour of the algorithm (e.g. cases (ii) and (iii)), but not always ((i) and (iv)). There was not clear evidence to support a particular fixed choice of $\lambda$, since larger values sometimes performed better ((ii) and (iv)) and sometime worse ((i) and (iii)). This suggests using a data-driven choice of $\lambda$ may be helpful in practice. As expected, larger choices for $\lambda$, which cause more shrinkage, result in smaller FDR and TPR. The runtime across hyper-parameter choices were broadly comparable.

\begin{table}[tbp]
\centering

\begin{footnotesize}
\begin{tabular}{ |l|l|c|c|c|c|    }
\hline
%\multicolumn{6}{ |c| }{Team sheet} \\
%\hline

Metric& Method  &(i) & (ii)&(iii)&(iv) \\ \hline
\hline
\multirow{4}{*}{$\ell_2-\text{error}$} & $\lambda=1/20$ &\textbf{0.56 $\pm$ 0.11}  &35.81 $\pm$ 2.17 &2.49 $\pm$ 0.50 &0.09 $\pm$ 0.02 \\
&$\lambda=1/4$ & 0.57 $\pm$ 0.12  & 35.28 $\pm$ 2.32 & \textbf{2.34 $\pm$ 0.50}  &0.08 $\pm$ 0.02 \\
 &$\lambda=1$& 0.57 $\pm$ 0.11 &16.38 $\pm$ 14.63 &\textbf{2.38 $\pm$ 0.48}&0.08 $\pm$ 0.02  \\
 & $\lambda=4$ & 0.67 $\pm$ 0.12 &\textbf{0.34 $\pm$ 0.03} &3.56 $\pm$ 0.51&\textbf{0.07 $\pm$ 0.02} \\
 & $\lambda=20$ & 1.85 $\pm$ 0.22&0.47 $\pm$ 0.05 &11.94 $\pm$ 1.01&\textbf{0.07 $\pm$ 0.02}\\ \hline
\hline
\multirow{4}{*}{FDR} & $\lambda=1/20$ &\textbf{0.00 $\pm$ 0.00} &0.92 $\pm$ 0.00 & \textbf{0.00 $\pm$ 0.00} & \textbf{0.00 $\pm$ 0.00}\\
& $\lambda=1/4$ &\textbf{0.00 $\pm$ 0.01}&0.92 $\pm$ 0.00 &\textbf{0.00 $\pm$ 0.00}&\textbf{0.00 $\pm$ 0.00}\\
 & $\lambda=1$ &\textbf{0.00 $\pm$ 0.01} & 0.51 $\pm$ 0.45& \textbf{0.00 $\pm$ 0.01} & \textbf{0.00 $\pm$ 0.00}\\
 & $\lambda=4$ &\textbf{0.00 $\pm$ 0.00}&\textbf{0.00 $\pm$ 0.00}&\textbf{0.00 $\pm$ 0.01}&\textbf{0.00 $\pm$ 0.01} \\
 & $\lambda=20$ &\textbf{0.00 $\pm$ 0.00}&\textbf{0.00 $\pm$ 0.00} &\textbf{0.00 $\pm$ 0.02}& \textbf{0.00 $\pm$ 0.00}\\ \hline
\hline
\multirow{4}{*}{TPR} & $\lambda=1/20$ &\textbf{1.00 $\pm$ 0.00}&\textbf{1.00 $\pm$ 0.00} &0.91 $\pm$ 0.04   & 0.95 $\pm$ 0.03\\
& $\lambda=1/4$ &\textbf{1.00 $\pm$ 0.00}&\textbf{1.00 $\pm$ 0.01} &\textbf{0.92 $\pm$ 0.04}  & 0.96 $\pm$ 0.03\\
 & $\lambda=1$ &\textbf{1.00 $\pm$ 0.00}&\textbf{1.00 $\pm$ 0.01} &\textbf{0.92 $\pm$ 0.04} &0.97 $\pm$ 0.03 \\
 & $\lambda=4$ & \textbf{1.00 $\pm$ 0.00}&\textbf{1.00 $\pm$ 0.01}  &0.90 $\pm$ 0.04  &\textbf{0.98 $\pm$ 0.03} \\
 & $\lambda=20$ &\textbf{1.00 $\pm$ 0.00}&\textbf{1.00 $\pm$ 0.00}  &0.59 $\pm$ 0.07 & \textbf{0.98 $\pm$ 0.03}\\ \hline
\hline
\multirow{4}{*}{runtime (sec)} &$\lambda=1/20$ &\textbf{0.04 $\pm$ 0.01}  & \textbf{0.75 $\pm$ 0.11}& \textbf{0.08 $\pm$ 0.02} &\textbf{4.23 $\pm$ 0.46}\\
& $\lambda=1/4$ &\textbf{0.04 $\pm$ 0.01}  & 0.81 $\pm$ 0.18 & \textbf{0.08 $\pm$ 0.02}  &\textbf{4.23 $\pm$ 0.50}\\
 & $\lambda=1$ & \textbf{0.04 $\pm$ 0.01}  & 1.71 $\pm$   0.63& \textbf{0.08 $\pm$ 0.02}&  \textbf{4.23 $\pm$ 0.46} \\
 & $\lambda=4$  &\textbf{0.04 $\pm$ 0.02} & 1.26 $\pm$ 0.15 & \textbf{0.08 $\pm$ 0.02}   & 4.25 $\pm$ 0.50 \\
 & $\lambda=20$  &\textbf{0.04 $\pm$ 0.01}  & \textbf{0.75 $\pm$ 0.06} & \textbf{0.08 $\pm$ 0.01}  & 4.29 $\pm$ 0.65 \\ \hline
\hline
\end{tabular}
\caption{Performance of sparsevb for different hyper-parameter values $\lambda$. We take Gaussian design $X_{ij}\stackrel{iid}{\sim}N(0,\tau^2)$, place the non-zero signal coefficients $\theta_{0,i}=A$ at the beginning of the signal, and set the parameters $(n,p,s,\tau,A)$ equal to (i) $(200, 300, 15, 0.5, 2\log n)$; (ii) $(500, 1000, 50, 1, 2\log n)$; (iii) $(200, 500, 20, 0.2, U(-10,10))$; (iv) $(1000, 2000, 15, 2, U(-8,8))$.}
\label{table:compare:lambda}
\end{footnotesize}
\end{table}

\subsection{Noise misspecification}\label{sec:sim:misspecification}

We investigate the robustness of the Bayesian model selection methods to misspecification of the noise distribution in practice. Note that our theoretical results are also robust to some misspecification, see Remark \ref{rem:misspecification} in Section \ref{sec: proof} below. We consider Gaussian design $X_{ij}\stackrel{iid}{\sim}N(0,2)$, set the model parameters $n=200$, $p=400$, $s=20$, and take non-zero signal coefficients $\theta_i\stackrel{iid}{\sim}U(-10,10)$ located in the beginning of $\theta$. We compare the correctly-specified Gaussian noise case (i) $Z_i\stackrel{iid}{\sim}N(0,1)$ in model \eqref{eq:model} with the misspecified noise cases: (ii) Laplace noise $Z_i\stackrel{iid}{\sim}\text{Lap}(0,1)$;  (iii) uniform noise $Z_i \stackrel{iid}{\sim} U(-2,2)$; (iv) Student noise with 3 degrees of freedom $Z_i \stackrel{iid}{\sim} t_3$. We apply the same parametrizations of the methods as in Section \ref{sec:unknown_var}. We ran each experiments 200 times and collect the results in Table \ref{table:misspecification}. Our method (sparsevb) gave similar results to varbvs, ebreg and EMVS, while the SSLASSO performed slightly worse. The noise distribution does not seem to have a major effect on the results, hence these algorithms seem robust to noise misspecification. It is worthwhile to further investigate this phenomenon both empirically and analytically.

\begin{table}[tbp]
\centering

\begin{footnotesize}
\begin{tabular}{ |l|l|c|c|c|c|    }
\hline
%\multicolumn{6}{ |c| }{Team sheet} \\
%\hline

Metric& Method&(i) $N(0,1)$ & (ii) $\text{Lap}(0,1)$ &(iii) $U(-2,2)$ &(iv) Student $t_3$ \\ \hline
\hline
\multirow{3}{*}{$\ell_2-\text{error}$} & sparsevb& 0.18 $\pm$ 0.05 &\textbf{0.24 $\pm$ 0.04}  &\textbf{0.21 $\pm$ 0.03} &0.30 $\pm$ 0.06  \\
&varbvs& \textbf{0.17 $\pm$ 0.03} &\textbf{0.24 $\pm$ 0.04}  & \textbf{0.21 $\pm$ 0.03} &0.30 $\pm$ 0.06  \\
 & EMVS & 0.59 $\pm$ 0.03 &1.03  $\pm$ 0.14 & 0.89 $\pm$ 0.16 & 1.13 $\pm$ 0.43 \\
 & SSLASSO  &5.99 $\pm$ 0.98 &4.07 $\pm$ 1.02  &4.88 $\pm$ 0.62 &4.87 $\pm$ 0.78\\
 & ebreg & 0.26 $\pm$ 0.05 &0.26 $\pm$ 0.07 & 0.23 $\pm$ 0.05 &\textbf{0.23 $\pm$ 0.05} \\
 \hline
\hline
\multirow{3}{*}{FDR} & sparsevb& \textbf{0.00 $\pm$ 0.00}& \textbf{0.00 $\pm$ 0.00}  &\textbf{0.00 $\pm$ 0.00} &\textbf{0.00 $\pm$ 0.00}\\
& varbvs& \textbf{0.00 $\pm$ 0.00} &\textbf{0.00 $\pm$ 0.00} &\textbf{0.00 $\pm$ 0.00} &\textbf{0.00 $\pm$ 0.01}\\
 & EMVS&\textbf{0.00 $\pm$ 0.00} &\textbf{0.00 $\pm$ 0.00}&\textbf{0.00 $\pm$ 0.00}&\textbf{0.00 $\pm$ 0.01}\\
 & SSLASSO & \textbf{0.00 $\pm$ 0.00}&\textbf{0.00 $\pm$ 0.00}&\textbf{0.00 $\pm$ 0.00}&\textbf{0.00 $\pm$ 0.00}\\
& ebreg & 0.01 $\pm$ 0.02 &0.01 $\pm$ 0.05 &0.01 $\pm$ 0.05 &0.01 $\pm$ 0.03  \\
\hline
\hline
\multirow{3}{*}{TPR} & sparsevb &\textbf{1.00 $\pm$ 0.01} &\textbf{1.00 $\pm$ 0.00}&\textbf{0.95 $\pm$ 0.00} &\textbf{0.90 $\pm$ 0.01}   \\
& varbvs &\textbf{1.00 $\pm$ 0.00} &\textbf{1.00 $\pm$ 0.00}&\textbf{0.95 $\pm$ 0.01}  &\textbf{0.90 $\pm$ 0.01}  \\
 & EMVS& 0.95 $\pm$ 0.02 &0.92 $\pm$ 0.02&0.89 $\pm$ 0.02 &0.81 $\pm$ 0.04  \\
 & SSLASSO & 0.67 $\pm$ 0.04 & 0.72 $\pm$ 0.05&0.64 $\pm$ 0.02 &0.64 $\pm$ 0.04   \\
& ebreg &\textbf{1.00 $\pm$ 0.01} &\textbf{1.00 $\pm$ 0.00} &\textbf{0.95 $\pm$ 0.01} &\textbf{0.90 $\pm$ 0.01}  \\
 \hline
\hline
\multirow{3}{*}{runtime (sec)} & sparsevb& 0.22 $\pm$ 0.03 &\textbf{0.24 $\pm$ 0.06} & 0.24  $\pm$ 0.06&\textbf{0.26 $\pm$ 0.08}\\
& varbvs& 0.32 $\pm$ 0.05&0.32 $\pm$ 0.05  & 0.35 $\pm$ 0.08 &0.35 $\pm$ 0.09 \\
 & EMVS& 1.24 $\pm$ 0.15 &1.19 $\pm$ 0.24& 1.26 $\pm$ 0.27 & 1.31 $\pm$ 0.34 \\
 & SSLASSO  &\textbf{0.16 $\pm$ 0.03} &0.28 $\pm$ 0.04 &\textbf{0.22 $\pm$ 0.05} & 0.28 $\pm$ 0.07  \\
& ebreg & 24.37 $\pm$ 7.10 &24.89 $\pm$ 3.78 &127.72 $\pm$ 4.51 &28.19 $\pm$4.51 \\
\hline
\hline
\end{tabular}
\caption{Noise misspecification: we compare the robustness of Bayesian model selection methods under misspecified noise. We take Gaussian design $X_{ij}\stackrel{iid}{\sim}N(0,2)$, set the model parameters $n=200$, $p=400$, $s=20$, and take non-zero coefficients $\theta_i\stackrel{iid}{\sim}U(-10,10)$ located in the beginning of the signal. We ran each experiment 200 times and report the means and standard deviations.}
\label{table:misspecification}
\end{footnotesize}
\end{table}

\setcounter{lemma}{0}
\setcounter{theorem}{0}
\setcounter{definition}{0}
\setcounter{remark}{0}
\setcounter{equation}{0}
     \renewcommand{\thelemma}{\Alph{section}.\arabic{lemma}}
    \renewcommand{\thetheorem}{\Alph{section}.\arabic{theorem}}
   \renewcommand{\thedefinition}{\Alph{section}.\arabic{definition}}
   \renewcommand{\theremark}{\Alph{section}.\arabic{remark}}

\subsection{Bayesian variable selection methods under correlated inputs}\label{sec:correlated_design}

We lastly consider the common situation of correlated input variables. We take each row $X_{i\cdot} \stackrel{iid}{\sim} N_p(0,\Sigma)$ with $\Sigma_{jk}=\rho$ for $j\neq k$ and $\Sigma_{jj}=1$, giving standard normal predictors with non-zero correlation $\rho$. We take (i) $(n,p,s,\varsigma)=(100,400,10,0.2)$, correlation $\rho=0.3$ and non-zero coefficients $\theta_i\stackrel{iid}{\sim} U(-3,3)$ at the beginning of the signal; (ii) the same setting as in (i), but with higher correlation $\rho=0.7$; (iii) $(n,p,s,\varsigma)=(200,800,20,5)$, correlation $\rho = 0.3$ and non-zero coefficients $\theta_i=2\log n$ at the end of the signal; (iv) the same setting as in (iii), but with higher correlation $\rho=0.7$. We apply the same parametrizations of the methods as in Section \ref{sec:unknown_var}. The results are summarized in Table \ref{table:error:compare:methods:correlation}.

One might expect that mean-field VB methods should not perform so well under correlated inputs due to their factorizable structure. This was not the case in our simulations, where the VB methods perform competitively with the other methods, often providing the best results  (except perhaps in (iv), where varbvs sometimes sometimes gave large $\ell_2$ error). The correlated design also does not seem to substantially influence the run time.

While our simulations are certainly not extensive, they suggest that mean-field VB can perhaps still be effective in certain correlated input settings and understanding the exact effect of correlation on VB seems to be a subtle question. It is currently not well understood how VB, or indeed even the true posterior, behaves in general correlated design settings. This important and practically very relevant setting requires further investigation, both theoretically and empirically.

\begin{table}[tbp]
\centering

\begin{footnotesize}
\begin{tabular}{ |l|l|c|c|c|c|    }
\hline
%\multicolumn{6}{ |c| }{Team sheet} \\
%\hline

Metric& Method &(i) & (ii)&(iii)&(iv) \\ \hline
\hline
\multirow{3}{*}{$\ell_2-\text{error}$} &sparsevb &\textbf{0.12 $\pm$ 0.06}  & 0.89 $\pm$ 1.40 &\textbf{1.97 $\pm$ 0.37} &\textbf{4.85  $\pm$ 1.29} \\
&varbvs &0.13 $\pm$ 0.06  & \textbf{0.30 $\pm$ 0.10}  & 2.10 $\pm$ 0.43 & 27.18 $\pm$ 23.59 \\
 & EMVS& 4.80 $\pm$ 0.21& 5.29 $\pm$ 0.26 & 4.04 $\pm$ 0.30 & 7.04 $\pm$ 0.98  \\
 & SSLASSO &1.62 $\pm$ 0.35 & 0.97 $\pm$ 0.36 & 56.70 $\pm$ 7.78 & 79.17 $\pm$ 4.95 \\
 & ebreg & 0.34 $\pm$ 0.06 & 0.56 $\pm$ 0.14 & 5.41 $\pm$ 0.67 & 6.41 $\pm$ 1.21 \\
 \hline
\hline
\multirow{3}{*}{FDR} & sparsevb &\textbf{0.00 $\pm$ 0.00}   &0.18 $\pm$ 0.34 & \textbf{0.00 $\pm$ 0.01} &\textbf{0.00 $\pm$ 0.00}\\
& varbvs& \textbf{0.00 $\pm$ 0.00}  &\textbf{0.00 $\pm$ 0.00} & 0.01 $\pm$ 0.02& 0.31 $\pm$ 0.26\\
 & EMVS& \textbf{0.00 $\pm$ 0.00} &\textbf{0.00 $\pm$ 0.00} &\textbf{0.00 $\pm$ 0.01} & 0.14 $\pm$ 0.08\\
 & SSLASSO &\textbf{0.00 $\pm$ 0.00}&\textbf{0.00 $\pm$ 0.00}& 0.18 $\pm$ 0.16& 0.41 $\pm$ 0.19\\
 & ebreg &\textbf{0.00 $\pm$ 0.00} & \textbf{0.00 $\pm$ 0.00} & 0.43 $\pm$ 0.05 & 0.28 $\pm$ 0.08 \\
\hline
\hline
\multirow{3}{*}{TPR} & sparsevb &\textbf{0.96 $\pm$ 0.05} &0.95 $\pm$ 0.10 & \textbf{1.00 $\pm$ 0.00}  &  \textbf{1.00 $\pm$ 0.00} \\
& varbvs & 0.95 $\pm$ 0.05 &\textbf{1.00 $\pm$ 0.00} & \textbf{1.00 $\pm$ 0.00} &0.69 $\pm$ 0.32\\
 & EMVS& 0.01 $\pm$ 0.03 &0.02 $\pm$ 0.04 &\textbf{1.00 $\pm$ 0.00} &\textbf{1.00 $\pm$ 0.00} \\
 & SSLASSO & 0.48 $\pm$ 0.04  &0.81 $\pm$ 0.08  & 0.34 $\pm$ 0.10  &0.18 $\pm$ 0.05 \\
 & ebreg &  0.90 $\pm$ 0.01 & 0.96 $\pm$ 0.05  & \textbf{1.00 $\pm$ 0.00} &  \textbf{1.00 $\pm$ 0.00} \\
 \hline
\hline
\multirow{3}{*}{runtime (sec)} & sparsevb &0.22 $\pm$ 0.06  & 0.30 $\pm$ 0.06& 0.91 $\pm$ 0.11 &1.28 $\pm$ 0.17\\
& varbvs&0.51 $\pm$ 0.19  & 0.80 $\pm$ 0.52 & 12.31 $\pm$ 3.34 &30.54 $\pm$ 6.38\\
 & EMVS&\textbf{0.14 $\pm$ 0.07} & \textbf{0.15 $\pm$ 0.06} & 0.90 $\pm$ 0.18& 1.06 $\pm$ 0.26 \\
 & SSLASSO  &0.33 $\pm$ 0.07 & 0.38 $\pm$ 0.19 & \textbf{0.16 $\pm$ 0.03}  &  \textbf{0.16 $\pm$ 0.02} \\
 & ebreg & 13.90 $\pm$ 1.62 & 15.06 $\pm$ 2.96  & 64.14 $\pm$ 6.36 &59.98 $\pm$ 4.41 \\

\hline
\hline
\end{tabular}
\caption{Linear regression with correlated Gaussian design $X_{i\cdot} \stackrel{iid}{\sim} N_p(0,\Sigma)$, with correlation $\Sigma_{jk}=\rho$ for $j\neq k$ and $\Sigma_{jj}=1$. The noise variance $\varsigma^2$ is unknown and the non-zero signal coefficients equal $\theta_i=A$. We take the parameters $(n,p,s,A,\rho,\varsigma)$ equal to
(i) $(100,400,10,\stackrel{iid}{\sim} U(-3,3),0.3,0.2)$ (non-zero coefficients at the beginning);
(ii) $(100,400,10,\stackrel{iid}{\sim} U(-3,3),0.7,0.2)$ (at the beginning);
(iii) $(200,800,20,2\log n,0.3,5)$ (at the end); 
(iv) $(200,800,20,2\log n,0.7,5)$ (at the end).
We compare the means and standard deviations over 100 runs.}
\label{table:error:compare:methods:correlation}
\end{footnotesize}
\end{table}

\section{Proofs}\label{sec: proof}

\subsection{Full oracle results}

The proofs of the full oracle results in Theorems \ref{thm:oracle_full_recovery} and \ref{thm:oracle_full_dim} below rely on Theorem \ref{thm: general:VB}, which allows one to exploit exponential probability bounds for the posterior to control the corresponding probability under the variational approximation. To prove our results, it therefore suffices to show that on a suitable event, one can (a) control the KL divergence between the variational approximation and the true posterior and (b) establish the appropriate posterior tail inequality \eqref{cond:exp_decay}. Part (a) is dealt with in Section \ref{sec:KL} and (b) in Section \ref{sec:contraction} below. Define the events
\begin{align}\label{eq:T_event}
\mathcal{T}_0 = \{\|X^T(Y-X\theta_0)\|_\infty \leq 2 \|X\|\sqrt{\log p}\}
\end{align}
and
\begin{equation}\label{T1}
\begin{split}
\mathcal{T}_{1} = \mathcal{T}_{1}(\Gamma,\eps,\kappa) & = \mathcal{T}_0 \cap \Big\{ \Pi\big(\theta: |S_\theta|>\Gamma \big|Y\big) \leq 1/4 \Big\}  \cap \Big\{ \Pi\big(\theta:\, \|\theta-\theta_0\|_2 > \eps |Y\big) \leq e^{-\kappa} \Big\},
\end{split}
\end{equation}
for $\Gamma ,\eps,\kappa> 0$. The middle event in $\mathcal{T}_1$ says that the posterior puts most of its mass on models of dimension at most $\Gamma$; the number $1/4$ is unimportant and any number less than $1/2$ suffices. The third event says the posterior places all but exponentially small probability on an $\ell_2$-ball of radius $\eps$ about the truth and is used for a localization argument when bounding the KL divergence. The proof uses an iterative structure, using successive posterior localizations to eventually bound the KL divergence in Section \ref{sec:KL}. This idea is a useful technique from Bayesian nonparametrics, see e.g. \cite{nicklray2019}.

For parameters $\theta_0,\theta_*\in \R^p$, set $S_* = S_{\theta_*}$ and $s_* = |S_*|$ and define
\begin{equation}\label{Delta}
\Delta_* = (1+\tfrac{16}{\phi(S_*)^2} \tfrac{\lambda}{\bar{\lambda}}) s_*\log p  + \|X(\theta_0 - \theta_*)\|_2^2.
\end{equation}
This quantity appears in the posterior exponential probabilities, which take the form $e^{-c\Delta_*}$. We require the following parameter choices for the event $\mathcal{T}_1$ in \eqref{T1}:
\begin{equation}\label{parameters}
\begin{split}
\Gamma & =  \Gamma_{\theta_0,\theta_*} = s_* + \frac{12}{A_4} \left(1 +\frac{16}{\phi(S_*)^2}\frac{\lambda}{\bar{\lambda}} \right)s_*  + \frac{12\|X(\theta_0-\theta_*)\|_2^2}{A_4 \log p} = s_* + \frac{12\Delta_*}{A_4\log p},\\
\eps & = \eps_{\theta_0,\theta_*} = \tfrac{ML_0^{1/2} }{\|X\| \widetilde{\psi}_{L_0+2}(S_0)^2} \left[ \tfrac{\sqrt{s_*\log p}}{\phi(S_*)} +\|X(\theta_0-\theta_*)\|_2 \right], \\
\kappa & = \kappa_{\theta_0,\theta_*} = (\Gamma_{\theta_0,\theta_*} + 1)\log p,\\
L_0 & = \max(3 + 12/A_4,2+A_4/2)
\end{split}
\end{equation}
for some $M>0$ large enough depending only on $A_1,A_3,A_4$.

\begin{lemma}[]\label{lem: prob}
(i) The event $\mathcal{T}_0$ defined in \eqref{eq:T_event} satisfies
$$\inf_{\theta_0 \in \R^p} P_{\theta_0}(\mathcal{T}_0) \geq 1- 2/p.$$
(ii) Suppose the prior satisfies \eqref{eq:prior_cond} and \eqref{prior_lambda}. For $\theta_0 \in \R^p \backslash \{0\}$, let $\theta_*\in \R^p$ be any vector satisfying $1 \leq s_* = |S_{\theta_*}| \leq |S_{\theta_0}| = s_0$,
$$\frac{s_*}{\phi(S_*)^2} \leq \frac{s_0}{\phi(S_0)^2} \quad \text{ and } \quad \|X(\theta_0 - \theta_*)\|_2^2 \leq (s_0 - s_*) \log p.$$
Then the event $\mathcal{T}_1$ given in \eqref{T1} with parameters $\Gamma,\eps,\kappa$ chosen according to \eqref{parameters} satisfies
$$P_{\theta_0} (\mathcal{T}_1) \to 1$$
uniformly over all $\theta_0$ and $\theta_*$ as above.
\end{lemma}

\begin{proof}
(i) Under $P_{\theta_0}$, $X^T(Y-X\theta_0) = X^T Z \sim N_p(0,X^TX)$. Since $(X^TZ)_i \sim N(0,(X^TX)_{ii})$ and $(X^TX)_{ii} \leq \|X\|^2$ for all $1 \leq i \leq p$, a union bound and the standard Gaussian tail inequality give
$$P_{\theta_0} (\mathcal{T}_0^c) = P(\|X^TZ\|_\infty \geq 2\|X\|\sqrt{\log p}) \leq \sum_{i=1}^p P( |N(0,1)| \geq 2\sqrt{\log p}) \leq p \frac{2}{\sqrt{2\pi}} e^{-2\log p}.$$
(ii) Applying Markov's inequality and Lemma \ref{lem: support} below with $M=3$ gives
\begin{align*}
& P_{\theta_0} \left( \big\{ \Pi\big(\theta: |S_\theta|> \Gamma_{\theta_0,\theta_*} |Y\big) > 1/4 \big\} \cap \mathcal{T}_0 \right) \\
& \quad  \leq 4 E_{\theta_0}\Pi\big(\theta: |S_\theta|> \Gamma_{\theta_0,\theta_*} |Y\big) 1_{\mathcal{T}_0} \\
& \quad  \leq C(A_2,A_4) \exp \left(  -\left( 1+\tfrac{16}{\phi(S_*)^2}\tfrac{\lambda}{\bar{\lambda}} \right) s_*\log p \right)\\
& \quad \leq C(A_2,A_4) e^{-s_* \log p} \leq C(A_2,A_4) e^{-\log p}.
\end{align*}
Since the right-hand side does not depend on $\theta_0$ or $\theta_*$, the probability tends to zero uniformly as required.

Under the assumptions on $\theta_*$,
\begin{equation}\label{exponent_comparison}
\begin{split}
(1+\tfrac{16}{\phi(S_*)^2}\tfrac{\lambda}{\bar{\lambda}}) s_*\log p + \|X(\theta_0-\theta_*)\|_2^2 & \leq s_* \log p + \tfrac{16}{\phi(S_0)^2}\tfrac{\lambda}{\bar{\lambda}} s_0\log p + (s_0 - s_*) \log p \\
& = (1+\tfrac{16}{\phi(S_0)^2}\tfrac{\lambda}{\bar{\lambda}}) s_0\log p.
\end{split}
\end{equation}
Therefore, applying Lemma \ref{lem: contraction:spike:slab} with $L \geq 1$ yields
\begin{align*}
& E_{\theta_0} \Pi \left( \theta : \|\theta-\theta_0\|_2 > \frac{ML^{1/2} }{\|X\| \overline{\psi}_{L+2}(S_0)^2} \left[ \frac{\sqrt{s_*\log p}}{\phi(S_*)} +\|X(\theta_0-\theta_*)\|_2 \right] \Big| Y \right)1_{\mathcal{T}_0},\\
& \leq C \exp \left(- \left[ L\wedge \tfrac{4(L+2)}{A_4} \right] \left[ (1+\tfrac{16}{\phi(S_*)^2}\tfrac{\lambda}{\bar{\lambda}}) s_*\log p + \|X(\theta_0-\theta_*)\|_2^2 \right] \right).
\end{align*}
Using Markov's inequality and the last display with $L = L_0 = \max(3 + 12/A_4,2+A_4/2)$,
\begin{align*}
& P_{\theta_0} \left( \{ \Pi\big(\theta: \|\theta-\theta_0\|_2> \eps |Y) > e^{-\kappa} \big\} \cap \mathcal{T}_0 \right) \\
& \leq  e^{\kappa}  E_{\theta_0}\Pi\big(\theta: \|\theta-\theta_0\|_2> \eps |Y\big) 1_{\mathcal{T}_0} \\
& \leq C \exp \left(- \left[ L\wedge \tfrac{4(L+2)}{A_4} -\tfrac{12}{A_4} \right] \left[ (1+\tfrac{16}{\phi(S_*)^2}\tfrac{\lambda}{\bar{\lambda}}) s_*\log p + \|X(\theta_0-\theta_*)\|_2^2 \right]  +(s_*+1) \log p \right) \\
& \leq C e^{-s_* \log p} \leq Ce^{-\log p}.
\end{align*}
Since the right-hand side again does not depend on $\theta_0$ or $\theta_*$, the probability tends to zero uniformly as required.
\end{proof}

\begin{theorem}[Full oracle recovery]\label{thm:oracle_full_recovery}
Suppose the model selection prior \eqref{eq:prior} satisfies \eqref{eq:prior_cond} and \eqref{prior_lambda}. For $\theta_0 \in \R^p \backslash \{0\}$, let $\theta_*\in \R^p$ be any vector satisfying $1 \leq s_* = |S_{\theta_*}| \leq |S_{\theta_0}| = s_0$ and $\|X(\theta_0 - \theta_*)\|_2^2 \leq (s_0 - s_*) \log p.$
Then the variational Bayes posterior $\widetilde{\Pi}$ satisfies, uniformly over all $\theta_0$ and $\theta_*$ as above,
\begin{align*}
& E_{\theta_0} \widetilde{\Pi} \left( \theta: \|X(\theta-\theta_0)\|_2\geq \frac{M\rho_n^{1/2} }{\overline{\psi}_{\rho_n}(S_0)} \left[ \frac{\sqrt{s_*\log p}}{\phi(S_*)} +\|X(\theta_0-\theta_*)\|_2 \right] \right) \\
& \lesssim \frac{1}{\rho_n} \left\{ 1+ \frac{\log (1/\widetilde{\phi}(\Gamma))}{\log p} + \frac{\lambda s_0 }{\|X\|\widetilde{\psi}_{L_0+2}(S_0)^2 \phi(S_0) \widetilde{\phi}(\Gamma)^2 \sqrt{\log p}}  \right\} + o(1)
\end{align*}
for any $\rho_n >2$, where $\Gamma, L_0$ are given in \eqref{parameters}. Moreover, both
$$E_{\theta_0} \widetilde{\Pi} \left( \theta : \|\theta-\theta_0\|_1 > \| \theta_0 - \theta_*\|_1 + \frac{M\rho_n }{\overline{\psi}_{\rho_n}(S_0)^2} \left[ \frac{s_* \sqrt{\log p}}{\|X\| \phi(S_*)^2} +\frac{\|X(\theta_0-\theta_*)\|_2^2}{\|X\| \sqrt{\log p}} \right] \right),$$
$$E_{\theta_0} \widetilde{\Pi} \left( \theta : \|\theta-\theta_0\|_2 > \frac{M\rho_n^{1/2} }{\|X\| \widetilde{\psi}_{\rho_n}(S_0)^2} \left[ \frac{\sqrt{s_*\log p}}{\phi(S_*)} +\|X(\theta_0-\theta_*)\|_2 \right]  \right),$$
satisfy the same inequality. Furthermore, the exact same inequalities hold for the variational Bayes posteriors $\widetilde{Q}$ and $\hat{Q}$.
\end{theorem}

\begin{proof}
Suppose first that $s_*/\phi(S_*)^2 \leq s_0/\phi(S_0)^2$. Let $\mathcal{T}_1$ denote the event in \eqref{T1} with parameters \eqref{parameters}, which by Lemma \ref{lem: prob}(ii) satisfies $P_{\theta_0}(\mathcal{T}_1) \to 1$ uniformly over all $\theta_0,\theta_*$ in the theorem hypothesis. Set
$$\Theta_n = \left\{ \theta: \|X(\theta-\theta_0)\|_2\geq \frac{M\rho_n^{1/2} }{\overline{\psi}_{\rho_n}(S_0)} \left[ \frac{\sqrt{s_*\log p}}{\phi(S_*)} +\|X(\theta_0-\theta_*)\|_2 \right] \right\}$$
and note $E_{\theta_0} \widetilde{\Pi}(\Theta_n)  \leq E_{\theta_0} \widetilde{\Pi} (\Theta_n) 1_{\mathcal{T}_1} + o(1).$ We now apply Theorem \ref{thm: general:VB} with this choice of $\Theta_n$ on the event $\mathcal{T}_1$. For $\Delta_*$ defined in \eqref{Delta}, it holds that $\Delta_* \leq (1+\tfrac{16}{\phi(S_0)^2}\tfrac{\lambda}{\bar{\lambda}}) s_0\log p$ by \eqref{exponent_comparison}. Using Lemma \ref{lem: contraction:spike:slab} below with $L+2= \rho_n$ thus gives
\begin{align*}
& E_{\theta_0} \Pi (\Theta_n|Y)1_{\mathcal{T}_0} \leq Ce^{-c\rho_n \Delta_*},
\end{align*}
for $p$ large enough depending on $A_1,A_3,A_4$, and where $C,c>0$ also depend only on the prior parameters. Since $\mathcal{T}_1 \subset \mathcal{T}_0$ by \eqref{T1}, condition \eqref{cond:exp_decay} is satisfied on $\mathcal{T}_1$ with $\delta_n = c\rho_n \Delta_*$. Applying Theorem \ref{thm: general:VB} gives
\begin{align*}
E_{\theta_0} \widetilde{\Pi} (\Theta_n) 1_{\mathcal{T}_1}  & \leq \tfrac{2}{c\rho_n \Delta_*} \KL (\widetilde{\Pi}\| \Pi(\cdot|Y))1_{\mathcal{T}_1} + o(1).
\end{align*}
Note that the parameters \eqref{parameters} satisfy $\Gamma \log p \lesssim \Delta_*$ and $\eps \lesssim \frac{\sqrt{s_0 \log p}}{\|X\|\widetilde{\psi}_{L_0+2}(S_0)^2\phi(S_0)}$. Using this and Lemma \ref{lem:KL_MF} below,
\begin{align*}
\tfrac{2}{c\rho_n \Delta_*} \KL (\widetilde{\Pi}\| \Pi(\cdot|Y))1_{\mathcal{T}_1} & \lesssim \frac{1}{\rho_n} \left\{1+ \frac{\log (1/\widetilde{\phi}(\Gamma))}{\log p} + \frac{\lambda s_0 }{\|X\|\widetilde{\psi}_{L_0+2}(S_0)^2 \phi(S_0) \widetilde{\phi}(\Gamma)^2 \sqrt{\log p}}  \right\} + o(1) 
\end{align*}
as required.

If $s_*/\phi(S_*)^2 > s_0/\phi(S_0)^2,$ then $\tfrac{\sqrt{s_*\log p}}{\phi(S_*)} +\|X(\theta_0-\theta_*)\|_2 > \tfrac{\sqrt{s_0\log p}}{\phi(S_0)}.$ The desired inequality then immediately follows from the stronger inequality with $\theta_* = \theta_0$ just established above. The results for $\ell_1$ and $\ell_2$ loss follow exactly as above by using the respective inequalities for the $\ell_1$ and $\ell_2$ oracle contraction rates in Lemma \ref{lem: contraction:spike:slab} to establish \eqref{cond:exp_decay}.

Similarly, the results for the variational Bayes posteriors $\hat{Q}$ and $\widetilde{Q}$ based on the mean-field variational families \eqref{def: variational:family2} and \eqref{def: variational:family3} follow identically upon using Lemmas \ref{lem:KL_non_diag} and \ref{lem:KL_Q_MF} instead of Lemma \ref{lem:KL_MF} to control the Kullback-Leibler divergence.
\end{proof}

\begin{theorem}[Full oracle dimension]\label{thm:oracle_full_dim}
Suppose the model selection prior \eqref{eq:prior} satisfies \eqref{eq:prior_cond} and \eqref{prior_lambda}. For $\theta_0 \in \R^p \backslash \{0\}$, let $\theta_*\in \R^p$ be any vector satisfying $1 \leq s_* = |S_{\theta_*}| \leq |S_{\theta_0}| = s_0$ and $\|X(\theta_0 - \theta_*)\|_2^2 \leq (s_0 - s_*) \log p.$
Then the variational Bayes posterior $\widetilde{\Pi}$ satisfies, uniformly over all $\theta_0$ and $\theta_*$ as above,
\begin{align*}
& E_{\theta_0} \widetilde{\Pi} \left( \theta: |S_\theta| \geq |S_*| + \tfrac{4(\rho_n+2)}{A_4} \left[  \left(1 +\tfrac{16}{\phi(S_*)^2}\tfrac{\lambda}{\bar{\lambda}} \right)|S_*|  + \tfrac{\|X(\theta_0-\theta_*)\|_2^2}{\log p} \right] \right)  \\
& \lesssim \frac{1}{\rho_n} \left\{ 1+\frac{\log (1/\widetilde{\phi}(\Gamma))}{\log p} + \frac{\lambda s_0 }{\|X\|\widetilde{\psi}_{L_0+2}(S_0)^2 \phi(S_0) \widetilde{\phi}(\Gamma)^2 \sqrt{\log p}}  \right\} + o(1)
\end{align*}
for any $\rho_n >0 $, where $\Gamma, L_0$ are given in \eqref{parameters}. Furthermore, the exact same inequality holds for the variational Bayes posteriors $\widetilde{Q}$ and $\hat{Q}$.
\end{theorem}

\begin{proof}
The proof follows similarly to that of Theorem \ref{thm:oracle_full_recovery} by applying Theorem \ref{thm: general:VB} with
$$ \Theta_n =  \left\{ \theta: |S_\theta| \geq |S_*| + \tfrac{4(\rho_n+2)}{A_4} \left[  \left(1 +\tfrac{16}{\phi(S_*)^2}\tfrac{\lambda}{\bar{\lambda}} \right)|S_*|  + \tfrac{\|X(\theta_0-\theta_*)\|_2^2}{\log p} \right] \right\},$$
again taking the event $A=\mathcal{T}_1$ and using Lemma \ref{lem: support} with $M=\rho_n + 2$ instead of Lemma \ref{lem: contraction:spike:slab} to verify \eqref{cond:exp_decay}.
\end{proof}

\begin{remark}[Misspecification of the error distribution]\label{rem:misspecification}
The Gaussian error distribution is assumed in model \eqref{eq:model} for concreteness and can be relaxed. For recovery and dimension control (Theorems \ref{thm:recovery} and \ref{thm:dim}), inspection of the contraction rate proofs in \cite{castillo:2015} and the KL bounds in Section \ref{sec:KL} show that it suffices that there exists a constant $C>0$ such that
$$P_{\theta_0}(\|X^T(Y-X\theta_0)\|_\infty > C \|X\| \sqrt{\log p}) \to 0,$$
which holds for much more general noise distributions. This condition is commonly imposed when studying the LASSO, see e.g. \cite{buhlmann:2011}. For the full oracle bounds, we further need that Lemma 3 of \cite{castillo:2015}, which concerns a change of measure, holds. This indeed holds under a wider range of noise distributions, see Remark 1 of \cite{castillo:2015}. The results for VB in this paper are thus robust under noise misspecification as for the true posterior \cite{castillo:2015}, see also Section \ref{sec:sim:misspecification} for an empirical study of noise misspecification for our method.
\end{remark}

\subsection{Kullback-Leibler divergences between variational classes and the posterior}\label{sec:KL}

We now show that on the event $\mathcal{T}_1$ in \eqref{T1}, we can bound the (minimized) Kullback-Leibler divergences between the posterior and the approximating variational classes. In particular, we need oracle-type bounds on the KL divergence to obtain our oracle results. This is the major technical difficulty in establishing our result. We first consider the family $\mathcal{Q}$ of distributions \eqref{def: variational:family2}, which consists of products of non-diagonal multivariate normal distributions with Dirac delta distributions for a single fixed support set $S$.%We recall that this family selects a single fixed support set $S$, on which it fits a non-diagonal multivariate normal distribution. The following result shows that despite its seemingly inflexible choice of support set, the variational class suitably approximates the true posterior.

For a given model $S \subseteq \{1,\dots,p\}$, let $X_S$ denote the $n\times |S|$-submatrix of the full regression matrix $X$, where we keep only the columns $X_{\cdot i}$, $i\in S$. Let $\hat{\theta}_S = (X_S^T X_S)^{-1}X_S^TY$ be the least squares estimator in the restricted model $Y = X_S \theta_S + Z$. If the restricted model were correctly specified, then $\hat{\theta}_S$ would have distribution $N_S (\theta_{0,S}, (X_S^T X_S)^{-1})$ under $P_{\theta_0}$. We approximate the posterior with a $N_S(\hat{\theta}_S,(X_S^T X_S)^{-1})\otimes \delta_{S^c}$ distribution, where $S$ is a suitable approximating set to which the posterior assigns sufficient probability. 

\begin{lemma}\label{lem:KL_non_diag}
If $4e^{1+\Gamma \log p - \kappa} \leq 1$, then the variational posterior $\hat{Q}$ arising from the family \eqref{def: variational:family2} satisfies
\begin{align*}
\emph{KL} (\hat{Q} \| \Pi(\cdot|Y))1_{\mathcal{T}_1} \leq \Gamma \log p + \frac{\lambda \Gamma}{\widetilde{\phi}(\Gamma)^2}  \left( 2s_0^{1/2} \eps+\frac{3\sqrt{\log p}}{\|X\|} \right) + \log (4e) .
\end{align*}
\end{lemma}

\begin{proof}
We construct our posterior approximation on the event $\mathcal{T}_1$ in \eqref{T1}. The posterior takes the form
\begin{equation}\label{eq:posterior}
\Pi(\cdot|Y)  = \sum_{S\subseteq\{1,...,p\}}\hat{q}_S \Pi_S(\cdot|Y) \otimes \delta_{S^c},
\end{equation}
where the weights $\hat{q}=(\hat{q}_S:\, S\subseteq\{1,...,p\})$ lie in the $2^p$-dimensional simplex and $\Pi_S(\cdot|Y)$ is the posterior for $\theta_S\in \R^{|S|}$ in the restricted model $Y = X_S \theta_S + Z$. Since
\begin{align*}
\Pi(\theta: \|\theta_{0,S_\theta^c}\|_2 > \eps |Y) & \leq \Pi(\theta: \|\theta - \theta_0\|_2 > \eps |Y),
\end{align*}
it follows that on $\mathcal{T}_1$,
$$\sum_{\substack{S:|S|\leq \Gamma \\ \|\theta_{0,S^c}\|_2 \leq \eps }}\hat{q}_S\geq 1-\frac{1}{4}-e^{-\kappa}\geq \frac{3}{4}-\frac{1}{4e}e^{-\Gamma \log p} \geq \frac{1}{2} $$
for all $p$ since $\Gamma>0$. Note further that
\begin{align*}
\Big|\big\{ S\subseteq\{1,...,p\}:\, |S|\leq \Gamma \big\}\Big|=\sum_{s=0}^{\lfloor \Gamma\rfloor} {p \choose s} \leq \sum_{s=0}^{\lfloor \Gamma\rfloor} \frac{p^s}{s!} \leq e p^{\Gamma}.
\end{align*}
Together, the last two displays show that on $\mathcal{T}_1$ and for all $p$, there exists a set $\tilde{S}$ satisfying
\begin{equation}\label{S}
|\tilde{S}|\leq \Gamma, \quad \|\theta_{0,\tilde{S}^c}\|_2 \leq \eps, \quad \hat{q}_{\tilde{S}} \geq (2e)^{-1} p^{-\Gamma}.
\end{equation}

Since an $N_S(\mu_S,\Sigma_S)\otimes \delta_{S^c}$ distribution is only absolutely continuous with respect to the $\hat{q}_S \Pi_S(\cdot|Y) \otimes \delta_{S^c}$ term of the posterior \eqref{eq:posterior},
\begin{equation}\label{eq:KL_radon_nik}
\begin{split}
\inf_{Q \in \mathcal{Q}}\KL(Q||\Pi(\cdot|Y))&=\inf_{S,\mu_S,\Sigma_S}E_{\theta \sim N_S(\mu_S,\Sigma_S)\otimes \delta_{S^c}}\log\frac{d N_S(\mu_S,\Sigma_S) \otimes \delta_{S^c}}{\hat{q}_S d\Pi_S(\cdot|Y) \otimes \delta_{S^c}}\\
&\leq \log\frac{1}{\hat{q}_{\tilde{S}}}+\inf_{\mu_{\tilde{S}},\Sigma_{\tilde{S}} }\KL\big(N_{\tilde{S}}(\mu_{\tilde{S}},\Sigma_{\tilde{S}})\| \Pi_{\tilde{S}}(\cdot|Y) \big),
\end{split}
\end{equation}
where the last Kullback-Leibler divergence is over $|\tilde{S}|$-dimensional distributions. On $\mathcal{T}_1$, $\log(1/\hat{q}_{\tilde{S}}) \leq \log (2ep^{\Gamma}) =\log (2e) + \Gamma \log p$. It thus remains to bound the second term in \eqref{eq:KL_radon_nik}.

Let $E_{\mu_S,\Sigma_S}$ denote the expectation under the law $\theta_S \sim N_S(\mu_S,\Sigma_S)$. Setting
\begin{equation}\label{eq:mean_cov}
\begin{split}
\mu_{\tilde{S}}  =  (X_{\tilde{S}}^TX_{\tilde{S}})^{-1}X_{\tilde{S}}^TY \quad \text{ and } \quad  \Sigma_{\tilde{S}} =(X_{\tilde{S}}^TX_{\tilde{S}})^{-1},
\end{split}
\end{equation}
one can check that the resulting normal distribution has density function proportional to $e^{-\frac{1}{2}\| Y-X_{\tilde{S}}\theta_{\tilde{S}}\|_2^2}$, $\theta_{\tilde{S}}\in \mathbb{R}^{|\tilde{S}|}$. Therefore,
\begin{equation}\label{eq:KL_split}
\begin{split}
\KL\big(N_{\tilde{S}}(\mu_{\tilde{S}},\Sigma_{\tilde{S}})\| \Pi_{\tilde{S}}(\cdot|Y)\big)&=E_{\mu_{\tilde{S}},\Sigma_{\tilde{S}}}\log\frac{D_\Pi e^{-\frac{1}{2}\|Y-X_{\tilde{S}}\theta_{\tilde{S}}\|_2^2-\lambda\|\theta_{0,\tilde{S}}\|_1}}{D_N e^{-\frac{1}{2}\|Y-X_{\tilde{S}}\theta_{\tilde{S}}\|_2^2-\lambda\|\theta_{\tilde{S}}\|_1}}\\
&=E_{\mu_{\tilde{S}},\Sigma_{\tilde{S}}}\lambda(\|\theta_{\tilde{S}}\|_1-\|\theta_{0,\tilde{S}}\|_1)+  \log(D_\Pi/D_N),
\end{split}
\end{equation}
with $D_\Pi=\int_{\R^{|\tilde{S}|}} e^{-\frac{1}{2}\|Y-X_{\tilde{S}}\theta_{\tilde{S}}\|_2^2-\lambda\|\theta_{\tilde{S}}\|_1} d\theta_{\tilde{S}}$ and $D_N=\int_{\R^{|\tilde{S}|}} e^{-\frac{1}{2}\|Y-X_{\tilde{S}}\theta_{\tilde{S}}\|_2^2-\lambda\|\theta_{0,\tilde{S}}\|_1} d\theta_{\tilde{S}}$ the normalizing constants.

We firstly upper bound $\log(D_\Pi/D_N)$. Define
$$B_{\tilde{S}} = \{ \theta_{\tilde{S}} \in \R^{|\tilde{S}|}: \|\theta_{\tilde{S}} - \theta_{0,\tilde{S}}\|_2 \leq 2\eps \}.$$
Let $\bar{\theta}_{\tilde{S}}$ denote the extension of a vector $\theta_{\tilde{S}}\in \R^{|\tilde{S}|}$ to $\R^p$ with $\bar{\theta}_{\tilde{S},j} = \theta_{\tilde{S},j}$ for $j\in \tilde{S}$ and $\bar{\theta}_{\tilde{S},j} = 0$ for $j\not\in \tilde{S}$. On $\mathcal{T}_1$, using \eqref{eq:posterior} and \eqref{S},
\begin{align*}
\Pi_{\tilde{S}}(B_{\tilde{S}}^c|Y) & \leq \frac{\hat{q}_{\tilde{S}} }{\hat{q}_{\tilde{S}} } \Pi_{\tilde{S}} (\theta_{\tilde{S}} \in \R^{|\tilde{S}|}:
\|\bar{\theta}_{\tilde{S}} - \theta_0\|_2 > 2\eps - \|\theta_{0,\tilde{S}^c}\|_2 |Y)\\
& \leq \hat{q}_{\tilde{S}}^{-1} \Pi(\theta\in \R^p: \|\theta-\theta\|_2 > \eps |Y) \\
& \leq 2e p^{\Gamma}  e^{-\kappa} = 2 e^{1+\Gamma \log p-\kappa} \leq 1/2,
\end{align*}
where the last inequality holds by assumption. Using Bayes formula, this yields
$$\Pi_{\tilde{S}}(B_{\tilde{S}}|Y)1_{\mathcal{T}_1} =\frac{\int_{B_{\tilde{S}}} e^{-\frac{1}{2}\|Y-X_{\tilde{S}}\theta_{\tilde{S}}\|_2^2-\lambda\|\theta_{\tilde{S}}\|_1} d\theta_{\tilde{S}}}{\int_{\R^{|\tilde{S}|}} e^{-\frac{1}{2}\|Y-X_{\tilde{S}}\theta_{\tilde{S}}\|_2^2-\lambda\|\theta_{\tilde{S}}\|_1} d\theta_{\tilde{S}} } 1_{\mathcal{T}_1} \geq \frac{1}{2} 1_{\mathcal{T}_1}$$
almost surely. In particular, $D_\Pi \leq 2\int_{B_{\tilde{S}}} e^{-\frac{1}{2}\|Y-X_{\tilde{S}}\theta_{\tilde{S}}\|_2^2-\lambda\|\theta_{\tilde{S}}\|_1} d\theta_{\tilde{S}}$ on $\mathcal{T}_1$. Therefore on $\mathcal{T}_1$,
\begin{align*}
\log\frac{D_\Pi}{D_N}&\leq \log \frac{2 \int_{B_{\tilde{S}}} e^{-\frac{1}{2}\|Y-X_{\tilde{S}}\theta_{\tilde{S}}\|_2^2-\lambda\|\theta_{\tilde{S}}\|_1} d\theta_{\tilde{S}}}{\int_{B_{\tilde{S}}} e^{-\frac{1}{2}\|Y-X_{\tilde{S}}\theta_{\tilde{S}}\|_2^2-\lambda\|\theta_{0,\tilde{S}}\|_1} d\theta_{\tilde{S}}} \\
&\leq \sup_{\theta_{\tilde{S}}\in B_{\tilde{S}}}\log e^{\lambda\|\theta_{0,\tilde{S}}\|_1-\lambda \|\theta_{\tilde{S}}\|_1} + \log 2 \\
&\leq \sup_{\theta_{\tilde{S}}\in B_{\tilde{S}}} \lambda\|\theta_{\tilde{S}}-\theta_{0,\tilde{S}}\|_1 + \log 2\\
&\leq \sup_{\theta_{\tilde{S}}\in B_{\tilde{S}}} \lambda |\tilde{S}|^{1/2} \|\theta_{\tilde{S}}-\theta_{0,\tilde{S}}\|_2 + \log 2 \\
&\leq 2\lambda\Gamma^{1/2}\eps + \log 2,
\end{align*}
where in the fourth inequality we have applied Cauchy-Schwarz.

We now turn to the first term in \eqref{eq:KL_split}. On $\mathcal{T}_1$, using the triangle inequality and Cauchy-Schwarz,
\begin{equation}\label{eq:post_bias_var}
\begin{split}
\lambda E_{\mu_{\tilde{S}},\Sigma_{\tilde{S}}} (\|\theta_{\tilde{S}}\|_1-\|\theta_{0,\tilde{S}}\|_1) & \leq \lambda \|\mu_{\tilde{S}} - \theta_{0,\tilde{S}}\|_1 + \lambda E_{0,\Sigma_{\tilde{S}}} \| \theta_{\tilde{S}}\|_1 \\
& \leq \lambda \Gamma^{1/2} \big( \|\mu_{\tilde{S}} - \theta_{0,\tilde{S}}\|_2 + \tr(\Sigma_{\tilde{S}})^{1/2} \big) 
\end{split}
\end{equation}
since $E_{0,\Sigma_{\tilde{S}}} \| \theta_{\tilde{S}}\|_2^2 = \tr(\Sigma_{\tilde{S}})$. Let $\Lambda_{\min}(A)$ and $\Lambda_{\max}(A)$ denote the smallest and largest eigenvalues, respectively, of a symmetric, positive definite matrix $A$. Using the variational characterization of maximal/minimal eigenvalues (\cite{horn:2013}, p. 234), for any $S\subseteq \{1,\dots,p\}$,
\begin{align}\label{eq:min_eigen}
\Lambda_{\min} (X_S^T X_S) = \min_{v\in\mathbb{R}^{|S|} :v\neq 0} \frac{v^T X_S^T X_Sv}{\|v\|_2^2} = \min_{u\in\mathbb{R}^p : u\neq 0, u_{S^c} = 0} \frac{\|Xu\|_2^2}{\|u\|_2^2} \geq \|X\|^2 \widetilde{\phi}(|S|)^2.
\end{align}
Therefore,
\begin{align*}
\tr(\Sigma_{\tilde{S}}) \leq \Gamma \Lambda_{\max}((X_{\tilde{S}}^T X_{\tilde{S}})^{-1}) \leq \frac{\Gamma}{\Lambda_{\min}(X_{\tilde{S}}^T X_{\tilde{S}})} \leq \frac{\Gamma}{\|X\|^2 \widetilde{\phi}(\Gamma)^2}.
\end{align*}
Under $P_{\theta_0}$, using \eqref{eq:model} and \eqref{eq:mean_cov}, the bias term can be decomposed as
$$\|\mu_{\tilde{S}} - \theta_{0,\tilde{S}}\|_2 \leq \|(X_{\tilde{S}}^T X_{\tilde{S}})^{-1}X_{\tilde{S}}^T X_{{\tilde{S}}^c}\theta_{0,\tilde{S}^c}\|_2 + \|(X_{\tilde{S}}^T X_{\tilde{S}})^{-1}X_{\tilde{S}}^T Z\|_2 = I + II.$$
For $I$, note first that the $\ell_2$-operator norm of $(X_{\tilde{S}}^T X_{\tilde{S}})^{-1}$ is bounded by $1/(\|X\|^2\widetilde{\phi}(|{\tilde{S}}|)^2)$ by \eqref{eq:min_eigen}. On $\mathcal{T}_1$, using Cauchy-Schwarz,
\begin{align*}
\|X_{\tilde{S}}^T X_{{\tilde{S}}^c} \theta_{0,{\tilde{S}}^c}\|_2^2 & = \sum_{i\in {\tilde{S}}} \left( \sum_{k=1}^n \sum_{j\in {\tilde{S}}^c} X_{ki} X_{kj}\theta_{0,j} \right)^2 \\
& = \sum_{i\in {\tilde{S}}} \left( \sum_{j\in {\tilde{S}}^c } \langle X_{\cdot i}, X_{\cdot j} \rangle \theta_{0,j} \right)^2\\
& \leq \|X\|^4 \sum_{i\in {\tilde{S}}} \left( \sum_{j\in {\tilde{S}}^c \cap S_{0}} |\theta_{0,j}| \right)^2\\
& \leq \|X\|^4 |\tilde{S}| s_0 \|\theta_{0,\tilde{S}^c}\|_2^2.
\end{align*}
Together with \eqref{S}, this gives
$$I \leq \frac{\|X\|^2 |\tilde{S}|^{1/2} s_0^{1/2} \|\theta_{0,\tilde{S}^c}\|_2}{\|X\|^2 \widetilde{\phi}(|{\tilde{S}}|)^2} \leq \frac{\Gamma^{1/2} s_0^{1/2} \eps}{ \widetilde{\phi}(|{\tilde{S}}|)^2}.$$
Using the same bound on the $\ell_2$-operator norm and \eqref{eq:model}, on the event $\mathcal{T}_1 \subset \mathcal{T}_0$ it holds that
\begin{align*}
II & \leq \frac{\|X_{\tilde{S}}^T Z\|_2}{\|X\|^2 \widetilde{\phi}(|{\tilde{S}}|)^2}  = \frac{1}{\|X\|^2 \widetilde{\phi}(|\tilde{S}|)^2} \left( \sum_{i\in \tilde{S}} \left( X^T(Y-X\theta_0) \right)_i^2 \right)^{1/2} \leq \frac{2|\tilde{S}|^{1/2}\sqrt{\log p}}{\|X\| \widetilde{\phi}(|\tilde{S}|)^2}.
\end{align*}
Combining all of the above bounds and using that $|\tilde{S}| \leq \Gamma$, on the event $\mathcal{T}_1$,
\begin{align*}
\lambda E_{\mu_{\tilde{S}},\Sigma_{\tilde{S}}} (\|\theta_{\tilde{S}}\|_1-\|\theta_{0,\tilde{S}}\|_1) & \leq \frac{\lambda \Gamma}{\widetilde{\phi}(\Gamma)^2}  \left( s_0^{1/2} \eps+\frac{2\sqrt{\log p}}{\|X\|} + \frac{\widetilde{\phi}(|\tilde{S}|) }{\|X\|} \right).
\end{align*}
Together with \eqref{eq:KL_split}, the bound $\log(D_\Pi/D_N) \leq 2\lambda \Gamma^{1/2} \eps + \log 2$ derived above and that $\widetilde{\phi}(|\tilde{S}|) \leq \widetilde{\phi}(1) \leq 1$, this yields
\begin{align*}
\KL\big(N_{\tilde{S}}(\mu_{\tilde{S}},\Sigma_{\tilde{S}})\| \Pi_{\tilde{S}}(\cdot|Y)\big)1_{\mathcal{T}_1} \leq \frac{\lambda \Gamma}{\widetilde{\phi}(\Gamma)^2}  \left( 2s_0^{1/2} \eps+\frac{3\sqrt{\log p}}{\|X\|} \right) + \log 2 .
\end{align*}
Combining this with \eqref{eq:KL_radon_nik} and that $\log(1/\hat{q}_{\tilde{S}}) \leq  \log (2e) + \Gamma \log p$ completes the proof.
\end{proof}

We next consider the mean-field subclass $\mathcal{Q}_{MF}$ of $\mathcal{Q}$ given by \eqref{def: variational:family3}. This again selects a single fixed support $S$ but further requires the fitted normal distribution to have diagonal covariance matrix. We consider a diagonal version of $N_S(\hat{\theta}_S,(X_S^T X_S)^{-1})\otimes \delta_{S^c}$ considered in Lemma \ref{lem:KL_non_diag}.

\begin{lemma}\label{lem:KL_Q_MF}
If $4e^{1+\Gamma \log p - \kappa} \leq 1$, then the variational posterior $\widetilde{Q}$ arising from the family \eqref{def: variational:family3} satisfies
\begin{align*}
\emph{KL} (\widetilde{Q} \| \Pi(\cdot|Y))1_{\mathcal{T}_1} \leq \Gamma \log \frac{p}{\widetilde{\phi}(\Gamma)} + \frac{\lambda \Gamma}{\widetilde{\phi}(\Gamma)^2}  \left( 2s_0^{1/2} \eps+\frac{3\sqrt{\log p}}{\|X\|} \right) + \log (4e) .
\end{align*}
\end{lemma}

\begin{proof}
We showed in the proof of Lemma \ref{lem:KL_non_diag} that on the event $\mathcal{T}_1$ given in \eqref{T1}, there exists a set $\tilde{S}$ satisfying \eqref{S}. Arguing as in \eqref{eq:KL_radon_nik},
\begin{equation*}
\begin{split}
\inf_{Q \in \mathcal{Q}_{MF}}\KL(Q||\Pi(\cdot|Y)) &\leq \log\frac{1}{\hat{q}_{\tilde{S}}}+\inf_{\mu_{\tilde{S}},D_{\tilde{S}} }\KL\big(N_{\tilde{S}}(\mu_{\tilde{S}},D_{\tilde{S}})\| \Pi_{\tilde{S}}(\cdot|Y) \big),
\end{split}
\end{equation*}
where the last Kullback-Leibler divergence is over the $|\tilde{S}|$-dimensional distributions and $D_{\tilde{S}}$ ranges over diagonal positive definite matrices. On $\mathcal{T}_1$ and for all $p$, we have $\log(1/\hat{q}_{\tilde{S}}) \leq \log (2ep^\Gamma) =\log (2e) + \Gamma\log p$ by \eqref{S}.

The latter Kullback-Leibler divergence equals
\begin{equation}\label{KL:div:diag}
\KL\big(N_{\tilde{S}}(\mu_{\tilde{S}},D_{\tilde{S}})\| \Pi_{\tilde{S}}(\cdot|Y) \big) = E_{\mu_{\tilde{S}},D_{\tilde{S}}} \left[ \log \frac{dN_{\tilde{S}}(\mu_{\tilde{S}},D_{\tilde{S}})}{dN_{\tilde{S}}(\mu_{\tilde{S}},\Sigma_{\tilde{S}})} + \log \frac{dN_{\tilde{S}}(\mu_{\tilde{S}},\Sigma_{\tilde{S}})}{d\Pi_{\tilde{S}}(\cdot|Y)}\right]
\end{equation}
for any covariance matrix $\Sigma_{\tilde{S}}$. For the first term in \eqref{KL:div:diag}, the formula for the Kullback-Leibler divergence between two multivariate Gaussians gives
\begin{align*}
\KL\big(N_{\tilde{S}}(\mu_{\tilde{S}},D_{\tilde{S}})\| N_{\tilde{S}}(\mu_{\tilde{S}},\Sigma_{\tilde{S}}) \big)=\tfrac{1}{2}\big(\log (|\Sigma_{\tilde{S}}|/|D_{\tilde{S}}|)-|\tilde{S}|+\tr(\Sigma_{\tilde{S}}^{-1}D_{\tilde{S}}) \big),
\end{align*}
where $|A|$ denotes the determinant of a square matrix $A$. Set now $\mu_{\tilde{S}}=(X_{\tilde{S}}^TX_{\tilde{S}})^{-1}X_{\tilde{S}}^TY$, $\Sigma_{\tilde{S}}=(X_{\tilde{S}}^TX_{\tilde{S}})^{-1}$ as in \eqref{eq:mean_cov} and define the diagonal matrix $D_{\tilde{S}}$ via $(D_{\tilde{S}})_{ii}=1/(\Sigma_{\tilde{S}}^{-1})_{ii} = 1/(X_{\tilde{S}}^T X_{\tilde{S}})_{ii}$. This gives $\tr (\Sigma_{\tilde{S}}^{-1} D_{\tilde{S}}) = |\tilde{S}|$, so that it remains to control $\tfrac{1}{2}\log (|\Sigma_{\tilde{S}}|/|D_{\tilde{S}}|) = \tfrac{1}{2}\log (|\Sigma_{\tilde{S}}||D_{\tilde{S}}^{-1}|)$. For our choice of $D_{\tilde{S}}$,
\begin{align*}
|D_{\tilde{S}}^{-1}| = \prod_{j=1}^{|\tilde{S}|} (\Sigma_{\tilde{S}}^{-1})_{jj} = \prod_{j=1}^{|\tilde{S}|} (X_{\tilde{S}}^T X_{\tilde{S}})_{jj}  \leq \|X\|^{2|\tilde{S}|},
\end{align*}
while for $\Lambda_{\min}(A)$ and $\Lambda_{\max}(A)$ the smallest and largest eigenvalues, respectively, of a matrix $A$ and using \eqref{eq:min_eigen},
\begin{align*}
|\Sigma_{\tilde{S}}| \leq \Lambda_{\max}((X_{\tilde{S}}^T X_{\tilde{S}})^{-1})^{|\tilde{S}|} = (1/\Lambda_{\min}(X_{\tilde{S}}^T X_{\tilde{S}}))^{|\tilde{S}|} \leq 1/(\|X\| \widetilde{\phi}(|\tilde{S}|))^{2|\tilde{S}|}.
\end{align*}
This yields that $\KL(N_{\tilde{S}}(\mu_{\tilde{S}},D_{\tilde{S}})\| N_{\tilde{S}}(\mu_{\tilde{S}},\Sigma_{\tilde{S}}) ) \leq |\tilde{S}| \log (1/\widetilde{\phi}(|\tilde{S}|)) \leq \Gamma \log (1/\widetilde{\phi}(\Gamma))$.

Note that the second term in \eqref{KL:div:diag} is identical to the expression \eqref{eq:KL_split}, except that the expectation is taken under $\theta_{\tilde{S}} \sim N_{\tilde{S}}(\mu_{\tilde{S}},D_{\tilde{S}})$ instead of $\theta_{\tilde{S}} \sim N_{\tilde{S}}(\mu_{\tilde{S}},\Sigma_{\tilde{S}})$. One may therefore use the exact same arguments as in Lemma \ref{lem:KL_non_diag} with the only difference occurring in the second term in \eqref{eq:post_bias_var}, where one instead has $\lambda E_{0,D_{\tilde{S}}} \| \theta_{\tilde{S}}\|_1 \leq \lambda |\tilde{S}|^{1/2} (E_{0,D_{\tilde{S}}} \| \theta_{\tilde{S}}\|_2^2)^{1/2} = \lambda |\tilde{S}|^{1/2}\tr(D_{\tilde{S}})^{1/2}$. For $e_i$ the $i^{th}$ unit vector in $\R^p$,
$$\tr(D_{\tilde{S}}) = \sum_{i=1}^{|\tilde{S}|} \frac{1}{(X_{\tilde{S}}^T X_{\tilde{S}})_{ii}} =\sum_{i\in \tilde{S}} \frac{1}{\|Xe_i\|_2^2} \leq \sum_{i\in \tilde{S}} \frac{1}{\|X\|^2 \|e_i\|_2^2 \widetilde{\phi}(1)^2} = \frac{|\tilde{S}|}{\|X\|^2 \widetilde{\phi}(1)^2},$$
so that $\lambda |\tilde{S}|^{1/2}\tr(D_{\tilde{S}})^{1/2} \leq \lambda \Gamma/(\|X\|\widetilde{\phi}(1))$. Combining the bounds as in Lemma \ref{lem:KL_non_diag} then gives the result.
\end{proof}

\begin{lemma}\label{lem:KL_MF}
If $4e^{1+\Gamma \log p - \kappa} \leq 1$, then the variational posterior $\widetilde{\Pi}$ arising from the family \eqref{def: variational:family} of spike-and-slab distributions satisfies
\begin{align*}
\emph{KL} (\widetilde{\Pi} \| \Pi(\cdot|Y))1_{\mathcal{T}_1} \leq \Gamma \log \frac{p}{\widetilde{\phi}(\Gamma)} + \frac{\lambda \Gamma}{\widetilde{\phi}(\Gamma)^2}  \left( 2s_0^{1/2} \eps+\frac{3\sqrt{\log p}}{\|X\|} \right) + \log (4e) .
\end{align*}
\end{lemma}

\begin{proof}
Since $\mathcal{Q}_{MF} \subset \mathcal{P}_{MF}$, we have $\KL (\widetilde{\Pi}\| \Pi(\cdot|Y) ) \leq \KL (\widetilde{Q}\| \Pi(\cdot|Y))$. The result then follows from Lemma \ref{lem:KL_Q_MF}.
\end{proof}

\subsection{Oracle contraction rates for the original posterior distribution}\label{sec:contraction}

Oracle type contraction rates for the original posterior were established in Castillo et al. \cite{castillo:2015}. However, their results are not stated with exponential bounds as needed in \eqref{cond:exp_decay}, so we must reformulate them in order to apply our Theorem \ref{thm: general:VB}. The required exponential bounds in fact follow from their proofs; we recall here the required results and, since \cite{castillo:2015} is a rather technical article, we provide a brief explanation why the exponential bounds hold.

\begin{lemma}[Theorem 10 of \cite{castillo:2015}]\label{lem: support}
Suppose the prior satisfies \eqref{eq:prior_cond} and \eqref{prior_lambda}. Then for $p$ large enough depending on $A_2,A_4$, any $M>0$ and any $\theta_0,\theta_*\in \R^p$,
\begin{align*}
& E_{\theta_0} \Pi\left( \theta: |S_\theta| \geq |S_*| + \frac{4M}{A_4} \left(1 +\frac{16}{\phi(S_*)^2}\frac{\lambda}{\bar{\lambda}} \right)|S_*|  + \frac{4M\|X(\theta_0-\theta_*)\|_2^2}{A_4 \log p}\Big| Y \right)1_{\mathcal{T}_0}  \\
& \quad \quad \leq C(A_2,A_4) \exp \left(  -(M-2) \left( 1+\tfrac{16}{\phi(S_*)^2}\tfrac{\lambda}{\bar{\lambda}} \right) |S_*|\log p - (M-1)\|X(\theta_0 - \theta_*)\|_2^2 \right),
\end{align*}
where $S_* = S_{\theta_*}$ and $\mathcal{T}_0$ is the event in \eqref{eq:T_event}.
\end{lemma}

\begin{proof}
Following the proof of Theorem 10 of \cite{castillo:2015}, one obtains using (6.3) and the second display on p. 2008 of \cite{castillo:2015} that for $\bar{\lambda} = 2\|X\|\sqrt{\log p}$, any $\theta_*$ and any measurable set $B \subseteq \R^p$,
$$\sup_{\theta_0 \in \R^p} E_{\theta_0} \Pi(B|Y)1_{\mathcal{T}_0} \leq e^{\|X(\theta_0-\theta_*)\|_2^2} \left( \frac{ep^{2s_*}}{\pi_p(s_*)} e^\frac{8\lambda\bar{\lambda}s_*}{\|X\|^2\phi(S_*)^2} \int_B e^{-(\lambda/4)\|\theta-\theta_*\|_1 + \lambda\|\theta\|_1} d\Pi(\theta) \right)^{1/2}.$$
Setting now $B = \{\theta:|S_\theta| > R\}$ for $R \geq s_*$, the third display on p. 2008 of \cite{castillo:2015} shows that
\begin{align*}
\int_B e^{-(\lambda/4)\|\theta-\theta_*\|_1 + \lambda\|\theta\|_1} d\Pi(\theta) & \leq \pi_p(s_*) 4^{s_*} \left( \frac{4A_2}{p^{A_4}}\right)^{R+1-s_*} \sum_{j=0}^\infty \left( \frac{4A_2}{p^{A_4}} \right)^j.\\
& \leq C(A_2,A_4) \pi_p(s_*) 4^{s_*} \left( \frac{4A_2}{p^{A_4}}\right)^{R+1-s_*}
\end{align*}
for $p$ large enough that $4A_2/p^{A_4}<1$. Substituting this into the second last display and using that $\bar{\lambda}^2 = 4\|X\|^2 \log p$,
$$\sup_{\theta_0 \in \R^p} E_{\theta_0} \Pi(B|Y)1_{\mathcal{T}_0} \leq  C(A_2,A_4) e^{\|X(\theta_0-\theta_*)\|_2^2}
(2p)^{s_*} e^\frac{16 \lambda  s_*\log p}{\bar{\lambda}\phi(S_*)^2} \left( \frac{4A_2}{p^{A_4}}\right)^{(R+1-s_*)/2}.$$
Choosing $R = (2\delta+1)s_*-1 + 2\eta$, the right-hand side equals
\begin{align*}
C(A_2,A_4) \exp & \big\{ \|X(\theta_0-\theta_*)\|_2^2 + \left( \log 2 + \delta \log (4A_2) \right)s_* + \left( 1 + \tfrac{16\lambda}{\bar{\lambda}\phi(S_*)^2} - \delta A_4 \right) s_* \log p \\
& \quad + \eta (\log (4A_2) - A_4 \log p) \big\}.
\end{align*}
Further picking $\delta = 2M(1+16\lambda/(\bar{\lambda}\phi(S_*)^2))/A_4$ and $\eta = 2M\|X(\theta_0-\theta_*)\|_2^2/(A_4 \log p)$, the right-hand side is bounded by
$$C(A_2,A_4) \exp\{ -(M-2)(1+\tfrac{16\lambda}{\bar{\lambda}\phi(S_*)^2}) s_*\log p - (M-1) \|X(\theta_0 - \theta_*)\|_2^2 \}$$ 
for $p$ large enough depending on $A_2,A_4$, as required.
\end{proof}

The following result is a modified version of the oracle inequality in Theorem 3 of \cite{castillo:2015} with $S_* = S_0$. Since it is stated somewhat differently in \cite{castillo:2015}, we sketch why this is true. 

\begin{lemma}[Theorem 3 of \cite{castillo:2015}]\label{lem: contraction:spike:slab}
Suppose the prior satisfies \eqref{eq:prior_cond} and \eqref{prior_lambda}. Then there exists a constant $M>0$ such that for $p$ large enough, both depending only on $A_1,A_3,A_4$, any $L\geq 1$, and uniformly over all $\theta_0,\theta_* \in \R^p$ with $|S_{\theta_*}| \leq |S_{\theta_0}|$,
%\begin{align*}
%& E_{\theta_0} \Pi \left( \theta: \|X(\theta-\theta_0)\|_2 > \frac{M\sqrt{s_0\log p}}{\overline{\psi}_{L+4}(S_0)}  (L^{1/2}+1) \left( 1 + \frac{1}{\phi(S_0)}\right)  \Big| Y \right)1_{\mathcal{T}_0} \\
%& \quad \quad \leq C \exp \left(  -L \left( 1+\tfrac{8}{\phi(S_0)^2}\tfrac{\lambda}{\|X\|\sqrt{\log p}} \right) s_0\log p \right) \\
%& \quad \quad \quad + C\exp\left(-\frac{2}{\overline{\psi}_{L+4}(S_0)^2} \left(1 + \frac{L+4}{A_4}\left( 1+ \frac{8}{\phi(S_0)^2} \frac{\lambda}{\|X\|\sqrt{\log p}} \right)  \right)s_0 \log p \right),
%\end{align*}
\begin{align*}
& E_{\theta_0} \Pi \left( \theta : \|X(\theta-\theta_0)\|_2 > \frac{ML^{1/2} }{\overline{\psi}_{L+2}(S_0)} \left[ \frac{\sqrt{s_*\log p}}{\phi(S_*)} +\|X(\theta_0-\theta_*)\|_2 \right] \Big| Y \right)1_{\mathcal{T}_0} \\
& \leq C \exp \left(- \left[ L\wedge \tfrac{4(L+2)}{A_4} \right] \left[ (1+\tfrac{16}{\phi(S_*)^2}\tfrac{\lambda}{\bar{\lambda}}) s_*\log p + \|X(\theta_0-\theta_*)\|_2^2 \right] \right) \\
& \qquad + C\exp(-L(1+\tfrac{16}{\phi(S_0)^2}\tfrac{\lambda}{\bar{\lambda}}) s_0\log p),
\end{align*}
where $s_0 = |S_{\theta_0}|$, $s_* = |S_{\theta_*}|$ and $C=C(A_2,A_4)$. Moreover, both
$$E_{\theta_0} \Pi \left( \theta : \|\theta-\theta_0\|_1 > \| \theta_0 - \theta_*\|_1 + \frac{ML }{\overline{\psi}_{L+2}(S_0)^2} \left[ \frac{s_* \sqrt{\log p}}{\|X\| \phi(S_*)^2} +\frac{\|X(\theta_0-\theta_*)\|_2^2}{\|X\| \sqrt{\log p}} \right] \Big| Y \right)1_{\mathcal{T}_0},$$
$$E_{\theta_0} \Pi \left( \theta : \|\theta-\theta_0\|_2 > \frac{ML^{1/2} }{\|X\| \widetilde{\psi}_{L+2}(S_0)^2} \left[ \frac{\sqrt{s_*\log p}}{\phi(S_*)} +\|X(\theta_0-\theta_*)\|_2 \right] \Big| Y \right)1_{\mathcal{T}_0},$$
satisfy the same inequality.
\end{lemma}

\begin{proof}Unless otherwise stated, we use here the notation from \cite{castillo:2015}. As on p. 2008 of \cite{castillo:2015}, define the event $E = \{ \theta: |S_\theta| \leq D_* \wedge D_0\}$ for
\begin{equation}\label{D_*}
D_* = D_*(L) = s_* + \frac{4(L+2)}{A_4}\left( 1+ \frac{16}{\phi(S_*)^2} \frac{\lambda}{\bar{\lambda}} \right) s_* + \frac{4(L+2)\|X(\theta_0-\theta_*)\|_2^2}{A_4 \log p},
\end{equation}
where $\bar{\lambda} = 2\|X\|\sqrt{\log p}$ and $D_0$ is the same expression with $\theta_*$ replaced by $\theta_0$. Note that we take different constants than in (6.7) of \cite{castillo:2015} to obtain the required exponential tail bound. Lemma \ref{lem: support} yields, with $M=L+2$ and since $s_* \leq s_0$,
\begin{equation}\label{localization}
\begin{split}
E_{\theta_0}\Pi(E^c|Y) 1_{\mathcal{T}_0} & = E_{\theta_0} \Pi(\theta: |S_\theta| > D_*\wedge D_0 |Y) 1_{\mathcal{T}_0} \\
& \leq C(A_2,A_4) \exp(-L(1+\tfrac{16}{\phi(S_0)^2}\tfrac{\lambda}{\bar{\lambda}}) s_0\log p) \\
& \quad + C(A_2,A_4) \exp(-L(1+\tfrac{16}{\phi(S_*)^2}\tfrac{\lambda}{\bar{\lambda}}) s_*\log p - L\|X(\theta_0-\theta_*)\|_2^2)
\end{split}
\end{equation}
for every $\theta_0 \in \R^p$, so we can intersect the desired set with $E$ in what follows.
%Let
%$$\Pi^E(\cdot) = \Pi(\cdot \cap E) / \Pi(E)$$
%be the prior conditioned to $E$ and denote by $\Pi^E(\cdot|Y)$ the posterior distribution arising from the prior $\Pi^E$. By a standard inequality (\cite{vdvaart:1998}, p. 142),
%$$\|\Pi(\cdot|Y) - \Pi^E(\cdot|Y)\|_{TV} = \sup_A |\Pi(A|Y) - \Pi^E(A|Y)| \leq 2 \Pi(E^c|Y),$$
%where the supremum is taken over all measurable sets $A \subseteq \R^p$. Using the last display and \eqref{localization}, we may therefore work on the restricted prior $\Pi^E$.

From definition \eqref{tilde_compat}, we have $\overline{\psi}_{L+2}(S_0) = \overline{\phi}(D_0 + s_0)$. Continuing through the proof, the third last display on p. 2009 of \cite{castillo:2015} (note that up to this point, the definitions of $D_*$ and $D_0$ only affect the definition of the compatibility type constants) gives
\begin{align*}
\Pi & (\theta\in E: \|X(\theta-\theta_0)\|_2 > 4\|X(\theta_*-\theta_0)\|_2 + R|Y) 1_{\mathcal{T}_0} \\
& \leq \frac{e}{\pi_p(0)A_1^{s_*}}  p^{(2+A_3)s_*} e^{\frac{32\bar{\lambda}^2(D_*+s_*)}{\|X\|^2 \overline{\psi}_{L+2}(S_0)^2}} e^{-\frac{R^2}{8}} \sum_{s=0}^p \pi_p(s) 2^s,
\end{align*}
where again $\bar{\lambda} = 2\|X\|\sqrt{\log p}$. By condition \eqref{eq:prior_cond}, $\sum_{s=0}^p \pi_p(s) 2^s\leq \pi_p(0) \sum_{s=0}^p (2A_2p^{-A_4})^s \leq \pi_p(0) C(A_2,A_4)$ for $p$ large enough. Using this and taking $R^2 = \overline{M}^2 (D_*+s_*)\log p/\overline{\psi}_{L+2}(S_0)^2$, the last display is bounded by
\begin{align*}
& C(A_2,A_4) \exp \left\{ -s_* \log A_1 + (2+A_3)s_*\log p + \frac{128(D_*+s_*)\log p}{\overline{\psi}_{L+2}(S_0)^2}-\frac{1}{8}R^2\right\} \\
& \leq C(A_2,A_4) \exp \left\{ - \left[ \frac{\overline{M}^2}{8} -130-A_3 -\frac{|\log A_1|}{\log p} \right] \frac{(D_*+s_*)\log p}{\overline{\psi}_{L+2}(S_0)^2} \right\},
\end{align*}
where we have also used $\overline{\psi}_{L+2}(S_0) \leq \overline{\phi}(1) \leq 1$ for any $S_0$. Using the definition \eqref{D_*} of $D_*$, that $\lambda/\bar{\lambda} \leq 2$ and the inequality $\sqrt{x+y} \leq \sqrt{x} + \sqrt{y}$ for any $x,y\geq 0$,
$$(D_*+s_*)^{1/2} \leq Cs_*^{1/2}L^{1/2}/\phi(S_*) + CL^{1/2} \|X(\theta_0-\theta_*)\|_2/\sqrt{\log p}$$
for a constant $C>0$ depending only on $A_4$, yielding
$$R \leq \frac{C\overline{M}L^{1/2}}{\overline{\psi}_{L+2}(S_0)} \left( \frac{\sqrt{s_*\log p}}{\phi(S_*)} + \|X(\theta_0-\theta_*)\|_2 \right).$$
Combining this with the third last display gives
\begin{align*}
& \Pi \left( \theta\in E: \|X(\theta-\theta_0)\|_2 > \frac{ML^{1/2} }{\overline{\psi}_{L+2}(S_0)} \left[ \frac{\sqrt{s_*\log p}}{\phi(S_*)} +\|X(\theta_0-\theta_*)\|_2 \right] \Big| Y \right)1_{\mathcal{T}_0} \\
& \qquad \qquad \leq C(A_2,A_4) \exp(-(D_*+s_*)\log p/\overline{\psi}_{L+2}(S_0)^2)
\end{align*}
for some $M>0$ large enough depending only on $A_1,A_3,A_4$. Using $\overline{\psi}_{L+2}(S_0) \leq 1$ and the definition \eqref{D_*}, the probability in the last display is smaller than that in \eqref{localization} if $4(L+2)/A_4 \geq L$. Considering these two cases separately establishes the required inequality for the prediction error $\|X(\theta-\theta_0)\|_2$.

For $\ell_1$-loss, the result follows from that for prediction error and the first display on p. 2010 of \cite{castillo:2015}.

For $\ell_2$-loss, note that $\|X(\theta-\theta_0)\|_2 \geq \widetilde{\phi}(|S_{\theta-\theta_0}|)\|X\| \|\theta-\theta_0\|_2 \geq \widetilde{\psi}_{L+2}(S_0) \|X\| \|\theta-\theta_0\|_2$ for any $\theta \in E$. The result then follows from that for prediction error and that $\overline{\psi}_{L+2}(S_0) \geq \widetilde{\psi}_{L+2}(S_0)$ by Lemma \ref{lem:compat}.
\end{proof}

\setcounter{lemma}{0}
\setcounter{theorem}{0}
\setcounter{definition}{0}
\setcounter{remark}{0}
\setcounter{equation}{0}

\section{Additional methodological details}\label{sec:method_additional}\label{sec:proofs2}

\subsection{Proofs for the variational algorithm}\label{sec: proof:VBalg}

We provide here the derivations of the formulas used in the CAVI update equations of our variational algorithm in Section \ref{sec: VBalgorithm}.

Proof of \eqref{eq: joint:KL}: We compute the Kullback-Leibler divergence between $P_{\boldsymbol{\mu},\boldsymbol{\sigma},\boldsymbol{\gamma}}$ and the posterior $\Pi(\cdot |Y)$, conditional on $z_i=1$, as a function of $\mu_i$ and $\sigma_i$. Since the variational probability distribution of $\theta_i$ conditional on $z_i=1$ (i.e. $P_{\mu_i,\sigma_i|z_i=1}$) is singular to the Dirac measure $\delta_0$, in the Radon-Nikodym derivative $dP_{\mu_i,\sigma_i|z_i=1}/d \Pi_i$, where $\Pi_i$ is the prior for $\theta_i$, it suffices to consider only the continuous part of the prior measure in the denominator. Write $d\Pi(\theta|Y) = D_\Pi^{-1} e^{-\|Y-X\theta\|_2^2/2} d\Pi(\theta)$ with $D_\Pi$ the normalizing constant. Using all of these and the prior product structure, $\KL (P_{\boldsymbol{\mu},\boldsymbol{\sigma},\boldsymbol{\gamma}|z_i=1} \| \Pi(\cdot|Y) )$ equals, as a function of $\mu_i$ and $\sigma_i$,
\begin{align*}
& E_{\boldsymbol{\mu},\boldsymbol{\sigma},\boldsymbol{\gamma}|z_i=1}\left[ \tfrac{1}{2}\|Y-X\theta\|_2^2 + \log D_\Pi + \log \frac{dP_{\boldsymbol{\mu}_{-i},\boldsymbol{\sigma}_{-i},\boldsymbol{\gamma}_{-i}} \otimes N(\mu_i,\sigma_i^2)}{d\Pi_{-i} \otimes \overline{w}_i \text{Lap}(\lambda)} \right] \\
& = E_{\boldsymbol{\mu},\boldsymbol{\sigma},\boldsymbol{\gamma}|z_i=1}\left[ \tfrac{1}{2}(Y-X\theta)^T(Y-X\theta) + \log \frac{dP_{\boldsymbol{\mu}_{-i},\boldsymbol{\sigma}_{-i},\boldsymbol{\gamma}_{-i}}}{d\Pi_{-i}}(\theta_{-i}) - \log \sigma_i -\frac{(\theta_i-\mu_i)^2}{2\sigma_i^2} + \lambda |\theta_i|  \right] + C,
\end{align*}
where $C>0$ is independent of $\mu_i,\sigma_i$ and $\overline{w}_i = a_0 /(a_0+b_0)$ is the prior mean for $w_i$. Recall that the expected value of the folded normal distribution with parameters $\mu\in\mathbb{R}$ and $\sigma>0$ is $\sigma\sqrt{2/\pi} e^{-\mu^2/(2\sigma^2)}+\mu(1-2\Phi(-\mu/\sigma))$. Using this and explicitly evaluating the expectation of the first term, the last display equals
\begin{align*}
& \mu_i \sum_{k\neq i} (X^TX)_{ik} \gamma_k\mu_k+\frac{1}{2}(X^TX)_{ii}(\sigma_i^2+\mu_i^2)-(Y^TX)_i\mu_i+\lambda \sigma_i\sqrt{2/\pi}e^{-\mu_i^2/(2\sigma_i^2)}\\
&\qquad+\lambda\mu_i(1-2\Phi(-\mu_i/\sigma_i))-\log\sigma_i+C',
\end{align*}
where $C'>0$ is again independent of $\mu_i,\sigma_i$. Minimizing the last display with respect to either $\mu_i$ or $\sigma_i$ (but not jointly) gives the same minimizers as minimizing $f_i$ and $g_i$ in \eqref{eq: joint:KL}.

Proof of \eqref{eq: VB:step:lambda}:  Similarly to the derivation of \eqref{eq: joint:KL} above, the KL divergence between $P_{\boldsymbol{\mu},\boldsymbol{\sigma},\boldsymbol{\gamma}}$ and $\Pi(\cdot |Y)$ as a function of $\gamma_i$ equals
\begin{align*}
& E_{\boldsymbol{\mu},\boldsymbol{\sigma},\boldsymbol{\gamma}}\left[ \tfrac{1}{2}\|Y-X\theta\|_2^2 + \log \frac{dP_{\boldsymbol{\mu}_{-i},\boldsymbol{\sigma}_{-i},\boldsymbol{\gamma}_{-i}} }{d\Pi_{-i}}(\theta_{-i}) + \log \frac{d (\gamma_i N(\mu_i,\sigma_i^2) + (1-\gamma_i)\delta_0) }{d (\overline{w}_i \text{Lap}(\lambda)+(1-\overline{w}_i)\delta_0)}(\theta_i) \right] + C,
\end{align*}
where $C>0$ is independent of $\gamma_i$ and $\overline{w}_i = a_0/(a_0+b_0)$. Since on an event of $P_{\boldsymbol{\mu},\boldsymbol{\sigma},\boldsymbol{\gamma}}$-probability one, $\theta_i = 0$ if and only if $z_i = 0$, the last display equals
\begin{align}
& E_{\boldsymbol{\mu},\boldsymbol{\sigma},\boldsymbol{\gamma}}\left[ \tfrac{1}{2}\|Y-X\theta\|_2^2 + 1_{\{z_i=1\}} \log \frac{\gamma_i dN(\mu_i,\sigma_i^2)}{\overline{w}_i d\text{Lap}(\lambda)}(\theta_i) + 1_{\{z_i=0\}} \log \frac{1-\gamma_i}{1-\overline{w}_i} \right] + C\nonumber\\
& = E_{\boldsymbol{\mu},\boldsymbol{\sigma},\boldsymbol{\gamma}}\left[ \tfrac{1}{2}\|Y-X\theta\|_2^2 + 1_{\{z_i=1\}}\left( \log \frac{\sqrt{2}}{\sqrt{\pi}\sigma_i\lambda} -\frac{(\theta_i-\mu_i)^2}{2\sigma_i^2} + \lambda |\theta_i|  \right)  \right]\nonumber \\
& \qquad \quad + \gamma_i \log \frac{\gamma_i}{\overline{w}_i} + (1-\gamma_i) \log \frac{1-\gamma_i}{1-\overline{w}_i} + C\nonumber\\
& =  \gamma_i \bigg\{ \mu_i \sum_{k\neq i} (X^TX)_{ki} \gamma_k\mu_k+\tfrac{1}{2} (X^TX)_{ii}(\sigma_i^2+\mu_i^2) -(Y^TX)_i\mu_i +\log \frac{\sqrt{2}}{\sqrt{\pi}\sigma_i\lambda} -\frac{1}{2} \nonumber\\
& \qquad \quad + \lambda \sigma_i\sqrt{2/\pi}e^{-\mu_i^2/(2\sigma_i^2)} +\lambda\mu_i(1-2\Phi(-\mu_i/\sigma_i)) + \log \frac{\gamma_i}{1-\gamma_i} + \log \frac{b_0}{a_0} \bigg\} + \log(1-\gamma_i) + C\nonumber\\
&=:h_i(\gamma_i|\boldsymbol{\mu},\boldsymbol{\sigma},\boldsymbol{\gamma}_{-i})\label{eg: def:h_gamma_i}
\end{align}
where $C>0$ may change from line to line and is independent of $\gamma_i$. Setting the derivative with respect to $\gamma_i$ of this last expression equal to zero and rearranging gives \eqref{eq: VB:step:lambda}.

\subsection{Algorithms for Gaussian slabs}\label{sec:gauss_VB_algos}

We collect here for completeness the variational algorithms for the spike-and-slab prior with Gaussian slabs with which we have compared our method. First we give the component-wise update of the parameters as in \cite{logsdon2010variational}, see Algorithm \ref{alg: VB_Gauss} below.

\begin{algorithm}
\caption{Component-wise variational Bayes for Gaussian prior slabs}\label{alg: VB_Gauss}
\begin{algorithmic}[1]
\BState \textbf{Initialize}: $(\Delta_H,\boldsymbol{\sigma},\boldsymbol{\gamma})$, $\boldsymbol{\mu}:=\hat{\mu}^{(0)}$ (for a preliminary estimator $\hat{\mu}^{(0)}$), $\boldsymbol{a}:=order( |\boldsymbol{\mu}|)$\While{$ \Delta_{H}\geq \eps$}
\State {$\boldsymbol{\gamma}_{old}:=\boldsymbol{\gamma}$}
\For {$ i=1$ to $p$}
\State{$\sigma_i:= 1/\sqrt{(X^TX)_{ii}+1} $}
\State{$\mu_i:=\sigma_i^2\big((Y^TX)_{i}- \sum_{j\neq i}(X^TX)_{j,i}\gamma_j \mu_j \big)$}
\State{$\gamma_i=\text{logit}^{-1}\big(\log(a_0/b_0)+\log \sigma_i+\mu_i^2/(2\sigma_i^2) \big)$}
\EndFor
\State{ $\Delta_{H}:=\max_i\{ | H(\gamma_i)-H(\gamma_{old,i}) | \}$}
\EndWhile
\end{algorithmic}
\end{algorithm}

In \cite{huang:2016} the authors argue that coordinate-wise parameter updates can accumulate error from each step leading to a suboptimal optimization procedure. To resolve this, they propose simultaneously updating the entire parameter vectors $\boldsymbol{\mu},\boldsymbol{\sigma}$ and $\boldsymbol{\lambda}$ without using a CAVI type of algorithm. A version of their proposed algorithm is given in Algorithm \ref{alg: VB_Gauss:batch}, where $diag(v)$, $v\in \mathbb{R}^p$, creates a diagonal square matrix in $\mathbb{R}^{p\times p}$ with diagonal elements $v$ (see also Algorithm 1 of \cite{yang:2017} with $\alpha=1$, $\sigma=1$ and $\nu_1=1$ for a related implementation). As in the other cases, we have taken the ridge regression estimator $(X^TX+I)^{-1} X^TY$ as our initialization for $\mu$.

\begin{algorithm}
\caption{Batch-wise variational Bayes for Gaussian prior slabs}\label{alg: VB_Gauss:batch}
\begin{algorithmic}[1]
\BState \textbf{Initialize}: $(\Delta_H,\boldsymbol{\sigma},\boldsymbol{\gamma})$, $\boldsymbol{\mu}:=\hat{\mu}^{(0)}$ (for a preliminary estimator $\hat{\mu}^{(0)}$), $\boldsymbol{a}:=order( |\boldsymbol{\mu}|)$
\While{$ \Delta_{H}\geq \eps$}
\State {$\boldsymbol{\gamma}_{old}:=\boldsymbol{\gamma}$}
\State {$\Gamma:=diag(\gamma)$}
\State{$\mu:= (X^TX+ \Gamma )^{-1}X^TY$}
\For{$ i=1$ to $p$}
\State{$\sigma_i:= 1/\sqrt{(X^TX)_{ii}+\gamma_i}$}
\State{$\gamma_i:=\text{logit}^{-1}\big(\text{logit}(1/p)+\log \sigma_i+\mu_i^2/(2\sigma_i^2) \big)$}
\EndFor
\State{ $\Delta_{H}:=\max_i\{ | H(\gamma_i)-H(\gamma_{old,i}) | \}$}
\EndWhile
\end{algorithmic}
\end{algorithm}

Lastly, we provide the VB algorithm for the $\mathcal{Q}_{MF}$ mean-field variational class using Laplace slabs in the prior.

\begin{algorithm}
\caption{Variational Bayes for Laplace prior slabs and variational class $\mathcal{Q}_{MF}$}\label{alg:VB:Laplace2}
\begin{algorithmic}[1]
\BState \textbf{Initialize}: $(\Delta_H,\boldsymbol{\sigma},\boldsymbol{\gamma})$, $\boldsymbol{\mu}:=\hat{\mu}^{(0)}$ (for a preliminary estimator $\hat{\mu}^{(0)}$), $\boldsymbol{a}:=order( |\boldsymbol{\mu}|)$
\While{$ \Delta_{H}\geq \eps$}
\State {$\boldsymbol{\gamma}_{old}:=\boldsymbol{\gamma}$}
\For {$ j=1$ to $p$}
\State{$i:= a_j$}
\State{$\mu_i:=\text{argmax}_{\mu_i}f_i(\mu_i|\boldsymbol\mu_{-i},\boldsymbol\sigma,\boldsymbol\gamma,z_i=1)$ \qquad \qquad \qquad \  // see equation \eqref{eq: joint:KL}}
\State{$\sigma_i:=\text{argmax}_{\sigma_i}g_i(\sigma_i|,\boldsymbol\mu,\boldsymbol\sigma_{-i},\boldsymbol\gamma,z_i=1)$  \qquad \qquad \qquad  // see equation \eqref{eq: joint:KL}}
\State{$\gamma_i:=\text{argmax}_{\gamma_i\in\{0,1\}}h_i(\gamma_i|\boldsymbol{\mu},\boldsymbol{\sigma},\boldsymbol{\gamma}_{-i}) $\qquad\qquad\qquad\qquad // see equation \eqref {eg: def:h_gamma_i}}
\EndFor
\State{ $\Delta_{H}:=\max_i\{ | H(\gamma_i)-H(\gamma_{old,i}) | \}$}
\EndWhile
\end{algorithmic}
\end{algorithm}

\section{Examples of compatible design matrices}\label{sec:design_matrix_supp}

In addition to the compatibility type constants defined in Section \ref{sec:design_matrix}, we also consider a stronger invertibility condition involving the `mutual coherence' of the design matrix, which is the maximal correlation between the different predictors in $X$.
\begin{definition}[Mutual coherence]\label{def:mc2}
The mutual coherence number is
\begin{align}\label{def:mc}
\emph{mc}(X) = \max_{1\leq i\neq j \leq p} \frac{|\langle X_{\cdot i},X_{\cdot j}\rangle|}{\|X_{\cdot i}\|_2\|X_{\cdot j}\|_2}.
\end{align}
\end{definition}
While we do not actually use the mutual coherence in our results, it provides an easy way to understand the compatibility constants in Definitions \ref{def:compat}-\ref{def:ssssv} in several well-studied design matrix examples below. The following result relates these notions.

\begin{lemma}[Lemma 1 of \cite{castillo:2015}]\label{lem:compat}
$\phi(S)^2 \geq \overline{\phi}(1)^2 - 15|S| \mc(X)$, $\overline{\phi}(s)^2 \geq \widetilde{\phi}(s)^2 \geq \overline{\phi}(1)^2 - s\mc(X)$.
\end{lemma}

By evaluating the infimum in Definition \ref{def:unif_compat} at the unit vectors, one obtains $\widetilde{\phi}(1) = \overline{\phi}(1) = \min_i \|X_{\cdot i}\|_2/\|X\| = \min_{i\neq j} \|X_{\cdot i}\|_2/\|X_{\cdot j}\|_2$, which is bounded away from zero if the columns of $X$ have comparable Euclidean norms. In this case, Lemma \ref{lem:compat} implies that the compatibility numbers and sparse singular values are bounded away from zero for models of size $O(1/\mc(X))$. The mutual coherence condition is thus the strongest of these notions. These conditions are illustrated via the following well-studied examples. 

\begin{enumerate}
\item\label{ex:seq_model} (Sequence model). We observe a vector $Y = (Y_1,\dots,Y_n)$ of independent random variables with $Y_i \sim N(\theta_i,1)$. This corresponds to model \eqref{eq:model} with $n=p$ and $X=I_p$ the identity matrix, so that $\|X\| = \|X_{\cdot i}\|_2 = 1$ for all $i$, the compatibility numbers are 1 and $\mc(X) = 0$. In this setting, all results below are valid for all sparsity levels.

\item (Sequence model, multiple observations). We observe $n$ independent $N(\theta_i,\sigma_n^2)$ random variables with $\sigma_n \to 0$. Defining $Y_i$ as $\sigma_n^{-1}$ times the original observations, this falls within the framework of model \eqref{eq:model} with $X = \sigma_n^{-1}I_p$, so that $\|X\| = \|X_{\cdot i}\|_2 = \sigma_n^{-1}$ for all $i$, the compatibility numbers are 1 and $\mc(X) = 0$, similar to Example \ref{ex:seq_model}.

\item (Regression with orthogonal design). If $X$ is an orthogonal design matrix such that $\langle X_{\cdot i}, X_{\cdot j} \rangle = 0$ for $i\neq j$, the regression problem can be transformed into a sequence model.

\item (Response model).  Suppose the entries of the original regression matrix are i.i.d. random variables $W_{ij}$. We may then normalize the entries of the design matrix by defining $X_{ij} = W_{ij}/\|W_{\cdot j}\|_2$, so that the column lengths satisfy $\|X\| = \|X_{\cdot i}\|_2 = 1$ for all $i$. If $|W_{ij}| \leq C$ for a constant $C>0$ and $\log p = o(n)$, or $Ee^{t_0|W_{ij}|^\alpha}<\infty$ for some $\alpha,t_0>0$ and $\log p = o(n^{\alpha/(4+\alpha)})$, then Theorems 1 and 2 of \cite{cai:2011} show that $\sqrt{n/\log p} \mc(W) \stackrel{P}{\to} 2$ as $n\to \infty$. Since $\mc(W) = \mc(X)$, this shows that for any $\eps>0$, $P(\mc(X) > (2+\eps) \sqrt{(\log p)/n} ) \to 0$. Thus with probability approaching one, the compatibility numbers are bounded away from zero for sparsity levels $s_n = o(\sqrt{n/\log p})$.

A classic example is $W_{ij} \stackrel{iid}{\sim} N(0,1)$. In this case, the above bound on the mutual coherence holds as long as $\log p = o(n^{1/3})$.

\item By rescaling the columns of $X$, one can set the $p\times p$ matrix $C:= X^TX/n$ to take value one for all diagonal entries. Then $\|X\| = \|X_{\cdot i}\|_2 = \sqrt{n}$ for all $i$ and the elements $C_{ij}$, $i\neq j$, are the correlations between columns. For some $m\in \mathbb{N}$, if $C_{ij}= r$ for a constant $0<r<(1+cm)^{-1}$ and all $i\neq j$ or $|C_{ij}| \leq c/(2m-1)$ for every $i\neq j$, then \cite{zhao:2006} show that models up to dimension $m$ satisfy the `strong irrepresentability condition' and are hence estimable. In particular, $\mc(X) = \max_{i\neq j} C_{ij} =O(1/m)$ and hence the compatibility numbers are bounded away from zero for sparsity levels $s_n = o(m)$.
\end{enumerate}

\bibliography{references}{}
\bibliographystyle{acm}

\end{document}